\documentclass[11pt, a4paper]{article}
\usepackage{a4wide}
\usepackage{soul} 

\usepackage{amsmath,amssymb,amsthm,bm}
\usepackage{hyperref, color}
\usepackage{float}
\usepackage[font=small]{caption}
\usepackage[font=footnotesize]{subcaption}
\usepackage[inline]{enumitem}
 
\usepackage[title]{appendix} 
 
\usepackage{tikz}
\usetikzlibrary{intersections,calc, angles, quotes, arrows.meta,through,backgrounds, positioning}
\usepackage{tkz-euclide}
\usepackage{pgfplots}
\usepgfplotslibrary{fillbetween}
\pgfplotsset{compat=1.17}

\usepackage{xcolor}
\definecolor{darkgreen}{rgb}{0.1953125, 0.34375, 0.30859375}
\definecolor{lightgreen}{rgb}{0.4921875, 0.6171875, 0.5078125}
\definecolor{darkbrown}{rgb}{0.80078125, 0.6796875, 0.55078125}
\definecolor{lightbrown}{rgb}{0.94921875, 0.79296875, 0.57421875}
\definecolor{hardcolour}{rgb}{0.937, 0.145, 0.373}

\usepackage{hyperref}
\hypersetup{
    colorlinks   = true,                    
    linkcolor    = blue!70!black,
    anchorcolor  = gray,
    citecolor    = blue!70!black,
    filecolor    = red,
    menucolor    = green,
    runcolor     = red,
    urlcolor     = [rgb]{0,0.2,0.5}
}

\newtheorem{theorem}{Theorem}
\newtheorem{lemma}[theorem]{Lemma}

\newtheorem{proposition}[theorem]{Proposition}

\newtheorem{conjecture}[theorem]{Conjecture}

\newtheorem{corollary}[theorem]{Corollary}
\newtheorem{remark}[theorem]{Remark}

\newtheorem{definition}[theorem]{Definition}

\newtheorem*{thmzeros}{Theorem~\ref{thm:ising:zero-free}}
\newtheorem*{thmhardness}{Theorem~\ref{thm:hardness}}
\newtheorem*{corfptas}{Corollary~\ref{cor:ising:fptas}}
\newtheorem*{lemimplementations}{Lemma~\ref{lem:ising:implementations}}
\newtheorem*{lemimplementingminusone}{Lemma~\ref{lem:implementing-1}}
\newtheorem*{conjecturezeros}{Conjecture~\ref{con:zeros-hardness}}

\newtheorem*{lemefficientcovering}{Lemma~\ref{lem:efficient-covering}}
\newtheorem*{lemmobiusimplementations}{Lemma~\ref{lem:implementations}}

\newcommand{\size}[1]{\mathrm{size}(#1)}

\def\numP{\mathrm{\#P}}

\def\ising{Z_{\mathrm{Ising}}}
\def\potts{Z_{\mathrm{Potts}}}
\def\hom{\mathrm{hom}}

\def\algebraic{\mathbb{A}}
\def\algcomplex{\mathbb{C}_{\algebraic}}

\newcommand{\I}{\textsc{Ising}}
\newcommand{\IN}{\textsc{IsingNorm}}
\newcommand{\IA}{\textsc{IsingArg}}
\newcommand{\RT}{\textsc{RatioTutte}}

\makeatletter
\def\prob#1#2#3{\goodbreak\begin{list}{}{\labelwidth\z@ \itemindent-\leftmargin
                        \itemsep\z@  \topsep6\p@\@plus6\p@
                        \let\makelabel\descriptionlabel}
                \item[\textbf{Name:}]#1   \vspace{-1ex}
               \item[\textbf{Instance:}]                #2   \vspace{-1ex}
                \item[\textbf{Output:}]#3  
                \end{list}}
\makeatother

 \let\epsilon=\varepsilon

\title{The complexity of approximating the complex-valued Ising model on bounded degree graphs}
\author{
Andreas Galanis \thanks{
  Department of Computer Science, University of Oxford, Wolfson Building, Parks Road, Oxford, OX1~3QD, UK.}
  \and
  Leslie Ann Goldberg \footnotemark[1]
\and
Andr\'es Herrera-Poyatos \footnotemark[1] \thanks{This author is supported by an Oxford-DeepMind Graduate Scholarship and a EPSRC Doctoral Training Partnership.}
 }

\date{8th April 2022}

\begin{document}

\maketitle

\begin{abstract} 
We study the complexity of approximating the partition function $\ising(G; \beta)$ of the Ising model in terms of the relation between the edge interaction $\beta$ and a parameter $\Delta$ which is an upper bound on the maximum degree of the input graph $G$. Following recent trends in both statistical physics and algorithmic research, we allow the edge interaction $\beta$ to be any complex number. Many recent partition function results focus on complex parameters, both because of physical relevance and because of the key role of the complex case in delineating the tractability/intractability phase transition of the approximation problem.

In this work we establish both new tractability results and new intractability results. Our tractability results show that $\ising(-; \beta)$ has an FPTAS when $\lvert \beta - 1 \rvert / \lvert \beta + 1 \rvert < \tan(\pi / (4 \Delta - 4))$. The core of the proof is showing that there are no inputs~$G$ that make the partition function $0$ when $\beta$ is in this range. Our result significantly extends the known zero-free region of the Ising model (and hence the known approximation results).

Our intractability results show that it is $\numP$-hard to multiplicatively approximate the norm and to additively approximate the argument of $\ising(-; \beta)$ when $\beta \in \mathbb{C}$ is an algebraic number such that $\beta \not \in \mathbb{R} \cup \{i, -i\}$ and $\lvert \beta - 1\rvert / \lvert \beta + 1 \rvert > 1 / \sqrt{\Delta - 1}$.  These are the first results to show intractability of approximating $\ising(-, \beta)$ on bounded degree graphs with complex $\beta$. Moreover, we demonstrate situations in which zeros of the partition function imply hardness of approximation in the Ising model.
\end{abstract}
\newpage


\section{Introduction}

The Ising model is a classical model of ferromagnetism in statistical physics  \cite{Ising1925} and is one of the most well studied particle spin models \cite[Chapter 4]{Welsh1993}.  
The model has a parameter~$\beta$ called the \emph{edge interaction}  which captures the interactions between   ``spins'' along the edges of 
a host graph.
Let $G = (V, E)$ be an undirected graph, possibly with multiple edges or loops. The configurations of the Ising model on~$G$ are the assignments $\sigma \colon V \to \{0,1\}$. Each configuration $\sigma$ has a weight $\beta^{m(\sigma)}$, where $m(\sigma)$ denotes the number of monochromatic edges of $G$ under the assignment~$\sigma$. The \emph{partition function of the Ising model} is the aggregate weight over all configurations, i.e.,
\begin{equation*}
  \ising \left(G; \beta \right) = \sum_{\sigma \colon V \to \{0,1\}} \beta^{m(\sigma)}.
\end{equation*}
The partition function is important because 
of
its relevance in capturing qualitative properties of the model such as free energy and magnetisation.
It is also the normalising factor in the
\emph{Gibbs distribution} which is used to sample configurations.

This work studies the complexity of approximating the partition function, especially in terms of the interaction between the edge interaction~$\beta$ and a parameter~$\Delta\geq 3$
which is an upper bound on the maximum degree of the input graph~$G$. 
It is well-known in statistical physics 
\cite{bena2005statistical,Heilmann1972, lieb1981general,Sokal2005,Welsh1993,yang1952statistical}
that physical phase transitions can be identified by viewing the parameter~$\beta$ 
as a complex number.
More recently, the consideration of complex parameters
has been discovered in  computer science
to be the key to characterising  the tractability/intractability boundary.
On the algorithmic side, important work of Barvinok and of Patel and Regts \cite{BarvinokBook, Patel2017} shows that, for some partition functions, 
finding an open zero-free region (in terms of a parameter in the complex plane)
leads directly to an  FPTAS for approximating the partition function. 
This approach has been exploited for several partition functions to obtain new approximability results for complex parameters, sometimes including real parameters that other methods had failed to address~\cite{
barvinokregts2019, GLLZ, GLLb, harrow2019classical,
LSS, Liu2019Fisher,  Liu2019Approximation, Bencs2018,  peters2018location,PetersRegts2019}.
The new algorithms discovered in this work are
also based on the identification of zero-free regions. 
On the complexity-theoretic side, 
recent publications have shown for various spin systems a connection between the location of the complex zeros of the partition function and provable inapproximability of the partition function~\cite{Bezakova2019, Bezb, DeBoer2021, BGPR}. 
A key approach in all of these applications, and one that we will also develop in this work, is the
connection to complex dynamics.

Before describing our results, we briefly describe existing work on the
problem of approximating the partition function of the Ising model.
Let $\Delta \ge 3$ be an integer. We are interested in the problem of approximating $\ising(G; \beta)$ when the input graph~$G$ has maximum degree at most $\Delta$. This problem has already been well studied in the case where $\beta$ is a positive real. When $\beta > (\Delta-2)/\Delta$ there is an FPRAS for $\ising(-; \beta)$ on graphs with maximum degree at most $\Delta$ \cite{Jerrum1993, Sinclair2014}. When $0 < \beta < (\Delta-2)/\Delta$, there is no FPRAS for $\ising(-; \beta)$ on graphs with maximum degree at most $\Delta$ unless $\mathsf{NP} = \mathsf{RP}$ \cite{Galanis2016}. The computational phase transition $(\Delta-2)/\Delta$ coincides precisely with the uniqueness / non-uniqueness phase transition of the $\Delta$-regular infinite tree. 

The complexity of approximation is mostly not understood when $\beta$ is complex. 
Prior to this work, there was no inapproximability result for any non-real edge interactions. 
Moreover, even though zero-free regions of the partition function have been the focus of recent publications~\cite{Barvinok2020, Liu2019Fisher, Mann2019}, these   regions 
turn out to be far from optimal. In this work we shed some light on this approximability problem by significantly extending the known zero-free regions 
(leading to approximation algorithms)
and by giving an inapproximability result that covers most of the complex plane. Our zero-free region for the Ising model is stated in Theorem~\ref{thm:ising:zero-free}.

\newcommand{\statethmzeros}{Let $\Delta$ be an integer with $\Delta\ge 3$. Let $G = (V, E)$ be a graph of maximum degree at most $\Delta$. Let $\varepsilon_{\Delta} = \tan(\pi/(4 (\Delta-1))) \in (0,1)$. Then $\ising(G; \beta) \ne 0$ for all $\beta \in \mathbb{C}$ with  $\lvert \beta - 1\rvert / \lvert \beta + 1 \rvert \le \varepsilon_\Delta$.}
\begin{theorem} \label{thm:ising:zero-free}
	\statethmzeros \footnote{This result can be extended  to the Potts model, demonstrating that, for $\varepsilon_{\Delta}' = \tan(\pi/(4 \Delta))$, and for any $q\geq 2$, $\potts(G; q, \beta) \ne 0$ for all  $\beta \in \mathbb{C}$ with  $\lvert \beta - 1\rvert / \lvert \beta + 1 \rvert \le \varepsilon_\Delta'$.}
\end{theorem}

Theorem~\ref{thm:ising:zero-free} can be applied in conjunction with the algorithms of Barvinok, and Patel and Regts \cite{BarvinokBook, Patel2017} to obtain an FPTAS for $\ising(-; \beta)$,  giving the following corollary\footnote{In this work we state our algorithmic results in terms of algebraic numbers as the required computations can be performed efficiently in the field of algebraic numbers, see Section~\ref{sec:pre:algebraic}.} (proved in Section~\ref{sec:easiness}).

\newcommand{\statecorfptas}{Let $\Delta$ be an integer with $\Delta \ge 3$. Let $\beta$ be an algebraic number such that  {$\lvert \beta - 1\rvert / \lvert \beta + 1 \rvert < \varepsilon_\Delta$}, where  $\varepsilon_{\Delta} = \tan(\pi/(4 (\Delta-1)))$. Then there is  an algorithm that, on inputs a graph $G$ with maximum degree at most $\Delta$  and a rational $\epsilon > 0$, runs in time $\mathrm{poly}(\mathrm{size}(G), 1/\epsilon)$ and outputs $\hat{Z} = \ising(G;\beta) e^z$ for some complex number $z$ with $\lvert z \rvert \le \epsilon$.}
\begin{corollary} \label{cor:ising:fptas}
	\statecorfptas
\end{corollary}

 \pgfkeys{
    /pgf/declare function={br1(\t) = e^(2*(1-\t) /3) * cos (deg(2 * \t^2/30) );},
    /pgf/declare function={br2(\t) = e^(-2*(1-\t) /3) * cos (deg(2 * \t^2/30));},
    /pgf/declare function={bi1(\t) =  e^(2*(1-\t) / 3) * sin (deg(2 * \t^2 / 30));},
    /pgf/declare function={bi2(\t) =  e^(-2*(1-\t) / 3) * sin (deg(2 * \t^2/30));}
}

\pgfplotsset{width=\textwidth,height=0.5\textwidth}
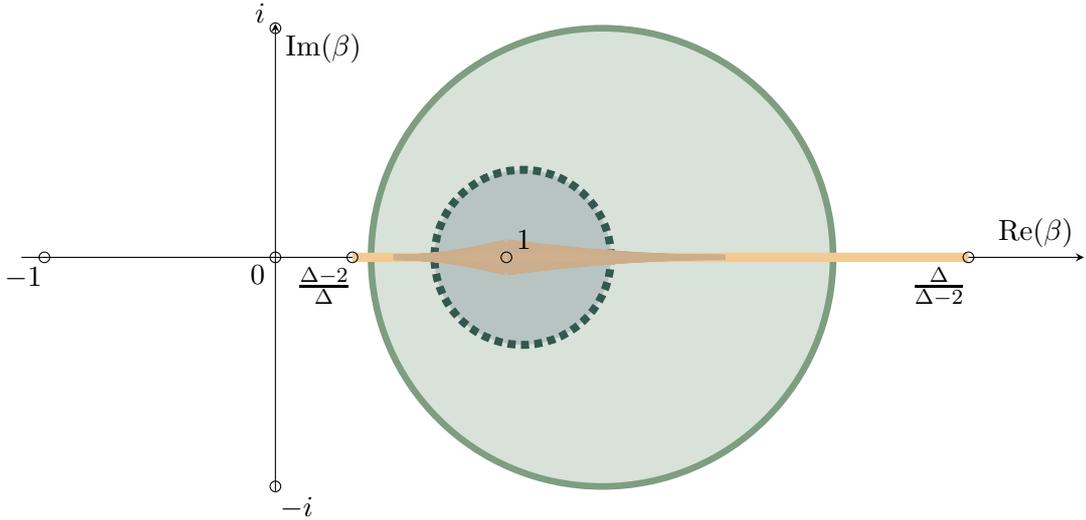
\begin{figure}[H]
\centering
    \begin{tikzpicture}
		\begin{axis}[
			ticks=none,
			xmin=-1.1,
			xmax=3.5,
			ymin=-1,
			ymax=1,
			axis equal,
			axis lines=middle,
			xlabel=Re($\beta$),
			ylabel=Im($\beta$),
			disabledatascaling,
			axis background/.style={fill=white}]
			
			\addplot[draw=black, mark=o, only marks, nodes near coords=$0$,every node near coord/.style={anchor=45}] coordinates{(0,0)};
			\addplot[draw=black, mark=o, only marks, nodes near coords=$1$,every node near coord/.style={anchor=225}] coordinates{(1,0)};
			\addplot[draw=black, mark=o, only marks, nodes near coords=$-1$,every node near coord/.style={anchor=45}] coordinates{(-1,0)};
			\addplot[draw=black, mark=o, only marks, nodes near coords=$\frac{\Delta-2}{\Delta}$,every node near coord/.style={anchor=45}] coordinates{(0.33333333,0)}; %
			\addplot[draw=black, mark=o, only marks, nodes near coords=$\frac{\Delta}{\Delta-2}$,every node near coord/.style={anchor=45}] coordinates{(3,0)};%
			\addplot[draw=black, mark=o, only marks, nodes near coords=$i$,every node near coord/.style={anchor=315}] coordinates{(0,1)};
			\addplot[draw=black, mark=o, only marks, nodes near coords=$-i$,every node near coord/.style={anchor=135}] coordinates{(0,-1)};
			
           \draw [line width=2.5pt,draw= lightgreen, fill=lightgreen!30] (1.41421,0) circle(1);
           \draw [line width=3pt, dotted, draw= darkgreen, fill=darkgreen!35] (1.0704830971,0) circle(0.3820131689);
		   \draw [line width=3.5pt, draw= lightbrown] (axis cs: 0.333333333,0) -- (axis cs: 3,0);
           
           \draw [name path=bru, line width=2pt, draw= darkbrown] plot [domain=0:1, samples=100, const plot] ({br1(\x)},{bi1(\x)});
           \draw [name path=brd, line width=2pt, draw= darkbrown] plot [domain=0:1, samples=100, const plot] ({br1(\x)},{-bi1(\x)});
           \draw [name path=blu, line width=2pt, draw= darkbrown] plot [domain=0:1, samples=100, const plot] ({br2(\x)},{bi2(\x)});
           \draw [name path=bld, line width=2pt, draw= darkbrown] plot [domain=0:1, samples=100, const plot] ({br2(\x)},{-bi2(\x)});
           \tikzfillbetween[of=blu and bld]{darkbrown, opacity=0.9};
           \tikzfillbetween[of=bru and brd]{darkbrown, opacity=0.9};
		\end{axis}
	\end{tikzpicture}
	\caption{Zero-free regions for the partition function of the Ising models on graphs with maximum degree $\Delta = 3$.  The following four regions have been plotted:\\\hspace{\textwidth} 
	{\color{lightgreen} $\bullet$} The large disk corresponds to the region given in Theorem~\ref{thm:ising:zero-free}.\\\hspace{\textwidth}
	{\color{darkgreen} $\bullet$} The small dotted disk corresponds to $\lvert \beta - 1 \rvert / \lvert \beta + 1 \rvert \le \delta_\Delta$, where $\delta_\Delta$ is as in \eqref{eq:delta}, and it contains the regions stated in Corollaries~\ref{cor:region-barvinok} and~\ref{cor:mann-bremner} due to Barvinok, Mann and Bremner \cite{BarvinokBook, Mann2019}. \\\hspace{\textwidth} 
	{\color{darkbrown} $\bullet$} The diamond-shaped region corresponds to the zero-free region given in \cite{Barvinok2020} by Barvinok and Barvinok (see Theorem~\ref{thm:barvinok:ising} for the statement).  \\\hspace{\textwidth} {\color{lightbrown} $\bullet$} The segment joining $(\Delta-2) / \Delta$ and $\Delta / (\Delta - 2)$ corresponds to the region given in \cite{Liu2019Fisher} by  Liu, Sinclair and Srivastava (see Theorem~\ref{thm:liu} for the statement).}
	\label{fig:zero-free}
\end{figure}

Theorem~\ref{thm:ising:zero-free} significantly extends the zero-free regions given in  \cite{Barvinok2020,BarvinokBook, Liu2019Fisher, Mann2019}. The case $\Delta = 3$ is depicted in Figure~\ref{fig:zero-free}. In fact, the zero-free regions of Barvinok, and Mann and Bremner \cite{BarvinokBook, Mann2019} are contained in our result for any $\Delta \ge 3$, see Section~\ref{sec:comparing} for a detailed description of these zero-free regions. In \cite{Mann2019} the authors also discuss how an FPRAS for the partition function of the Ising model on bounded-degree graphs can be used to strongly simulate certain classes of IQP circuits. We note that their quantum simulation results are also extended as a consequence of Theorem~\ref{thm:ising:zero-free}.

When it comes to hardness results on complex edge interactions, we are not aware of any hardness result in the literature that covers non-real edge interactions. Our hardness result is given in Theorem~\ref{thm:hardness}. First, let us introduce some notation. We consider the problem of multiplicatively approximating the norm of $\ising(G; \beta)$ and
the problem of additively approximating the principal argument of $\ising(G; \beta)$ for a fixed algebraic number $\beta$\footnote{For $z \in \mathbb{C}\backslash\{0\}$, we denote by $|z|$ the norm of $z$, by $\mathrm{Arg}(z)\in [0,2\pi)$ the principal argument of $z$ and  by $\arg(z)$ the set $\{\mathrm{Arg}(z) + 2 \pi j : j \in \mathbb{Z}\}$ of all the arguments of $z$, so that for any $a\in \arg(z)$ we have $z = |z|\exp(i a)$.}.  These computational problems can be formally stated as follows.  Let $K > 1$ and $\rho \in (0, \pi/2)$ be real numbers.

\prob{$\IN(\beta, \Delta, K)$}{A (multi)graph $G$ with maximum degree at most $\Delta$.}{ If $\ising(G; \beta)= 0$, then the algorithm may output any rational number. Otherwise, it must output a rational number $\hat{N}$ such that $\hat{N}/ K \le \left|\ising(G; \beta)\right| \le K \hat{N} $.}

\prob{$\IA(\beta, \Delta, \rho)$}{A (multi)graph $G$ with maximum degree $\Delta$.}{ If $\ising(G; \beta)= 0$, then the algorithm may output any rational number. Otherwise, it must output a rational number $\hat{A}$ such that for some $a \in \arg(\ising(G; \beta))$  we have 
$ | a - \hat{A} | \le \rho $.}

It is important to note that each choice of the parameters $\beta, \Delta, K, \rho$ gives a different computational problem. It turns out that, by a standard powering argument of the partition function, the choice of $K$ and $\rho$ does not change the hardness of the problem (as long as $K > 1$ and $\rho \in (0, \pi/2)$), see \cite[Lemma 3.2]{Goldberg2017}.

\newcommand{\statethmhardness}{Let $\Delta$ be an integer with $\Delta \ge 3$ and let $\beta \in \mathbb{C}$ be an algebraic number such that $\beta \not \in \mathbb{R} \cup \{i, -i\}$ and $\lvert \beta - 1\rvert / \lvert \beta + 1 \rvert > 1 / \sqrt{\Delta - 1}$. Then the problems  $\IN(\beta, \Delta, 1.01)$ and $\IA(\beta, \Delta, \pi/3)$ are $\numP$-hard. }
\begin{theorem} \label{thm:hardness}
  \statethmhardness
\end{theorem}

Corollary \ref{cor:ising:fptas} and Theorem \ref{thm:hardness} leave the complexity of the problems $\IN(\beta, \Delta, 1.01)$ and $\IA(\beta, \Delta, \pi/3)$ unaddressed for those edge interactions $\beta \not \in \mathbb{R}$ such that 
\begin{equation} \label{eq:unaddressed-region}
    \tan \big(\frac{\pi}{4 (\Delta - 1)} \big) \le \left| \frac{\beta - 1}{\beta + 1 } \right|  \le \frac{1}{\sqrt{\Delta - 1}}.
\end{equation}
 It turns out that the partition function has zeros inside the region given by \eqref{eq:unaddressed-region} (see Corollary~\ref{cor:hardness-inside}). Moreover, we show that if there is a ``nice'' graph $G$ such that $\ising(G; \beta) = 0$, then  $\IN(\beta, \Delta, 1.01)$ and $\IA(\beta, \Delta, \pi/3)$ are $\numP$-hard, see Lemma~\ref{lem:zeros-hard} and Corollary~\ref{cor:zeros-hard}. This allows us to find points $\beta$ as in \eqref{eq:unaddressed-region} such that the approximation problems are $\numP$-hard, as depicted in Figure~\ref{fig:results}.  
 
\pgfplotsset{width=\textwidth,height=0.5\textwidth}
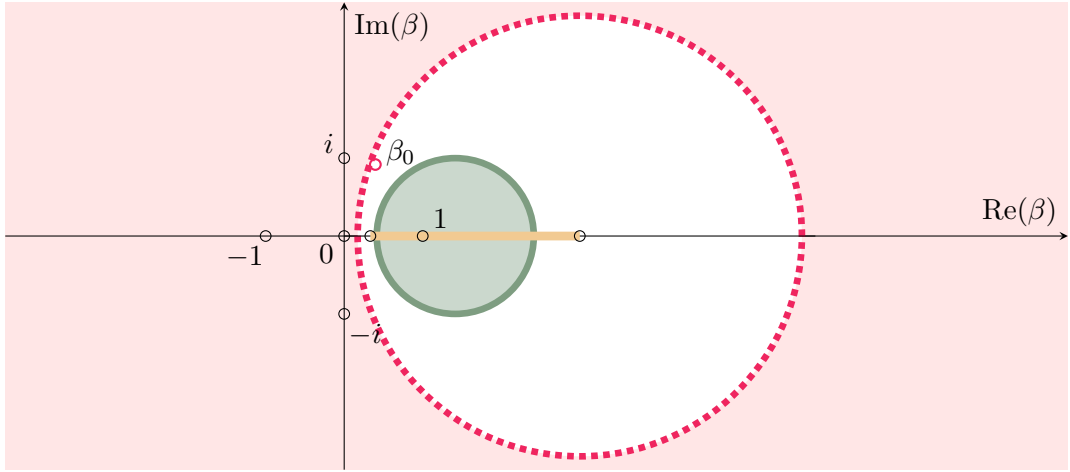
\begin{figure}[H]
\centering
    \begin{tikzpicture}
		\begin{axis}[
			ticks=none,
			xmin=-1.1,
			xmax=6,
			ymin=-3,
			ymax=3,
			axis equal,
			axis lines=middle,
			xlabel=Re($\beta$),
			ylabel=Im($\beta$),
			disabledatascaling,
			axis background/.style={fill=red!10}]
			
			\addplot[draw=black, mark=o, only marks, nodes near coords=$0$,every node near coord/.style={anchor=45}] coordinates{(0,0)};
			\addplot[draw=black, mark=o, only marks, nodes near coords=$1$,every node near coord/.style={anchor=225}] coordinates{(1,0)};
			\addplot[draw=black, mark=o, only marks, nodes near coords=$-1$,every node near coord/.style={anchor=45}] coordinates{(-1,0)};
			\addplot[draw=black, mark=o, only marks, nodes near coords=$i$,every node near coord/.style={anchor=315}] coordinates{(0,1)};
			\addplot[draw=black, mark=o, only marks, nodes near coords=$-i$,every node near coord/.style={anchor=135}] coordinates{(0,-1)};
			\addplot[thick, draw=hardcolour, mark=o, only marks, nodes near coords=$\beta_0$,every node near coord/.style={anchor=205}] coordinates{(0.396608,0.917988)};
			\addplot[draw=black, mark=o, only marks, nodes near coords=,every node near coord/.style={anchor=45}] coordinates{(0.33333333,0)}; %
			\addplot[draw=black, mark=o, only marks, nodes near coords=,every node near coord/.style={anchor=45}] coordinates{(3,0)};%
			
		   \draw [line width=2.5pt, draw= hardcolour, dotted, fill=white] (3,0) circle(2.82843);
           \draw [line width=2.5pt, draw= lightgreen, fill=lightgreen!40] (1.41421,0) circle(1);
           \draw [draw=black] (axis cs: 0,0) -- (axis cs: 6,0);
		   \draw [line width=3.3pt, draw= lightbrown] (axis cs: 0.333333333,0) -- (axis cs: 3,0);

		\end{axis}
	\end{tikzpicture}
	\caption{The complexity of approximating the  partition function of the Ising model on graphs with maximum degree $\Delta = 3$ and $\beta \in \mathbb{C} \setminus \mathbb{R}$.  \\\hspace{\textwidth} {\color{hardcolour} $\bullet$} Theorem~\ref{thm:hardness}: when $\lvert \beta - 1 \rvert / \lvert \beta + 1 \rvert > 1 / \sqrt{\Delta -1}$ and $\beta \not \in \{i, -i\}$,  $\IN(\beta, \Delta, 1.01)$ and $\IA(\beta, \Delta, \pi/3)$ are $\numP$-hard (region outside the large dotted red circle). \\\hspace{\textwidth} 
	{\color{hardcolour} $\bullet$} Corollary~\ref{cor:hardness-inside}: there are points $\beta_0 \in \mathbb{C} \setminus \mathbb{R}$ with $\lvert \beta_0 - 1 \rvert / \lvert \beta_0 + 1 \rvert < 1 / \sqrt{\Delta -1}$ such that $\ising(G; \beta_0) = 0$ for some graph $G$ with maximum degree $\Delta$. The problems  $\IN(\beta_0, \Delta, 1.01)$ and $\IA(\beta_0, \Delta, \pi/3)$ are $\numP$-hard.
	\\\hspace{\textwidth} {\color{lightgreen} $\bullet$} Theorem~\ref{thm:ising:zero-free}: there is an FPTAS for $\ising(-; \beta)$ when $\lvert \beta - 1 \rvert / \lvert \beta + 1 \rvert < \tan(\pi/(4 \Delta - 4))$ (region inside the small green circle). \\\hspace{\textwidth} 
	{\color{lightbrown} $\bullet$} Theorem~\ref{thm:liu} by Liu, Sinclair and Srivastava \cite{Liu2019Fisher}: the interval $((\Delta-2) / \Delta, \Delta / (\Delta - 2))$ is contained in an open zero-free region (thick segment on the real line), so there is an FPTAS for $\ising(-; \beta)$. \\\hspace{\textwidth} 
	{$\circ$} The points $0, 1, -1, i$ and  $-i$ are easy points of the Ising model: the partition function can be evaluated at these points in polynomial time in the size of the input graph \cite{Jaeger1990}. }
	\label{fig:results}
\end{figure}

\subsection{Proof approach}

In the proof of Theorem~\ref{thm:ising:zero-free} we use the SAW tree construction of Godsil and Weitz \cite{Godsil1981, Weitz2006} to reduce the study of zero-free regions of partition functions on graphs to the study of zero-free regions of partition functions on trees (see Section~\ref{sec:easiness:tsaw} for details).  The partition function of a two-spin system on a tree admits a recurrence expression that can be studied to find zero-free regions on trees. This approach has been successfully applied in the literature for the Ising model and other partition functions \cite{Liu2019Fisher, Bencs2018, Bezakova2019}. In our work we exploit the properties of the Mobius function $h_\beta(z) = (\beta z + 1) / (\beta + z)$ appearing in this recurrence for the Ising model. This Mobius function satisfies the equality 
\begin{equation} \label{eq:hbeta:property}
	\frac{h_{\beta}(z) -1 }{h_{\beta}(z) +1} = \frac{(\beta-1)(z-1)}{(\beta+1)(z+1)},
\end{equation}
which neatly relates properties of  $(\beta -1)/(\beta+1)$ to properties of the partition function of the Ising model on trees, and greatly simplifies the derivation of our zero-free region.  The translation of Theorem~\ref{thm:ising:zero-free} to an FPTAS for the partition function then follows from the work of Barvinok, Patel and Regts \cite{Barvinok2020, Patel2017}, see the proof of Corollary~\ref{cor:ising:fptas}.

In order to obtain our inapproximability results, we construct graphs $H$ with maximum degree at most $\Delta$ and two distinguished vertices $s,t$ with degree $1$ such that 
substituting an edge in the host graph  with $(H, s, t)$  has the effect of altering the edge interaction $\beta$ of the original edge to a new edge interaction $\beta'$. In this case, we say that $H$ $(\beta, \Delta)$-implements $\beta'$, see Section~\ref{sec:pre:implementations} for a formal definition. Implementations have played an important role in  proofs of hardness of evaluating and approximating partition functions, and they are the main tool to reduce exact computation to approximate computation via a binary search \cite{Bezb, Goldberg2014, Galanis2020}. Initiated in \cite{Jaeger1990, Welsh1993}, these constructions have now become more elaborate in recent inapproximability results \cite{Goldberg2014, Galanis2020, Bezb, BGPR}, using connections to the iteration of complex dynamical systems. The following definition captures the relevant framework for our implementations.

\begin{definition}\label{def:impplane}
Let $\Delta \ge 3$ be an integer and $\beta \in \algcomplex$.\footnote{We denote by $\algebraic$ the set of real algebraic numbers and we denote by $\algcomplex$ the set of complex algebraic numbers.} We say that the pair $(\Delta, \beta)$ implements the complex plane (resp. the real line) in polynomial time for the Ising model if there is an algorithm such that, on input $\lambda \in  \algcomplex$ (resp. $\lambda \in \algebraic$) and rational $\epsilon > 0$, computes a graph $G$ that $(\Delta, \beta)$-implements a complex number $\hat{\lambda}$ with $\lvert \lambda - \hat{\lambda} \rvert \le \epsilon$. The running time of this algorithm must be polynomial in the size of the representations of $\lambda$ and $\epsilon$.
\end{definition}

Our main contribution is that we can $(\Delta, \beta)$ implement the real line (and, in fact, the complex plane), for those pairs $(\Delta, \beta)$ given in the following lemma.

\newcommand{\statelemimplementations}{Let $\Delta$ be an integer with $\Delta \ge 3$ and let $\beta \in  \algcomplex \setminus \mathbb{R}$ with $\beta \not \in \{i, -i \}$ and $1/\sqrt{\Delta -1} < \lvert \beta - 1\rvert / \lvert \beta + 1 \rvert$. Then the pair $(\Delta, \beta)$ implements the complex plane in polynomial time for the Ising model.}
\begin{lemma}  \label{lem:ising:implementations}
	\statelemimplementations
\end{lemma}

The requirement that it be possible to implement the real line is the main bottleneck when reducing exact computation to approximate computation. Even when it is possible to identify some parameter values which enable the implementation of the real line,  the complete determination of  the set of parameter values which make this possible seems out of reach, see, for instance, \cite{Goldberg2014, Bezb}.

The proof of Lemma~\ref{lem:ising:implementations} uses connections with complex dynamics, following recent developments in the area. The main idea in this line of works is to analyse what can be implemented with trees, which can be done via understanding the properties of the underlying dynamical system. A key difference in the case of the Ising model relevant to previous works is that vertex-style implementations are useless; due to the perfect symmetry of the Ising model nothing interesting can be implemented through that route. Instead, we have to consider more elaborate edge gadgets, cf. Section~\ref{sec:hardness:implementations}, and obtain tree-style recursions for them. Surprisingly, we are able to recover the tree-recursions for vertex activities (even though our gadgets are not trees and simulate edge activities instead), albeit with a bit different value of $\beta$ which yields the square root in Lemma~\ref{lem:ising:implementations}. We leave as a tantalising open problem how to remove this square root, which seems inherent in our edge-style approach.

The good news is that once this edge-framework of the gadgets is in place, we can adapt suitably the arguments given in \cite{Bezb}. We have in fact generalised these arguments in Appendix~\ref{sec:A}, so that they are more amenable to be used for other spin systems.   A quick summary of the main idea behind Appendix~\ref{sec:A} is as follows. We assume that we have access to a recursively-constructed gadget that implements a weight $f(z)$ assuming that we can implement $z$ (for us, this is the gadget given in Section~\ref{sec:hardness:implementations}). Then we apply results of complex dynamics to the function $f$ in order to understand which points we can implement by iterating $f$, which involves studying the neighbourhood of fixed points of $f$. There are two steps in the constructions. In the first step, we show how to implement approximations of any number near a fixed point of $f$ (which is done using Lemma~\ref{lem:efficient-covering} in our setting). In the second step, this implementation result is translated to implementing the complex plane in polynomial time when the fixed point under consideration is repelling (Lemma~\ref{lem:implementations}). As explained in Appendix~\ref{sec:A}, it is important in this generalisation that the function $f$ is of the form $g(z^d)$, where $g$ is a Mobius map\footnote{The results of \cite{Bezb} are derived for $g(z) = 1/(1 + \lambda z)$, where $\lambda$ is an activity of the independent set polynomial. As noted in Appendix~\ref{sec:A} some extra work is needed to generalise these results to any Mobius map $g$.}. 

To conclude this section, we comment on the connection between zeros of the partition function and hardness. It turns out that our hardness and implementation results can be applied to conclude hardness of approximation at some zeros of the partition function, such as the zero plotted in Figure~\ref{fig:results}. Our main result on this matter is the following  lemma.

\newcommand{\statelemimplementingminusone}{Let $\Delta$ be an integer with $\Delta \ge 3$. Let $\beta \in \algcomplex \setminus (\mathbb{R} \cup \{i, -i\})$. Let us assume that $(\Delta, \beta)$ implements the edge interaction $-1$. Then $\IN(\Delta, \beta, 1.01)$ and $\IA( \Delta, \beta, \pi/3)$ are $\numP$-hard.
}
\begin{lemma} \label{lem:implementing-1}
\statelemimplementingminusone
\end{lemma}

Typically, if we have a graph $G$ with maximum degree $\Delta$ such that $\ising(G; \beta) = 0$, then  this can be used to $(\Delta, \beta)$-implement $-1$ (provided that we can make the terminals have degree 1) and conclude hardness. See Lemma~\ref{lem:zeros-hard}, where we conclude hardness of approximation based on Lemma~\ref{lem:implementing-1} and appropriate graphs with zero partition function. This is the first result of this style for the Ising model, though building a  connection between zeros and inapproximability for bounded-degree graphs has also been explored thoroughly in a recent work  \cite{ DeBoer2021} for the independence polynomial. These observations lead us to propose the following conjecture.

\newcommand{\stateconjecturezeros}{
Let $\Delta$ be an integer with $\Delta \ge 3$ and let $\beta\in \algcomplex$  with $\beta \not \in \mathbb{R}\cup \{i, -i\}$. If there is a graph $G$ with maximum degree at most $\Delta$ such that $\ising(G; \beta) = 0$, then the problems $\IN(\beta, \Delta, 1.01)$ and $\IA(\beta, \Delta, \pi/3)$ are $\numP$-hard.
}
\begin{conjecture} \label{con:zeros-hardness} 
   \stateconjecturezeros
\end{conjecture}

We make progress toward Conjecture~\ref{con:zeros-hardness} in Corollary~\ref{cor:zeros-hard}, where we have to weaken the result, concluding hardness of  $\IN(\beta, \Delta, 1.01)$ and $\IA(\beta, \Delta, \pi/3)$ when the graph $G$ has  maximum degree at most $\Delta - 1$. Unfortunately, our implementation results seem not enough to prove the full conjecture.

This paper is organised as follows. In Section~\ref{sec:easiness} we prove Theorem~\ref{thm:ising:zero-free} and Corollary~\ref{cor:ising:fptas}. In Section~\ref{sec:hardness} we prove Theorem~\ref{thm:hardness}. In Section~\ref{sec:zeros} we give explicit evidence that zeros imply hardness of approximation and use these results to find more edge interactions where the approximation problem is $\numP$-hard. In Appendix~\ref{sec:A} we generalise the implementation results of \cite{Bezb} so that they can be applied to other two spin systems, including the Ising model. This generalisation may be useful in its own.

\section{Easiness: a zero-free region for the Ising model} \label{sec:easiness}

In this section we prove Theorem~\ref{thm:ising:zero-free} and Corollary~\ref{cor:ising:fptas}. We also compare Theorem~\ref{thm:ising:zero-free} to the zero-free regions appearing in Figure~\ref{fig:zero-free}. First, let us introduce some notation and observations that we use in the rest of this section.  
 
\subsection{The tree of self-avoiding walks} \label{sec:easiness:tsaw}

In this section we recall some results concerning the self-avoiding walk tree  (SAW tree) of a graph and its connection to the partition function of the Ising model. SAW trees were introduced in the study of partition functions by Godsil in \cite{Godsil1981} to study the matching polynomial. SAW trees gained in popularity after the work of Weitz on the independent set polynomial \cite{Weitz2006}. The idea of Godsil and Weitz was reducing the study of the partition function of a two-spin system on graphs to the study of the same partition function on trees. This idea is at the core of our proof of Theorem~\ref{thm:ising:zero-free}. 

Intuitively the SAW tree $T$ of a graph $G = (V, E)$ and a vertex $v \in V$ is constructed by considering 
all of the self-avoiding walks from~$v$ in~$G$ and
storing these in a tree~$T$. The root of~$T$ is the walk consisting of
the single vertex~$v$, and two self-avoiding walks are connected in~$T$ if one of them is a strict sub-walk of the other with maximal length. We refer to \cite[Appendix~A]{Liu2019Fisher} for a formal construction. Some of the leaves of the tree $T$ are pinned according to a systematic procedure that is described in \cite[Appendix~A]{Liu2019Fisher}. It is important to note that if $G$ has maximum degree $\Delta$, then every node of $T$ has at most $d := \Delta-1$ children, except possibly the root of $T$, which might have $\Delta$ children. We will use the following result that relates $\ising(G; x)$ and $\ising(T; x)$.
 
 \begin{proposition}[{\cite[Proposition B.1]{Liu2019Fisher}}] \label{prop:tsaw}
 	Let $G$ be a connected graph and let $v$ be a vertex of $G$. Let $T$ be the SAW tree of $(G,v)$. Then the polynomial $\ising(G; x)$ divides the polynomial $\ising(T; x)$. In particular, if $\beta \in \mathbb{C}$ is such that $\ising(T; \beta) \ne 0$, then it also holds that $\ising(G; \beta) \ne 0$.
 \end{proposition}
 
 As a consequence of Proposition~\ref{prop:tsaw}, we can translate zero-free results for trees to zero-free results for graphs. Our proof of Theorem~\ref{thm:ising:zero-free} uses this approach. In the rest of this section we recall some tools to study the partition function of the Ising model on trees.
 
\begin{definition} \label{def:ratios}
Let $T$ be a tree (possibly with some pinned leaves) and let $v$ be its root. For each $j \in \{0,1\}$,  we define $Z_{v}^{j}(T; \beta)$ as the sum of $\beta^{m(\sigma)}$ over the configurations $\sigma$ of $T$ that have $\sigma(v) = j$, so $\ising(T; \beta) = Z_{v}^{0}(T; \beta) + Z_{v}^{1}(T; \beta)$. We define the ratio
\begin{align*}
  R\left(T, v; \beta \right) = \frac{Z_v^1\left(T; \beta\right)}{ Z_v^0(T; \beta)}.
\end{align*}
\end{definition}
The ratio $ R\left(T, v; \beta \right)$ is a rational function on $\beta$. If $Z_{v}^{0}(T; \beta) \ne 0$, we note that $\ising(T; \beta) = 0$ if and only if $ R\left(T, v; \beta \right) = -1$, so we can study the zeros of the partition function by studying these ratios. It turns out that the ratios $R(T, v, \beta)$ can be computed recursively.  Let us consider the function
\begin{equation*}
  \label{eq:f}
  F_{\beta, k}\left(z_1, \ldots, z_k\right) = \prod_{j = 1}^k h_{\beta}\left(z_j\right),
\end{equation*}
where $h_{\beta}(z) = (\beta z + 1) / (\beta + z)$ for any $z \in \mathbb{C}$. Then if $(T_1, v_1), \ldots, (T_d, v_k)$ are the trees with roots $v_j$ hanging from the root of $T$, one can check that
\begin{equation} \label{eq:tree-recursion}
	R(T, v; \beta) =  F_{\beta, k}(r_1, \ldots, r_k),
\end{equation}
 where $r_j = R(T_j, v_j; \beta)$ for all $j \in [k] := \{1, \ldots, k\}$, see for instance \cite{Liu2019Fisher}\footnote{In  \cite{Liu2019Fisher} the authors work with $Z_G(b) = \ising(G; 1/b)  b^{\lvert E(G) \rvert}$, so they get $h_{b}(z) = (b + z) / (b z + 1) $ instead.}. 
 
\subsection{Proof of Theorem~\ref{thm:ising:zero-free}}

First, we introduce some notation that will be used repeatedly in this work.

\begin{definition} \label{def:R}
	Let $\delta > 0$. We define $\mathcal{R}(\delta)$ as the set of complex numbers $z$ such that $\left| (z-1) / (z+1) \right| \le \delta$.
\end{definition}

Definition~\ref{def:R} allows us to conveniently restate Theorem~\ref{thm:ising:zero-free} as $\ising(G; \beta) \ne 0$ for any graph $G$ with maximum degree at most $\Delta$ and any $\beta \in \mathcal{R}(\varepsilon_\Delta)$, where $\varepsilon_\Delta = \tan(\pi / (4\Delta - 4))$. Proposition~\ref{prop:R} gives some properties of the region $\mathcal{R}(\delta)$ that we need in our proofs. For $x \in \mathbb{C}$ and $r > 0$ real we denote $B(x, r) = \{z \in \mathbb{C} : \lvert z - x \rvert < r\}$, $\overline{B}(x, r) = \{z \in \mathbb{C} : \lvert z - x \rvert \le r\}$ and $C(x, r) = \{z \in \mathbb{C} : \lvert z - x \rvert = r\}$.

\begin{proposition} \label{prop:R}  
	 Let $\delta > 0$. The region $\mathcal{R}(\delta)$ satisfies the following properties:
	\begin{enumerate}
		\item \label{item:R:1} We have $1 \in \mathcal{R}(\delta)$ and $-1 \not \in \mathcal{R}(\delta)$.
		\item \label{item:R:2} If $\beta \in  \mathcal{R}(\delta)$, then $\beta^{-1} \in \mathcal{R}(\delta)$.
		\item \label{item:R:3} The map $\phi(z) = (z-1)/(z+1)$ has the following property. We have $\phi(C(0, 1)) = i \mathbb{R} = \phi^{-1}(C(0, 1))$, and $\{z \in \mathbb{C} : \mathrm{Re}(z) > 0 \} = \phi^{-1}(B(0, 1))$.
		 In particular, if $\delta = 1$, then $\mathcal{R}(\delta)$ is the set of complex numbers $z$ with $\mathrm{Re}(z) \ge 0$.
		\item \label{item:R:4} If $\delta \in (0,1)$, then $\mathcal{R}(\delta)$ is the closed disk $\overline{B}(c_\delta, r_\delta)$ with centre $c_\delta = (1+\delta^2)/(1-\delta^2)$ and radius $r_\delta = 2\delta /(1 - \delta^2)$. Moreover, in this case for every $z \in \mathcal{R}(\delta)$ we have $\lvert z \rvert \le c_\delta + r_\delta = (1 + \delta) / (1 - \delta)$.
	\end{enumerate} 

\end{proposition}
\begin{proof} We prove each property separately.
	\begin{enumerate}
		\item This property is trivial.
		\item Note that $(z^{-1} - 1) / (z^{-1} + 1) = (1 - z) / (1 + z) = - (z - 1) / (1 + z)$, so  $\beta \in  \mathcal{R}(\delta)$ if and only if $\beta^{-1} \in \mathcal{R}(\delta)$.
		\item One can check that the inverse of $\phi(z) = (z-1)/(z+1)$ is the Mobius map $\phi^{-1}(y) = - (1+y)/(y-1)$. Hence, $\lvert \phi(z) \rvert = 1$ if and only if $\lvert \phi^{-1}(z) \rvert = 1$, which happens exactly when $\lvert z + 1 \rvert = \lvert z-1\rvert$ or, equivalently, $z \in i \mathbb{R}$. This proves $\phi(C(0, 1)) = i \mathbb{R} = \phi^{-1}(C(0, 1))$. Note that $\lvert z - 1 \rvert < \lvert z + 1 \rvert$ if and only if $\lvert \mathrm{Re}(z) - 1 \rvert < \lvert \mathrm{Re}(z) + 1 \rvert$. The latter is equivalent to $\mathrm{Re}(z) > 0$. This shows that $\{z \in \mathbb{C} : \mathrm{Re}(z) > 0 \} = \phi^{-1}(B(0, 1))$, and the result follows.
		\item We note that $\mathcal{R}(\delta) = \{z \in \mathbb{C} : \lvert \phi(z) \rvert \le \delta \}  = \phi^{-1}(\overline{B}(0, \delta))$. We claim that $\phi^{-1}$ sends the circle $C(0, \delta)$ to the circle $C(c_\delta, r_\delta)$, where $c_\delta$ and $r_\delta$ are as in the statement. As $\phi^{-1}$ is a Mobius map, $\phi^{-1}(C(0, \delta))$ is a circle or a line of $\mathbb{C}$, see Section~\ref{sec:pre:cd} on rational maps. We take $3$ points in the circle $C(0, \delta)$ and show that they are in $C(c_\delta, r_\delta)$. The three points are $\delta, -\delta$ and $\delta i$. One can easily check that $\phi^{-1}(\delta) = (1 + \delta) / (1 - \delta) = c_\delta + r_\delta$ and  $\phi^{-1}(\delta) = (1 - \delta) / (1 + \delta) = c_\delta - r_\delta$. We also have
		\begin{equation*}
			\phi^{-1}(\delta i) - c_\delta = \frac{i - \delta}{i + \delta} - \frac{1 + \delta^2}{1- \delta^2} = - \frac{1 + i \delta}{i + \delta} \frac{2 \delta }{1 - \delta^2},
		\end{equation*}
	    so $\lvert \phi^{-1}(\delta i) - c_\delta  \rvert = r_\delta$ as we wanted. We conclude that $\phi^{-1}(C(0, \delta)) = C(c_\delta, r_\delta)$. Since $\phi^{-1}$ is holomorphic in $B(0,1)$ and $\phi^{-1}(0) = 1 \in \overline{B}(c_\delta, r_\delta)$, we obtain $\phi^{-1}(\overline{B}(0, \delta)) = \overline{B}(c_\delta, r_\delta)$ as we wanted. Finally, the point in $\overline{B}(c_\delta, r_\delta)$ with the largest norm is $c_\delta + r_\delta = (1 + \delta) / (1 - \delta)$. \qedhere
	\end{enumerate} 
\end{proof}

\begin{remark} \label{rem:tan}
	Let $\alpha, \beta \in [-\pi, \pi]$. Then it is well-known that if $\alpha, \beta, \alpha+\beta \not \in \pi/2 + \pi \mathbb{Z}$, we have
	\begin{equation*}
		\tan(\alpha+ \beta) = \frac{\tan(\alpha) + \tan(\beta)}{1 - \tan(\alpha) \tan(\beta)}.
	\end{equation*}
	In particular, we obtain the equality 
	\begin{equation*}
		\tan(2 \alpha) = \frac{2 \tan(\alpha)}{1-\tan(\alpha)^2}.
	\end{equation*}
	Let $f(x) = 2x/(1-x^2)$. This function is strictly increasing in $x \in [-1, 1]$. Hence, if we have $\alpha \in [-\pi/4, \pi/4]$ and $\tan(2 \alpha) = f(\delta)$ for some $\delta \in (0,1)$, we can conclude that $\tan(\alpha) = \delta$. This argument will be used in the proof of Theorem~\ref{thm:ising:zero-free}.
\end{remark}

\begin{lemma} \label{lem:R}  
	Let $\delta, \epsilon > 0$, let $\beta \in \mathcal{R}(\delta)$ and let $h_\beta(z) = (\beta z + 1) / (\beta + z)$. Then $h_\beta(\mathcal{R}(\epsilon)) \subseteq \mathcal{R}(\delta \epsilon)$. Moreover, we have $h_\beta(\infty) = \beta \in \mathcal{R}(\delta)$.
\end{lemma}
\begin{proof}
It is straightforward to check that for any $z \in \mathbb{C}$ with $z \ne -1$, we have
Equation~\eqref{eq:hbeta:property}, namely
	\begin{equation*}  
		\frac{h_{\beta}(z) -1 }{h_{\beta}(z) +1} = \frac{(\beta-1)(z-1)}{(\beta+1)(z+1)}. 
	\end{equation*}
    The result now follows from \eqref{eq:hbeta:property} and the definition of $\mathcal{R}(\varepsilon),  \mathcal{R}(\delta \varepsilon)$.
\end{proof}

We are now ready to prove Theorem~\ref{thm:ising:zero-free}. 

\begin{thmzeros}
	\statethmzeros
\end{thmzeros}

\begin{proof} 
	Let $\beta \in \mathcal{R}(\varepsilon_{\Delta})$. In light of Proposition~\ref{prop:tsaw}, we only have to prove that $\ising(T; \beta) \ne 0$ for all trees $T$ with maximum degree at most $\Delta$ with possibly some pinned leaves.  Let $v$ be the root of such a tree $T$. We are going to prove that $R(T, v; \beta) \ne -1$ and 
(unless $T$ consists of a single vertex, pinned to~$1$, in which case the Theorem is trivial) that
	$Z_{v}^0(T, v; \beta) \ne 0$. Note that both assertions combined imply that 
	\begin{equation*}
		\ising(T; \beta) = Z_{v}^0(T, v; \beta)  \left(1 + \frac{Z_{v}^1(T, v; \beta)  }{ Z_{v}^0(T, v; \beta) } \right) = Z_{v}^0(T, v; \beta)  \left(1 + R(T, v; \beta) \right) \ne 0
	\end{equation*}
	 as we want.
		
	First, we restrict ourselves to trees such that every node has at most $d:= \Delta-1$ children and possibly some its leaves are pinned. We claim that for such a tree $T$ with root $v$ we have
	\begin{enumerate}  
		\item $R(T, v; \beta) \in \mathcal{R}(1) \cup \{\infty\}$, that is, $R(T, v; \beta)$ has non-negative real part or $R(T, v; \beta) = \infty$ (Proposition~\ref{prop:R});
		\item if $T$ has height at least $1$, then $Z_{v}^0(T, v; \beta) \ne 0$ (a tree with only one vertex has height $0$ by definition).
	\end{enumerate}
	 We carry out the proof by induction on the height of the tree.  Let us consider the case when the tree  $T$ consists of only one vertex. Depending on whether the vertex is pinned or not, either $R(T, v; \beta) = 1$ and $Z^0_v(T; \beta) = 1$, or $R(T, v; \beta) \in \{0, \infty\}$  and $Z^0_v(T; \beta) \in \{1, 0\}$. In either case, $R(T, v; \beta) \in \mathcal{R}(1) \cup \{\infty\}$.
	 
	 Now let $T$ be a tree of height $l > 0$ and let us assume that our claim holds for any of the desired trees with height at most $l-1$. Let $T$ be a tree of height $l$ such that all nodes have at most $d$ children. Let $v$ be its root and let $(T_1, v_1), \ldots, (T_{k}, v_{k})$ be the trees hanging from this root. By assumption, $k \le d$.    Let $r_j = R(T_j, v_j; \beta)$ for all $j \in [k]$.  In view of \eqref{eq:tree-recursion}, we have
	\begin{equation}  \label{eq:ratios-product}
		R(T, v; \beta) =  \prod_{j = 1}^{k} h_\beta(r_j). 
	\end{equation}
 By our induction hypothesis, $r_j \in \mathcal{R}(1) \cup \{\infty\}$ for all $j \in [k]$. In light of Lemma~\ref{lem:R}, we find that $h_{\beta}(r_j) \in \mathcal{R}(\varepsilon_{\Delta})$ for all $j \in [k]$. This property will be enough to ensure that the product in \eqref{eq:ratios-product} yields a complex number with non-negative real part. Let us study the argument of any element of $\mathcal{R}(\varepsilon_{\Delta})$. By a trigonometry reasoning shown in Figure~\ref{fig:proof-base-case}, the argument of any element in $\mathcal{R}(\varepsilon_{\Delta})$ is in the interval $[-\theta, \theta]$ for $\theta$ such that 
   \begin{equation*}
   	\tan(\theta)  = \frac{r_{\varepsilon_{\Delta}}}  { \sqrt{c_{\varepsilon_{\Delta}}^2 -
   			r_{\varepsilon_{\Delta}}^2} } = \frac{2\varepsilon_{\Delta}}{1-\varepsilon_{\Delta}^2},
   \end{equation*}
   where $r_\delta$ and $c_\delta$ are defined in Proposition~\ref{prop:R}.
   In view of Remark~\ref{rem:tan} for $\alpha = \theta/2$, we conclude that $\tan(\theta/2) = \varepsilon_{\Delta}$ and, thus, $\theta/2 = \arctan(\varepsilon_{\Delta}) = \pi / (4 d)$. Therefore, the complex number $R(T, v; \beta)$ is the product of $k$ numbers with argument in $[-\theta, \theta] =  [- \pi / (2 d), \pi / (2 d)]$, so its argument is in $[-k \theta, k \theta] \subseteq [-\pi/2, \pi/2]$, where we used $k\le d$. This is equivalent to saying that $R(T, v; \beta)$ has non-negative real part as we wanted.  Note that when $l \ge 1$, we have also shown that  $R(T, v; \beta) \in \mathbb{C}$.
	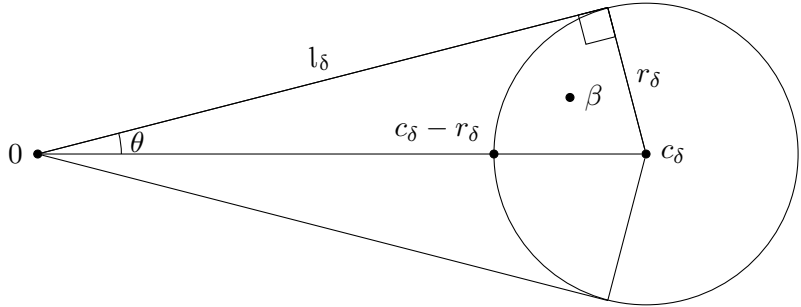
\begin{figure}[H]
	\centering
	\begin{tikzpicture}
		\node [fill, draw, circle, minimum width=3pt, inner sep=0pt, label=left:$0$] (0) at (0,0) {};
		
		\node [fill, draw, circle, minimum width=3pt, inner sep=0pt, label=right:$c_\delta$] (o) at (8, 0) {};

		\node [fill, draw, circle, minimum width=3pt, inner sep=0pt, label=above left:$c_\delta - r_\delta$] (a1) at (6, 0) {};
		
		\node [fill, draw, circle, minimum width=3pt, inner sep=0pt, label=right:$\beta$] (beta) at (7, 0.75) {};
		
		\node [circle,draw,name path=circle] (c) at (o) [minimum size=4cm] {};
		
		\draw (0)  --  (tangent cs:node=c,point={(0)},solution=1) 
		coordinate (B)
		(0) -- (tangent cs:node=c,point={(0)},solution=2)
		coordinate (C)
		(o)  -- (tangent cs:node=c,point={(0)},solution=2) 
		(o)  -- (tangent cs:node=c,point={(0)},solution=1) 
		(0) -- (o);
		
		\draw (0) -- (C) node [midway, above, sloped] (TextNode1) {$l_\delta$};
		\draw (o) -- (C) node [midway, right] (TextNode1) {$r_\delta$};

		\tkzMarkRightAngle[size=.4](o,C,0); 
		\pic [draw, -, angle radius=11mm, angle eccentricity=1.2, "$\theta$"] {angle = o--0--C};
	\end{tikzpicture}
	\caption{The disk $\mathcal{R}(\delta)$.}
	\label{fig:proof-base-case}
\end{figure}
 Let us now prove that $Z_{v}^0(T, v; \beta) \ne 0$. We have 
	\begin{equation}  \label{eq:z-product}
	Z_{v}^0(T; \beta) =  \prod_{j = 1}^{k} \left(\beta Z_{v_j}^0(T_j; \beta) +  Z_{v_j}^1(T_j; \beta)\right). 
  \end{equation}
  If $T_j$ has height at least $1$, then 
  \begin{equation*}
  	\beta Z_{v_j}^0(T_j; \beta) +  Z_{v_j}^1(T_j; \beta) = Z_{v_j}^0(T_j; \beta) \left(\beta + R(T_j, v_j; \beta)  \right) \ne 0,
  \end{equation*}
  where we used that $Z_{v_j}^0(T_j; \beta) \ne 0$  and $\mathrm{Re}(\beta + R(T_j, v_j; \beta)) > 0$ by the induction hypothesis (recall that $\mathrm{Re}(\beta) > 0$). If $T_j$ has height $0$, that is, $T_j$ has only one vertex, then, depending on whether this vertex is pinned or not,
  \begin{equation*}
  	\beta Z_{v_j}^0(T_j; \beta) +  Z_{v_j}^1(T_j; \beta) \in \{1, \beta, 1+\beta\}.
  \end{equation*}
  Therefore, the product in \eqref{eq:z-product} is a product of complex numbers that are non-zero, so $Z_{v}^0(T; \beta) \ne 0$ as we wanted.

Finally, to prove the Theorem, we consider a tree $T$ with maximum degree at most $\Delta$ and possibly some pinned leaves. Let $v$ be its root and let $(T_1, v_1), \ldots, (T_{k}, v_{k})$ be the trees hanging from this root. 
By the claim,
$R(T_j, v_j; \beta)\in \mathcal{R}(1) \cup \{\infty\}$ for all $j \in [k]$.  
By Lemma~\ref{lem:R},
$h_{\beta}(r_j) \in \mathcal{R}(\varepsilon_{\Delta})$ for all $j \in [k]$, so the argument of $h_{\beta}(r_j)$ is in  $[- \pi / (2 d), \pi / (2 d)]$ for all $j \in [k]$. It follows from this fact, $k \le \Delta$, and \eqref{eq:ratios-product}, that the argument of $R(T, v; \beta)$ is in 
	\begin{equation*}
	  [- k \pi / (2 d), k \pi / (2 d)] \subseteq 	[- \Delta \pi / (2 d), \Delta \pi / (2 d)] \subseteq [- 3\pi / 4, 3 \pi / 4],
	\end{equation*}
	 where we used that $\Delta \ge 3$. In particular,  $R(T, v; \beta)$ is not a negative real number, so  $R(T, v; \beta) \ne -1$. The fact that $Z_{v}^0(T; \beta) \ne 0$ follows analogously from \eqref{eq:z-product}, $Z_{v}^0(T; \beta)$ is a product of non-zero complex numbers.
\end{proof}

\begin{corfptas}
\statecorfptas
\end{corfptas}
\begin{proof}
  The proof follows from combining Theorem~\ref{thm:ising:zero-free}, the work of Patel and Regts \cite{Patel2017} and the work of Barvinok \cite{BarvinokBook}\footnote{The idea presented in the proof of Corollary~\ref{cor:ising:fptas} is known among experts, see for example, \cite{Liu2019Fisher}; we include it here for completeness. We note that the only properties of $\mathcal{S} = \{z \in \mathbb{C} : \lvert z - 1 \rvert / \lvert z + 1 \rvert < \epsilon_\Delta\}$ needed are that $\mathcal{S}$ is open and $\{t + (1-t)\beta : t \in [0,1]\} \subseteq \mathcal{S}$ for all $\beta \in \mathcal{S}$.}. Let $G$ and $\epsilon > 0$ be the inputs of our algorithm. We consider the polynomial $q_{G, \beta}(z) = \ising(G; 1 + z(\beta-1))$. We want to give an FPTAS for $q_{G; \beta}(1) = \ising(G; \beta)$. We claim that, on graphs with maximum degree at most $\Delta$, we can compute the $k$-th coefficient of $q_{G, \beta}(z)$ in polynomial time in $2^k$ and the size of $G$. This claim is proved for the more general case of the graph homomorphism partition function in the proof of \cite[Theorem 6.1]{Patel2017}. Recall that $1$ and $\beta$ are in the interior of the disk $\mathcal{R}(\varepsilon_{\Delta})$ (Proposition~\ref{prop:R}) so this is also true of an open interval around the line segment between them. Hence, there is $\delta > 0$ such that $1 + z (\beta - 1) \in \mathcal{R}(\varepsilon_{\Delta})$ for all $z \in R_\delta$, where $R_\delta$ is a strip of the form $R_\delta = \{z \in \mathbb{C} : -\delta \le \mathrm{Re}(z) \le 1+\delta, \lvert \mathrm{Im}(z) \rvert \le \delta \}$.  In light of Theorem~\ref{thm:ising:zero-free}, we conclude that  $q_{G, \beta}(z) \ne 0$ for all $z \in R_\delta$. In \cite[Section 2.2.2]{BarvinokBook} Barvinok constructs a polynomial $\phi_\delta$ and a real number $b_\delta > 1$ such that $\phi_\delta(0) = 0$, $\phi_\delta(1)$ and $\phi_\delta(z) \in R_\delta$ for any $z \in \overline{B}(0, b_\delta)$. Note that the polynomial $p_{G, \beta}(z) = q_{G, \beta}(\phi_\delta(z))$ does not vanish in $\overline{B}(0, b_\delta)$. Finally, we compute an approximation of $p_{G, \beta}(1) = \ising(G; \beta)$ as in \cite[Lemma 2.2.1]{BarvinokBook} using the truncated Taylor series of $\log p_{G, \beta}(z)$. The algorithm of Barvinok uses $O(\log(\deg( p_{G, \beta})/\epsilon)) = O(\log(\mathrm{size}(G) / \epsilon))$ coefficients of the Taylor series of $\log p_{G, \beta}(z)$.  Here the implicit ``$O$'' notation depends only on $\beta$. These coefficients can be computed using the algorithm of Patel and Regts in polynomial time in $\mathrm{size}(G)$ and $1/\epsilon$. We conclude that \cite[Lemma 2.2.1]{BarvinokBook} computes $Y$ such that $\lvert \log p_{G, \beta}(1) - Y \rvert \le \epsilon$ in polynomial time in $\size{G}$ and $1/\epsilon$. Let $z = \log p_{G, \beta}(1) - Y$ and $\hat{Z} = \exp(Y)$. Then we have $\hat{Z} = \ising(G; \beta) e^z$ and $\lvert z \rvert  \le \epsilon$ as we wanted.
\end{proof}

\subsection{Comparing our result to the state of the art}\label{sec:comparing}

In this section we gather all the results we are aware of on the zeros of the partition function of the Ising model and compare them to Theorem~\ref{thm:ising:zero-free}. We show that our result extends the state of the art significantly.

Results on the zeros of the graph homomorphism partition function can be particularised to the Ising model. Before stating these results, let us introduce this partition function. Let $G = (V, E)$ be an undirected graph, possibly with multiple edges or loops, and let $A = (a_{ij})$ be a $k \times k$ symmetric matrix of complex numbers. The \emph{graph homomorphism partition  function} is defined as
\begin{equation*}
	\hom(G; A) = \sum_{\phi \colon V \to [k]} \prod_{\{u,v\} \in E}  a_{\phi(u) \phi(v)},
\end{equation*}
where $[k]$ denotes $\{1, \ldots, k\}$. When $k = 2$ and $a_{11} = a_{22}$ we have
\begin{equation} \label{eq:homo}
	\hom(G; A) = a_{12}^{|E|} \sum_{\phi \colon V \to \{1,2\}} \prod_{\substack{\{u,v\} \in E : \\ \phi(u) = \phi(v)}}  \frac{a_{11}}{a_{12}} =  a_{12}^{|E|} \ising\left(G; \frac{a_{11}}{a_{12}}\right),
\end{equation}
recovering the partition function of the Ising model as a particular case.

To the best of our knowledge, the best result on  the zeros of the graph homomorphism partition function known up to date is the following result of Barvinok.

\begin{theorem}[{\cite[Theorem 7.1.4]{BarvinokBook}}] \label{thm:barvinok}
	For a positive integer $\Delta$, let
	\begin{equation} \label{eq:delta}
		\delta_{\Delta} = \max \left\{\sin \left(\frac{\alpha}{2}\right) \cos \left(\Delta \frac{\alpha}{2}\right) : 0 < \alpha < \frac{2 \pi}{3 \Delta} \right\}.
	\end{equation}
	Then for any graph $G = (V, E)$ with maximum degree at most $\Delta$, we have $\hom(G; A) \ne 0$ for any complex symmetric matrix $A$ with dimension $k \times k$ such that $|1- a_{ij}| \le \delta_{\Delta}$ for any  $i, j \in \{1, \ldots, k\}$. 
\end{theorem}

Theorem~\ref{thm:barvinok} can be naively translated to the Ising model by considering matrices of the form
\begin{equation*}
	\left[
	\begin{matrix}
		\beta & 1 \\
		1       & \beta
	\end{matrix}
	\right].
\end{equation*}
For those matrices, Theorem~\ref{thm:barvinok} says that $\ising(G, \beta) \ne 0$ when $|1 - \beta| \le \delta_{\Delta}$. One can obtain a stronger result for the Ising model if we apply \eqref{eq:homo} together with Theorem~\ref{thm:barvinok}.

\begin{corollary} \label{cor:region-barvinok}
	Let $\Delta$ be a positive integer, let $\delta_{\Delta}$ as in \eqref{eq:delta} and let 
	\begin{equation*}
		\beta \in \bigcup_{a \in \overline{B}(1, \delta_{\Delta})} \frac{1}{a} \overline{B}(1, \delta_{\Delta}).
	\end{equation*}
Then $\ising(G, \beta) \ne 0$ for any graph $G$ with maximum degree at most $\Delta$. 
\end{corollary}
\begin{proof}
	We  can write $\beta = a_{11} / a_{12}$ for $a_{11}, a_{12}\in B(1, \delta_\Delta)$. We consider the matrix
	\begin{equation*}
		A = \left[
		\begin{matrix}
			a_{11} & a_{12} \\
			a_{12}      & a_{11}
		\end{matrix}
		\right].
	\end{equation*}
	 By \eqref{eq:homo} and Theorem~\ref{thm:barvinok} we have $\ising(G, \beta) = a_{12}^{-\lvert E \rvert } \hom(G; A) \ne 0$ for any graph $G = (V, E)$ with maximum degree at most $\Delta$.
\end{proof}

The case $a = 1$ is the naive application of Theorem~\ref{thm:barvinok} mentioned after Theorem~\ref{thm:barvinok}. Taking $a = 1/\sqrt{\beta}$ in Corollary~\ref{cor:region-barvinok} gives the following corollary that can be found in the work of Mann and Bremner \cite{Mann2019}.

\begin{corollary}[{\cite[Corollary 7]{Mann2019}}]
	\label{cor:mann-bremner}
		Let $\Delta$ be a positive integer, let $\delta_{\Delta}$ as in \eqref{eq:delta} and let $\beta \in \mathbb{C}$ such that $|1 - 1/\sqrt{\beta}| \le \delta_{\Delta}$ and $|1 - \sqrt{\beta}|\le \delta_{\Delta}$. Then $\ising(G, \beta) \ne 0$ for any graph $G$ with maximum degree at most $\Delta$. 
\end{corollary}
\begin{proof}
	This is a particular case of Corollary~\ref{cor:region-barvinok} where $a$ is set to $ 1/\sqrt{\beta}$.
\end{proof}
In Lemma~\ref{lem:region-R} we show that the sets $\bigcup_{a \in \overline{B}(1, \delta)} \overline{B}(1, \delta) / a$ and $\mathcal{R}(\delta)$ are related.

\begin{lemma}\label{lem:region-R}
	For any $\delta \in (0,1)$, we have
	\begin{equation*}	
		\left[\frac{1-\delta}{1+\delta}, \frac{1+\delta}{1-\delta}\right] \subseteq \bigcup_{a \in \overline{B}(1, \delta)} \frac{1}{a} \overline{B}\left(1, \delta\right) \subseteq \mathcal{R}\left(\delta \right).
	\end{equation*}
\end{lemma}
\begin{proof}
	The first inclusion follows from the fact that 
	\begin{equation*}	
		\left[\frac{1-\delta}{1+\delta}, \frac{1+\delta}{1-\delta}\right] \subseteq \frac{1}{1+\delta} \overline{B}\left(1, \delta\right) + \frac{1}{1-\delta} \overline{B}\left(1, \delta\right).
	\end{equation*}
	In the rest of the proof we focus on the second inclusion. First, let us consider $a$ of the form $a = 1 + \delta e^{i \theta}$ for some $\theta \in [0, 2\pi)$. We show that $\overline{B}(1, \delta)/a  \subseteq \mathcal{R}(\delta)$. Note that $\overline{B}(1, \delta)/a = \overline{B}(1/a, \delta/|a|)$. Since $\overline{B}(1, \delta)/a$ and $\mathcal{R}(\delta)$ are convex, we only have to show that the border of $\overline{B}(1, \delta)/a$ is contained in  $\mathcal{R}(\delta)$. Let $\beta$ be in the border of $\overline{B}(1, \delta)/a$. We can write $\beta = (1 + \delta e^{\tau i}) / a  = (1 + \delta e^{\tau i}) / (1 + \delta e^{\theta i})$ for some $\tau \in [0, 2\pi)$. We have
	\begin{equation} \label{eq:region-R}
		\frac{\beta - 1}{\beta+1} = \delta \frac{e^{\tau i}-e^{\theta i}}{2 + \delta(e^{\tau i}+e^{\theta i})}.
	\end{equation}
	The norm of the right hand size of \eqref{eq:region-R} is bounded by $\delta$. This can be shown using \texttt{Mathematica} (see Appendix~\ref{sec:B} for the code). Hence, $\beta$ is in $\mathcal{R}(\delta)$ as we wanted. 
	
	Now we consider the case when $a = 1 + r e^{i \theta}$ for some $r \in (0, \delta)$. Let $f(z) = 1/z$ and $\Omega = f(\overline{B}(1, \delta))$. We claim that the set $f(\overline{B}(1, \delta))$ is convex. Let us finish the proof assuming this claim. The map $f$ maps the border of $\overline{B}(1, \delta)$ to the border of $f(\overline{B}(1, \delta))$. Thus, we can write $1/a = \lambda f(a_1) + (1-\lambda) f(a_2)$ for some $\lambda \in (0,1)$ and $a_1$ and $a_2$ in the circle of centre $1$ and radius $\delta$. We obtain,
	\begin{equation*}
		\frac{1}{a} \overline{B}\left(1, \delta\right)  = \lambda \frac{1}{a_1} \overline{B}\left(1, \delta\right) + (1-\lambda) \frac{1}{a_2} \overline{B}\left(1, \delta\right),
	\end{equation*}
	which is contained in $\mathcal{R}(\delta)$ due to the convexity of $\mathcal{R}(\delta)$ (Proposition~\ref{prop:R}, $\mathcal{R}(\delta)$  is a closed disk) and the fact that $\overline{B}(1, \delta)/a_1$ and $\overline{B}(1, \delta)/a_2$ are contained in $\mathcal{R}(\delta)$ as we argued at the beginning of this proof.
	
	Finally we prove that $f(\overline{B}(1, \delta))$ is a convex set.   The map $f(z) = 1/z$ is a Mobius map, so it sends lines and circles to lines and circles, see Section~\ref{sec:pre:cd} on rational maps. The points $f(1-\delta), f(1+\delta), f(1+ i\delta)$ are not aligned so $f$ sends the circle of center $1$ and radius $\delta$ to a circle determined by the points $1/(1-\delta), 1/(1+\delta)$ and $1/(1+i \delta)$.  Note that $f(1) = 1$ is in the disk determined by this circle, so $f(\overline{B}(1, \delta))$ is a closed disk and, in particular, convex.
\end{proof}

As a consequence of Lemma~\ref{lem:region-R}, the non-zero regions of the partition function of the Ising model given by Corollaries~\ref{cor:region-barvinok} and~\ref{cor:mann-bremner}  are contained in $\mathcal{R}(\delta_{\Delta})$, where $\delta_\Delta$ is as in \eqref{eq:delta}.  Recall that the zero-free region given in Theorem~\ref{thm:ising:zero-free} is $\mathcal{R}(\varepsilon_{\Delta})$, where $\varepsilon_{\Delta} = \tan(\pi/(4 \Delta - 4))$. By the definition of $\mathcal{R}(\delta)$ (see Definition~\ref{def:R}), we have $\mathcal{R}(\delta_{\Delta}) \subseteq \mathcal{R}(\varepsilon_{\Delta})$ if and only if $\delta_\Delta \le \varepsilon_{\Delta}$. Hence, now we compare  $\delta_\Delta$ and $\varepsilon_{\Delta}$. Figure~\ref{fig:epsilon} shows that our $\varepsilon_\Delta$ is significantly larger than $\delta_{\Delta}$, so Theorem~\ref{thm:ising:zero-free} improves the results of Barvinok, Mann and Bremner (Corollaries~\ref{cor:region-barvinok} and~\ref{cor:mann-bremner}) considerably, particularly for the case $\Delta = 3$. 
See also Figure~\ref{fig:zero-free}.
The limit of $\varepsilon_{\Delta} (\Delta -1)$ is $\pi / 4=0.785...$, whereas we have  numerically checked that  $\delta_\Delta (\Delta -1)$ tends to 
$0.561...$. Thus, our result is stronger for all~$\Delta$, and in the limit as $\Delta\to \infty$.

\pgfplotsset{width=0.7\textwidth,height=0.4\textwidth}
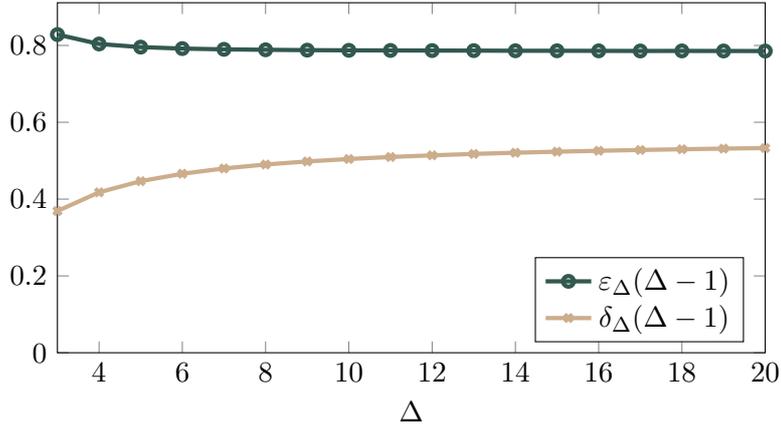
\begin{figure}[H]
	\centering
	\begin{tikzpicture}
		\begin{axis}[ 
			xlabel={$\Delta$},
			legend pos=south east,
			legend entries={$\varepsilon_\Delta (\Delta - 1)$,$\delta_\Delta (\Delta - 1)$},
			xmin=3, xmax=20,
			ymin=0, 
			samples at={3,...,20}
			] 
			\addplot[ultra thick, color=darkgreen, mark=o] {tan(deg(pi/(4*(x-1)))) * (x-1)}; 
			\addplot[ultra thick, color=darkbrown,mark=x] coordinates {
				(3, 0.369009) 
				(4, 0.417603) 
				(5, 0.446673) 
				(6, 0.465984) 
				(7, 0.479732) 
				(8, 0.490015) 
				(9, 0.497994) 
				(10, 0.504365) 
				(11, 0.509568) 
				(12, 0.513898) 
				(13, 0.517557) 
				(14, 0.52069) 
				(15, 0.523403) 
				(16, 0.525774) 
				(17, 0.527865) 
				(18, 0.529723) 
				(19, 0.531383) 
				(20, 0.532877)
			};
		\end{axis}
	\end{tikzpicture}
	\caption{Plot of the quantities $\varepsilon_\Delta (\Delta - 1)$  and $\delta_\Delta (\Delta - 1)$.}
	\label{fig:epsilon}
\end{figure}

When the edge interaction~$\beta$ is  real,  the partition function of the 
anti-ferromagnetic  Ising model 
(corresponding to the case $\beta<1$)
has an FPTAS when  $\beta$ is in the uniqueness region of the infinite $\Delta$-regular tree \cite{Sinclair2014}. This uniqueness region turns out to be the interval $((\Delta -2) / \Delta, \Delta/ (\Delta -2))$. When $\beta > 1$ (corresponding to the ferromagnetic Ising model) the partition function has an FPRAS on arbitrary graphs (with no restrictions on the degree) by the work of Jerrum and Sinclair \cite{Jerrum1993}. However, in the case of the anti-ferromagnetic Ising model ($\beta \in (0,1)$) this uniqueness/non-uniqueness phase transition is also a computational transition for the complexity of approximating the partition function of the Ising model: unless $\mathsf{RP} = \mathsf{NP}$, for all $\Delta \ge 3$, there is no FPRAS for approximating the partition function on graphs of maximum degree $\Delta$ when $\beta \in (0, (\Delta-2)/\Delta)$ \cite{Galanis2016}. Interestingly, the uniqueness interval $((\Delta -2) / \Delta, \Delta/ (\Delta -2))$ is contained in a complex zero-free region of the partition function of the Ising model.

 \begin{theorem}[{\cite[Theorem 1.2]{Liu2019Fisher}}] \label{thm:liu}
   Let $\Delta$ be an integer with $\Delta \ge 3$. For any $\beta \in ((\Delta-2) / \Delta, \Delta/ (\Delta-2))$, there exists a $\delta > 0$ such that for all $\beta' \in \mathbb{C}$ with $|\beta' - \beta| < \delta$, we have $\ising(G; \beta') \ne 0$ for any graph $G$ with maximum degree at most $\Delta$.
 \end{theorem}

The argument given in the proof \cite[Theorem 1.2]{Liu2019Fisher} uses continuity to prove the existence of $\delta > 0$ as in the statement. Hence, the zero-free region is not given explicitly. We note that Theorem~\ref{thm:liu} cannot be extended to include more edge interactions $\beta \in (0, (\Delta-2)/\Delta)$ unless   $\mathsf{RP} = \mathsf{NP}$ as, by the work of Patel and Regts \cite{Patel2017}, this would imply easiness of approximating $\ising(G; \beta)$ on graphs with maximum degree $\Delta$.

A recent paper of Barvinok and Barvinok gives another region where $\ising(G; \beta)$ is non-zero \cite{Barvinok2020}. This result actually applies to the multivariate Ising model with a field but it can be stated for our particular case as follows.
\begin{theorem}[{\cite[Theorem 1.1]{Barvinok2020}}] \label{thm:barvinok:ising}
  Let $\Delta$ be a positive integer with $\Delta \ge 3$.  Let $a \in \mathbb{C}$ and let $\beta = e^{2a}$. Suppose that for some $0 < \delta < 1$ we have $|\mathrm{Re}(a)| < (1-\delta)/\Delta$ and $| \mathrm{Im}(a) | \le \delta^2/ (10 \Delta)$. Then $\ising(G; \beta) \ne 0$ for any graph $G$ with maximum degree at most $\Delta$.   
 \end{theorem}

Generally Theorems~\ref{thm:ising:zero-free} and~\ref{thm:barvinok:ising} are incomparable for $\Delta$ large enough, both of them cover edge interactions that escape from the other result. However, for $\Delta = 3$ Barvinok's region is contained in 
the region $\mathcal{R}(\tan(\pi / (4\Delta - 4)))$ covered by
Theorem~\ref{thm:ising:zero-free}. This is depicted in Figure~\ref{fig:zero-free}, where all the regions introduced in this section have been plotted for $\Delta = 3$.

\section{Proof of our hardness results} \label{sec:hardness}

In this section we prove Theorem~\ref{thm:hardness}. Our hardness proof uses the reduction developed by Goldberg and Jerrum \cite{Goldberg2014}. Goldberg and Jerrum developed this reduction to obtain $\numP$-hardness results for determining the sign of the Tutte polynomial, which (in its random-cluster formulation) is a two-variable polynomial with variables $q$ and $y$ that includes the partition function of the Ising model as a particular case with the change of variables $q = 2$ and $y = \beta$. This reduction has been further refined in \cite{Goldberg2017} to obtain $\numP$-hardness results for the problem of approximating the norm of the Tutte polynomial. Further refinements have been obtained in \cite{Galanis2020}, where the authors give a reduction from exact evaluation of the Tutte polynomial to approximation of this polynomial with complex edge interactions. This later refinement is particularly useful when obtaining hardness results 
for restricted families of graphs 
for which exact evaluation of the Tutte polynomial remains hard. In \cite{Galanis2020} this was exploited to prove hardness of approximation for planar graphs whereas here we exploit this reduction to obtain hardness of approximation for bounded-degree graphs for the partition function of the Ising model.

In order to apply the reduction given in  \cite{Galanis2020}  there are a few technical results that we have to develop. The reduction is based on the binary search / interval shrinking technique of Goldberg and Jerrum \cite{Goldberg2014} and this requires us to be able to implement approximations of any real edge interaction efficiently. We formalised this property in Definition~\ref{def:impplane} (recall that we denote by $\algebraic$ the set of real algebraic numbers and we denote by $\algcomplex$ the set of complex algebraic numbers).

Our work shows that we can implement the complex plane in polynomial time for most pairs $(\Delta, \beta)$.   These pairs $(\Delta, \beta)$ are those where $\beta \not \in \mathbb{R}$ and $\lvert (\beta -1) / (\beta + 1) \rvert > 1 / \sqrt{\Delta-1}$, see Lemma~\ref{lem:ising:implementations}. If we could extend Lemma~\ref{lem:ising:implementations} to other pairs  $(\Delta, \beta)$, then we could automatically extend Theorem~\ref{thm:hardness} to these pairs. In other words, the limiting factor in the proof of Theorem~\ref{thm:hardness} is being able to $(\Delta, \beta)$ implement the real line. In fact, most of our work is devoted to this task. The proof of Lemma~\ref{lem:ising:implementations} heavily uses the results of \cite{Bezb} as an input. In \cite{Bezb}, the authors $(\Delta, \lambda)$ implement the complex plane in polynomial time for the independent set polynomial for most complex activities $\lambda$\footnote{Here $\lambda$ is a vertex activity of the independent set polynomial, and a graph $G$ with terminal $v$ $(\Delta, \lambda)$-implements $\lambda'$ if $\deg(v) = 1$ and $\lambda' = R(G, v; \lambda)$}. Their arguments apply results of complex dynamics in conjunction with the tree recurrence for the independent set polynomial.
It turns out that the arguments presented in \cite{Bezb} can be generalised so that they can be applied to other spin systems, and we do so in Appendix~\ref{sec:A}. We refer to \cite[Section 2]{Bezb} or our Appendix~\ref{sec:A} for a description of this complex dynamics approach.

This section is organised as follows. First, in Section~\ref{sec:pre:algebraic} we indicate how we represent algebraic numbers and how we compute with them. In Section~\ref{sec:pre:implementations} we introduce implementations formally as well as some constructions that we use in our proofs. In Section~\ref{sec:pre:cd} we state the results of complex dynamics that are used in our implementation results. In Section~\ref{sec:hardness:programs} we introduce the framework needed to implement the real line in polynomial-time (using Appendix~\ref{sec:A} as an input). In Section~\ref{sec:hardness:implementations} we use this framework to prove Lemma~\ref{lem:ising:implementations}. Then in Section~\ref{sec:hardness:reduction} we use Lemma~\ref{lem:ising:implementations} in conjunction with the reductions of \cite{Galanis2020} to prove our hardness results.

\subsection{Computing with algebraic numbers} \label{sec:pre:algebraic}

Our hardness result (Theorem~\ref{thm:hardness}) involves algebraic edge interactions. Here we specify how we represent algebraic numbers and how we perform arithmetic computations with them.  An algebraic number $z$ can be represented as its minimal polynomial $p$ and a rectangle $R$ of the complex plane such that $z$ is the only root of $p$ in $R$. We can compute the addition, subtraction, multiplication, division and conjugation of algebraic numbers in polynomial time in the length of their representations, see~\cite{Strzebonski1997} for details. As a consequence, we can also compute the real and imaginary parts of $z$ and the norm of $z$, which are algebraic numbers themselves, in polynomial time. Note that an algebraic number is $0$ if and only if its minimal polynomial is $x$, which can be easily checked in this representation. Hence, we can also determine in polynomial time whether two algebraic numbers $z_1$ and $z_2$ are equal by checking if $z_1 - z_2$ is $0$. 

When $z$ is a real algebraic number, we can simply represent it as its minimal polynomial $p$ and an interval $I$ with rational endpoints such that $z$ is the only root of $p$ in $I$. If we are given a real algebraic number $z$ with this representation, then we can approximate it as closely as we want by applying Sturm sequences and binary search~\cite{Emiris2008}. In fact, for $z_1$ and $z_2$ real algebraic numbers, Sturm sequences also allow us to check whether $z_1 \ge z_2$ in time polynomial in the length of the representations of $z_1$ and $z_2$. See~\cite{Emiris2008} for more details and complexity analysis. In some of our algorithms we have to check, for $z, x \in \algcomplex$ and $r \in \algebraic$ with $r > 0$, if $z \in B(z, x)$ or $z \in \overline{B}(z, x)$. Note that $\lvert z - x \rvert$ is a real algebraic number, and thus, we can check if $\lvert z - x \rvert  < r$ or  $\lvert z - x \rvert  = r$ in polynomial time in the sizes of $x, z$ and $r$.

\subsection{Implementing weights, series compositions and parallel compositions} \label{sec:pre:implementations}

In this section we define implementations and series and parallel compositions for the Ising model. These concepts have been used several times to obtain hardness results for partition functions, see for instance  \cite{Goldberg2014} for definitions in the context of Tutte polynomials or  \cite{Dyer2000} for definitions for the graph homomorphism partition function. Here we restrict ourselves to the partition function of the Ising model.

We will make use of the following notation. Let $H = (V, E)$ be a graph and let $s$ and $t$ be two distinct vertices of $H$. For $j, k \in \{0,1\}$ we define
\begin{align*}
  Z_{st}^{jk}(H; \beta) = \sum_{\substack{\sigma \colon V \to \{0,1\} \\ \sigma(s) = j,\, \sigma(t) = k}} \beta^{m(\sigma)} .
\end{align*}
The \emph{interaction matrix of $H$} at $(s,t)$ is the matrix
\begin{equation*}
  I_{st}(H; \beta) = \left[
    \begin{matrix}
      Z_{st}^{00}(H; \beta) & Z_{st}^{01}(H; \beta)  \\
      Z_{st}^{10}(H; \beta)  & Z_{st}^{11}(H; \beta)
    \end{matrix}
  \right].
\end{equation*}

We say that the graph $H$  $\beta$-\emph{implements} the weight $w$ if there are vertices $s$ and $t$ in $H$ such that the interaction matrix $I_{st}(H; \beta)$ is of the form
\begin{equation*}
	C
	\left[
	\begin{matrix}
		w & 1 \\
		1 & w
	\end{matrix}
	\right]
\end{equation*}
for some complex number $C$ with $C \ne 0$ or, equivalently, $Z_{st}^{01}(H; \beta)\ne 0$ and we have $Z_{st}^{11}(H; \beta) / Z_{st}^{01}(H; \beta) = w$. The point of implementations is that if we substitute an edge $e$ with weight $w$ of a graph $G$ by the graph $H$ (identifying the endings of $e$ with the vertices $s$ and $t$), the value of the partition function stays the same up to the factor $C = Z_{st}^{01}(H; \beta)$, see for instance \cite[Lemma 9]{Galanis2020}. Hence, if we have an oracle to evaluate the partition function of the Ising model at $\beta$ and we know $C$, we can use this oracle to evaluate this partition function at $w$. This idea is exploited in many hardness reductions, see \cite{Galanis2020} and the references therein. In this paper we are interested in graphs with bounded degree, so in order to use this construction while maintaining the maximum degree of the graphs involved, the vertices $s$ and $t$ should have degree $1$ in $H$. This is formalised in the following definition.
\begin{definition}  \label{def:implementation}
Let $\Delta \ge 2$ be an integer and $\beta \in \mathbb{C} \setminus \{0\}$. Let $G$ be a graph. We say that $G$ $(\Delta, \beta)$-implements the edge interaction $\beta' \in \mathbb{C}$ if $G$ has maximum degree at most $\Delta$  and distinct vertices $s$ and $t$ of degree $1$ such that $G$ $\beta$-implements $\beta'$ with the terminals $s$ and $t$.  We say that $(\Delta, \beta)$ implements the edge interaction $\beta' \in \mathbb{C}$ if there is a graph $G$  that $(\Delta, \beta)$-implements $\beta'$.  More generally, we say that $(\Delta, \beta)$ implements a set of edge interactions $S \subseteq \mathbb{C}$ if $(\Delta, \beta)$ implements $\beta'$ for any $\beta' \in S$.
\end{definition}

It is important to know that implementations are transitive, that is, if $H$ $(\Delta, \beta)$-implements the weight $w$ and $J$ $(\Delta, w)$-implements the weight $\gamma$, it is not difficult to construct a graph that $(\Delta, \beta)$-implements $\gamma$.

We now state a systematic way to implement edge interactions. Let $H_1 = (V_1, E_1)$ and $H_2 = (V_2, E_2)$ be two graphs. For each $j \in \{1,2\}$, let $s_j, t_j \in V_j$ be two distinct vertices.
\begin{enumerate}
\item The \emph{parallel composition} of $(H_1, s_1, t_1)$ and $(H_2, s_2, t_2)$ is the graph $H$  constructed by considering the union of $H_1$ and $H_2$ and identifying $s_1$ with $s_2$ and $t_1$ with $t_2$. The interaction matrix $I_{s_1 t_1}(H; y)$ is the Hadamard product (or component-wise product) of the interaction matrices  $I_{s_1 t_1}(H_1; y)$ and  $I_{s_2 t_2}(H_2; y)$.  Hence, if $(H_j, s_j, t_j)$ implements $w_j$ for $j \in \{1,2\}$, then $(H, s_1, t_1)$ implements $w = w_1 w_2$.

\item The \emph{series composition} of $(H_1, s_1, t_1)$ and $(H_2, s_2, t_2)$ is the graph $H$ constructed by considering the union of $H_1$ and $H_2$ and identifying $t_1$ with $s_2$. The interaction matrix $I_{s_1 t_2}(H; y)$ is the product of the interaction matrices  $I_{s_1 t_1}(H_1; y)$ and  $I_{s_2 t_2}(H_2; y)$. Hence, if $(H_j, s_j, t_j)$ implements $w_j$ for $j \in \{1,2\}$,  $(H, s_1, t_1)$ implements the edge interaction $w = (w_1 w_2 + 1) / (w_1 + w_2)$. Note that this operation is commutative, the series composition of $(H_1, s_1, t_1)$ and $(H_2, s_2, t_2)$ implements the same weight as the series composition of $(H_2, s_2, t_2)$ and $(H_1, s_1, t_1)$. 
\end{enumerate}

Series compositions are particularly helpful when working with graphs with bounded degree. In our constructions we usually consider the series composition of a graph $H$ that $(\Delta, \beta)$-implements a weight $w$ and a path of length $1$ with edge interaction $\beta$. This allows us to have a terminal vertex with degree $1$ in the resulting graph. This construction implements the edge interaction
\begin{equation} \label{eq:hbeta}
    h_\beta(w) = \frac{\beta w + 1}{\beta + w}.
\end{equation}
The Mobius map $h_\beta$ arises very frequently in this work and plays an important role in our arguments.

 \subsection{Iteration of complex rational maps} \label{sec:pre:cd}
 
In Appendix~\ref{sec:A} we extend the work on implementations for the independent set polynomial given in \cite{Bezb} to a more general setting so that these results can be applied to other partition functions, such as the partition function of the Ising model.   The technique developed in \cite{Bezb} uses several results from complex dynamics that we recall here. These complex dynamics results are also used in this section when implementing edge interactions for the Ising model. We gather all this material in this section. We refer to \cite{RiemannSurfaces} for an introduction to Riemann surfaces and to \cite{Beardon, Milnor2006} for an introduction to complex dynamics.

By $\widehat{\mathbb{C}} = \mathbb{C} \cup \{\infty\}$ we denote the Riemann sphere. The Riemann sphere is a metric space with the \emph{chordal metric} $d(\cdot, \cdot)$, given by
\begin{equation*}
  d(z, w) = \frac{2 \lvert z - w \rvert}{ \left( 1 + \left| z \right|^2 \right)^{1/2} \left( 1 + \left| w \right|^2 \right)^{1/2}}, \quad \text{ and } \quad d(z, \infty) = \lim_{w \to \infty} d(z, w) = \frac{2}{\left( 1 + \left| z \right|^2 \right)^{1/2}}.
\end{equation*}
The Riemann sphere is a Riemann surface, meaning that locally the Riemann sphere is homeomorphic to open subsets of $\mathbb{C}$. One can translate several results from complex analysis to Riemann surfaces. An example of such a result is the \emph{open mapping theorem}, see, for example, \cite[Theorem 2.2.2]{RiemannSurfaces}. 

\begin{proposition}[{Open mapping theorem for Riemann surfaces, \cite[Theorem 2.2.2]{RiemannSurfaces}}]\label{thm:openmapping}
  Let $X$ and $Y$ be Riemann surfaces. If $\phi \colon X \to Y$ is a non-constant holomorphic mapping, then $\phi$ is open, that is, $\phi(O)$ is an open subset of $Y$ for any open set $O \subseteq X$.
\end{proposition}

One can show that the set of holomorphic functions on the Riemann sphere is exactly the set of rational functions. A rational function of degree $d$ is a $d$-fold map on $\widehat{\mathbb{C}}$. Hence, the automorphisms on the Riemann sphere are precisely the rational functions of degree $1$. These are also known as \emph{Mobius maps} or \emph{Mobius transformations}. We use the following two properties of Mobius maps.

\begin{proposition}[{\cite[Theorem 5.7.3, part (f)]{RiemannSurfaces}}] \label{prop:mobius}
  If $C$ is a circle in $\widehat{\mathbb{C}}$ (i.e., $C$ is a circle in $\mathbb{C}$ or $C = L \cup \{\infty\}$ for some line $L$ in $\mathbb{C}$), then the image of $C$ under any Mobius map is also a circle in $\widehat{\mathbb{C}}$.
\end{proposition}

\begin{proposition}[{\cite[Proof of Theorem 5.8.2]{RiemannSurfaces}}] \label{prop:mobius:disk}
  Let $a \in \mathbb{C}$ with $\lvert a \rvert < 1$, $\theta \in \mathbb{R}$ and let $\phi(z) = e^{i \theta}(a z + 1) / (\overline{a} + z)$. Then the Mobius map $\phi$ fixes the circle $C(0, 1)$.
\end{proposition}

It is well-known that holomorphic complex maps are locally Lipschitz and this is exploited in \cite{Bezb}. Here we use a global Lipschitz property on the Riemann sphere, see Lemma~\ref{lem:lipschitz}.

\begin{lemma}[{\cite[Theorem 2.3.1]{Beardon}}] \label{lem:lipschitz}
  Let $f$ be a rational map. Then $f$ is a Lipschitz map on the Riemann sphere, that is, there is a constant $L > 0$ such that $d(f(z), f(w)) \le L d(z,w)$ for every $z,w \in \hat{\mathbb{C}}$, where $d$ is the chordal metric.
\end{lemma}

We conclude this section by introducing some results from complex dynamics. For a non-negative integer $n$ we denote by $f^n$ the $n$-fold iterate of $n$ (for $n = 0$, $f^0$ denotes the identity map). Let $f \colon \widehat{\mathbb{C}} \to \widehat{\mathbb{C}}$ be a rational map. Suppose that $\omega \in \widehat{\mathbb{C}}$ is a fixed point of $f$. If $\omega \in \mathbb{C}$, the \emph{multiplier} of $f$ at $\omega$ is defined as $f'(\omega)$. If $\omega = \infty$, the \emph{multiplier} of $f$ at $\omega$ is defined as $1/f'(\infty)$. The behaviour of the iterates $f^n$ near a fixpoint is characterised in terms of the multiplier $q$ of $f$ at this fixpoint. With this in mind, there are three types of fixpoints: \emph{attracting} if $\lvert q \rvert < 1$, \emph{neutral} or \emph{indifferent} if $\lvert q \rvert = 1$,  and \emph{repelling} if $\lvert q \rvert > 1$. We also need to introduce the \emph{Julia} set of $f$. We refer to \cite{Beardon} for a definition, here we only use the two following properties of Julia sets, Lemma~\ref{lem:repelling} and Theorem~\ref{thm:cd}.

\begin{lemma}[{\cite[Lemma 4.6]{Milnor2006}}] \label{lem:repelling}
  Let $f \colon \widehat{\mathbb{C}} \to \widehat{\mathbb{C}}$ be a rational map. Every repelling fixpoint of $f$ belongs to the Julia set of $f$.
\end{lemma}

 A set $U$ is a neighbourhood of $x$ if it contains a ball $B(x, r)$ for some $r >0$.  The exceptional set of a rational map $f$ is the set of points $z \in \widehat{\mathbb{C}}$ such that $[z] = \{z'\in \widehat{\mathbb{C}} : f^n(z') = f^m(z) \text{ for some integers } n,m \ge 0\}$ is finite.

\begin{theorem}[{\cite[Theorem 4.2.5]{Beardon}}] \label{thm:cd}
   Let $f \colon \widehat{\mathbb{C}} \to \widehat{\mathbb{C}}$ be a rational map with exceptional set $E_f$. Let $z_0$ be a point in the Julia set of $f$ and let $U$ be a neighbourhood of $z_0$. Then $\bigcup_{n = 0}^{\infty} f^n(U) = \widehat{\mathbb{C}} \setminus E_f$.
\end{theorem}

The exceptional points of $f$ can be characterised as follows.

\begin{lemma}[{\cite[Lemma 4.9]{Milnor2006} and \cite[Theorem 4.1.2]{Beardon}}] \label{lem:Ef}
  Let $f \colon \widehat{\mathbb{C}} \to \widehat{\mathbb{C}}$ be a complex rational map of degree at least $2$, and let $E_f$ be its exceptional set. Then, $\lvert E_f \rvert \le 2$. Moreover,
  \begin{itemize}
  \item if $E_f = \{\zeta\}$, then $\zeta$ is a fixed point of $f$ with multiplier $0$;
  \item if $E_f = \{\zeta_1, \zeta_2\}$ where $\zeta_1 \ne \zeta_2$, then $\zeta_1, \zeta_2$ have multiplier $0$ and either they are fixed points of $f$, or $f(\zeta_1) = \zeta_2$ and $f(\zeta_2) = \zeta_1$.
  \end{itemize}
\end{lemma}

\subsection{Ising and Mobius programs} \label{sec:hardness:programs}

In this section we introduce the framework that we use to implement the real line in polynomial time for the Ising model. Our proofs are based on the techniques developed in  \cite{Bezb} for the hardcore model. The idea behind the implementation results of \cite{Bezb} is the following one. First, we have to come up with a recursively-constructed gadget that implements a weight $f(z_1, z_2, \ldots, z_d)$ assuming that we can implement $z_1, \ldots, z_d$. Then we apply results of complex dynamics to the function $f$ in order to understand which points we can implement by iterating $f$. As we will see, it is important that the function $f$ is of the form $g(z_1  z_2 \cdots z_d)$, where $g$ is a Mobius map. In \cite{Bezb} the function $f$ naturally arises from the tree-recurrence for vertex implementations in the hardcore model. Unfortunately vertex-style implementations are useless in the Ising model; due to the perfect symmetry nothing interesting can be implemented through that route. Hence, we need to devise another way to obtain this type of recurrence in the Ising model. This is done in Proposition~\ref{prop:ising-program} for the Mobius map $g_\beta$, which is introduced in Definition~\ref{def:ising-program}.

\begin{definition} \label{def:ising-program}
	Let $\Delta \ge 3$ and $\beta \in \mathbb{C}$, and set $d := \Delta -1$. Let $h_\beta(x) = (\beta x + 1) / (\beta + x)$ and let $g_\beta(x) = h_\beta(h_\beta(x))$.  An Ising-program for $\beta$ is a sequence $a_0, a_1, \ldots$, starting with $a_{0} = \beta$ and satisfying
	\begin{equation*}
		a_k = g_{\beta}(a_{i_{k,1}} \cdots a_{i_{k,d_k}}) \quad \text{for } k \ge 1,
	\end{equation*}
	where $d_k \in [d]$ and $i_{k,1}, \ldots, i_{k,d_k} \in \{0, \ldots, k-1\}$. We say that the Ising program $a_0, a_1, \ldots$ generates $x \in \mathbb{C}$ if there exists an integer $k \ge 0$ such that $a_k = x$.
\end{definition}

We use these definitions for $h_\beta$ and $g_\beta$ several times in the rest of Section~\ref{sec:hardness}. We work with Ising-programs from a computational point of view. We represent an element $a_k$ of an Ising-program by the tuples $(i_{j,1}, \ldots, i_{j,d_{j}})$ for $j \in \{2, \ldots, k\}$, so computing $a_k$ means computing its representation as a sequence of tuples. Proposition~\ref{prop:ising-program} gives a gadget that implements the edge-interactions generated by an Ising-program.

\begin{proposition} \label{prop:ising-program}
	Let $\Delta \ge 3$ and $\beta \in \mathbb{C}$. Suppose that $a_0, a_1, \ldots$ is an Ising-program for $\beta$. Then, for every $k \ge 0$, we can compute from the representation of $a_k$  a graph $H_k$ with maximum degree at most $\Delta$ that $(\Delta, \beta)$-implements the edge interaction $a_k$. This computation takes $\mathrm{poly}(\Delta, k)$ steps.
\end{proposition}
\begin{proof}
	Set $d := \Delta -1$.  We give a recursive algorithm for the task of the statement. For $k = 0$, our algorithm outputs the graph with two vertices and one edge joining them. This graph implements the edge interaction $a_0 = \beta$. For $k > 0$, first our algorithm computes recursively graphs $H_0, \ldots, H_{k-1}$ such that $H_j$ $(\Delta, \beta)$-implements $a_j$ for every $j \in \{0, \ldots, k-1 \}$. Since $a_0, a_1, \ldots$ is an Ising-program, we have  $a_k = g_{\beta}(a_{i_{k,1}} \cdots a_{i_{k,d_k}})$ for $d_k \in [d]$ and some indexes $i_{k,1}, \ldots, i_{k,d_k} \in \{0, \ldots, k-1\}$. We have access to these indexes since we have access to the representation of $a_k$. Our algorithm constructs $H_k$ as the series composition of the following graphs: $H_0$, the parallel composition of the graphs $H_{i_{k,1}}, \ldots, H_{i_{k,d_k}}$, and $H_0$. The graph $H_k$ implements the same edge interaction as that implemented by the graph shown in Figure~\ref{fig:Hk}.
	\begin{figure}[H]
		\centering
			\begin{tikzpicture}
			\node [circle, draw, inner ysep=1mm] (S1) {$s$};
			\node [circle, draw, inner ysep=1mm, right = 1cm of S1] (S2) {};
			\node [circle, draw, inner ysep=1mm, right = 1.75cm of S2] (T2) {};
			\node [right = 0.7cm of S2] (dots) {$\vdots$};
			\node [circle, draw, inner ysep=1mm, right = 1cm of T2] (T1) {$t$};
			\path[-] 
			(S1) edge node [above] {$\beta$} (S2)
			(T1) edge node [above] {$\beta$} (T2)
			(S2) edge[bend left=90] node [above] {$a_{i_{k,1}}$} (T2)
			(S2) edge[bend right=90] node [below] {$a_{i_{k,d_k}}$} (T2);
		\end{tikzpicture}
		\caption{The recursive construction for $H_k$.} 
		\label{fig:Hk}
	\end{figure}
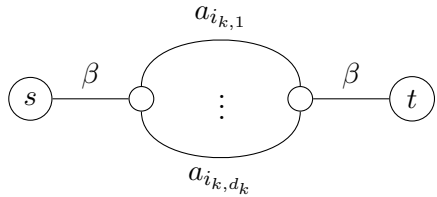
	By the properties of series and parallel compositions, see \eqref{eq:hbeta}, the graph $H_k$ implements the edge interaction $h_{\beta}(h_\beta(a_{i_{k,1}} \cdots a_{i_{k,d_k}})) = a_k$. Note that constructing $H_k$ from  $H_0, \ldots, H_{k-1}$ takes $\mathrm{poly}(\Delta, k)$ steps, so in total our algorithm has performed at most $k$ times that number of steps.
\end{proof}

 Note that $g_\beta$ is the composition of two Mobius maps and, thus, is a Mobius map. Hence, Ising-programs can be viewed as a particular case of Mobius-programs (see Definition~\ref{def:mobius-program}).

\newcommand{\statedefmobius}{Let $d \ge 2$ be an integer, $g$ be a Mobius map and $a_{0}\in \mathbb{C}$. A Mobius-program for $g$ and $d$ starting at $a_0$ is a sequence of complex numbers $a_0, a_1, \ldots$ of the form
  \begin{equation*}
    a_k = g(a_{i_{k,1}} \cdots a_{i_{k,d_k}}) \quad \text{for } k \ge 1,
  \end{equation*}
   where $d_k \in [d]$ and  $i_{k,1}, \ldots, i_{k,d_k} \in \{0, \ldots, k-1\}$. We say that the Mobius-program $a_0, a_1, \ldots$ generates $x \in \mathbb{C}$ if there exists a non-negative integer $k$ such that $a_k = x$. We usually omit $d$ when its value is clear from the context.}
\begin{definition} \label{def:mobius-program}
  \statedefmobius
\end{definition}

 In \cite{Bezb} the authors studied the points that can be generated by those Mobius-programs starting at $a_0 = \lambda$ for
\begin{equation*}
  g(x) = \frac{1}{1 + \lambda x}, 
\end{equation*}
where $\lambda$ is an activity for the independent set polynomial. They called this program a hardcore-program. The study of hardcore-programs is at the core of the hardness results for the independent set polynomial derived in \cite{Bezb}. It turns out that their results on hardcore-programs can be generalised to our setting of Mobius-programs. In short, their techniques imply that under some hypothesis we can efficiently generate approximations of any complex number with Mobius-programs algorithmically. First, let us introduce some notation that we use to generalise their results. 

\newcommand{\statedefpa}{Let $d \ge 2$ be an integer and let $g$ be a Mobius map.  Let $a_0 \in \mathbb{C}$. We say that $\gamma \in \mathbb{C} \setminus \mathbb{R}$ is program-approximable for $g$, $d$ and $a_0$ if for each $\epsilon > 0$, there is a Mobius-program for $g$ and $d$ starting at $a_0$ that generates a number $x \in \mathbb{C} \setminus \mathbb{R}$ with $0 < |\gamma - x| \le \epsilon$.}
\begin{definition} \label{def:pa}
    \statedefpa
\end{definition}

\newcommand{\statedefdpa}{Let $d \ge 2$ be an integer and let $g$ be a Mobius map with coefficients in $\algcomplex$. Let $a_0 \in \algcomplex$. We say that $\gamma \in \mathbb{C}$ is densely program-approximable in polynomial time for $g$, $d$ and $a_0$ if there is  
$r_{\gamma}\in \algebraic_{>0}$  such that for each positive integer $k$ there is an algorithm whose 
inputs are a rational $\epsilon > 0$ and $\lambda' \in B(\gamma, r_{\gamma}) \cap \algcomplex$ that computes,  in polynomial time in $\mathrm{size}(\epsilon)$ and $\mathrm{size}(\lambda')$,  $k$ distinct complex numbers  $x_1, x_2, \ldots, x_k$ generated by Mobius-programs for $g$ starting at $a_0$ with $|\lambda' - x_j| \le \epsilon$ for all $j \in [k]$. }
\begin{definition} \label{def:dpa}
  \statedefdpa
\end{definition}
In both definitions, we usually omit $d$ when its value is clear from the context. 

In \cite{Bezb} the author consider a fixed point of $f(x):= g(x^d)$ that is program-approximable for their choice of $g$ and $a_0$ and show that this fixed point is densely program-approximable in polynomial time for $g$ and $a_0$. Then they use this property in conjunction with results of complex dynamics to generate approximations of any complex number when the fixed point under consideration is repelling. This idea is made precise in Lemmas~\ref{lem:efficient-covering} and~\ref{lem:implementations}. We include the proofs of Lemmas~\ref{lem:efficient-covering} and~\ref{lem:implementations} in Appendix~\ref{sec:A}.

\newcommand{\statelemefficientcovering}{  Let $d$ be an integer with $d \ge 2$ and let $g$ be a Mobius map with coefficients in $\algcomplex$. Let $f(x):= g(x^d)$ and let $\omega$ be a fixed point of $f$. Let us assume that the following assumptions hold.
\begin{enumerate}
\item  $\omega$ is program-approximable for $g$, $d$ and $a_0 \in \mathbb{C}$;
\item  $\omega \ne 0$, $g'(\omega^d) \not \in \{0, \infty\}$ and $g''(\omega^d) \ne \infty$; 
\item Let $z := f'(\omega) / d = g'(\omega^d) \omega^{d-1}$. We have $0 < |z| < 1$ and $z \not \in \mathbb{R}$.
\end{enumerate}
Then $\omega$ is densely program-approximable in polynomial time for $g$, $d$ and $a_0$.}
\begin{lemma}[{\cite[Proposition 2.6 for Mobius-programs]{Bezb}}] \label{lem:efficient-covering}
  \statelemefficientcovering
\end{lemma}

\newcommand{\statelemmobiusimplementations}{Let $d$ be an integer with $d \ge 2$ and let $g$ be a Mobius map  with coefficients in
  $\algcomplex$   such that $g(\infty) \in \mathbb{C}$. Let $\omega \in \mathbb{C}$ be a repelling fixed point of $f(z) := g(z^d)$ that is densely program-approximable in polynomial time  for $g$ and $a_0 \in \algcomplex $. Let $E_f$ be the exceptional set of the rational map $f$. If $0, \infty \not \in E_f$, then the following holds.

  There is a polynomial-time algorithm such that, on input $\lambda \in \algcomplex$  and rational $\epsilon > 0$, computes an element $a_k$ of a Mobius-program for $g$ starting at $a_0$ with $\lvert \lambda - a_k \rvert \le \epsilon$.}
\begin{lemma}[{\cite[Proposition 2.2 for Mobius-programs]{Bezb}}]  \label{lem:implementations}
  \statelemmobiusimplementations
\end{lemma}

In order to apply Lemmas~\ref{lem:efficient-covering} and~\ref{lem:implementations}, first one has to find a fixed point $\omega$ with the properties described in Lemma~\ref{lem:efficient-covering}. When applying this result to the Ising model, we will set $\omega = 1$. Then one has to find the region of activities / edge interactions where the fixed point is repelling. All this work is carried out in Section~\ref{sec:hardness:implementations}.


\subsection{Proof of Lemma~\ref{lem:ising:implementations}} \label{sec:hardness:implementations}

In this section we use the framework introduced in Section~\ref{sec:hardness:programs} to prove Lemma~\ref{lem:ising:implementations}. This proof strongly uses the properties of the map $h_\beta$, which naturally arises in the context of the Ising model. The proof is divided into several technical lemmas. First, we show that $\omega = 1$ is a program-approximable fixed point for the Ising model (Lemma~\ref{lem:ising-program:1}). Then we prove Lemma~\ref{lem:ising:implementations} when $1/\sqrt{\Delta - 1} < \lvert \beta - 1 \rvert / \lvert \beta + 1 \rvert < 1$. Finally, we address the cases $\lvert \beta - 1 \rvert / \lvert \beta + 1 \rvert = 1$ and $\lvert \beta - 1 \rvert / \lvert \beta + 1 \rvert > 1$ separately, as they do not directly follow from the results of Section~\ref{sec:hardness:programs}.  We will use the following remark.

\begin{remark} \label{rem:stretch}
  Let $\beta, x \in \mathbb{C} \setminus \{1, -1\}$. Then it is straightforward to check that
  \begin{equation*}
    \frac{h_{\beta}(x) -1 }{h_{\beta}(x) +1} = \frac{(\beta-1)(x-1)}{(\beta+1)(x+1)}. 
  \end{equation*}
This equation was observed in \eqref{eq:hbeta:property} and plays a key role in the proof of Theorem~\ref{thm:ising:zero-free}.
 By induction we conclude that, for any positive integer $n$,
  \begin{equation*}
    \frac{h_{\beta}^{n}(x) -1 }{h_{\beta}^n(x) +1} =  \left( \frac{\beta-1}{\beta+1}  \right)^{n} \left(\frac{x-1}{x+1}\right).
  \end{equation*}
By rearranging this equation, we obtain, for any positive integer $n$, 
  \begin{equation*}
    h_{\beta}^n(x) =  -1  - \frac{2}{\left( \frac{\beta-1}{\beta+1}  \right)^{n} \left(\frac{x-1}{x+1}\right) - 1} = 1  + \frac{2}{\left( \frac{\beta+1}{\beta-1}  \right)^{n} \left(\frac{x+1}{x-1}\right) - 1}.
  \end{equation*}
Therefore, we have
  \begin{equation*}
    g_{\beta}^n(x) =  -1  - \frac{2}{\left( \frac{\beta-1}{\beta+1}  \right)^{2n} \left(\frac{x-1}{x+1}\right) - 1} = 1  + \frac{2}{\left( \frac{\beta+1}{\beta-1}  \right)^{2n} \left(\frac{x+1}{x-1}\right) - 1}.
  \end{equation*}
\end{remark}

\begin{lemma} \label{lem:ising-program:1}
  Let $\beta \in \mathbb{C}$ with $\beta \not \in \{i, -i\} \cup \mathbb{R}$. Then there is an Ising-program $a_0, a_1, \ldots$ such that the following holds. For every $\epsilon > 0$, there is a positive integer $k$ such that  $0 < \lvert 1 - a_k \rvert \le \epsilon$ and $a_k \not \in \mathbb{R}$.  
\end{lemma}
\begin{proof}
	We note $(x +1 ) / ( x  - 1) = 1 + 2/(x -1) \in \mathbb{R}$ if and only if $2/(x-1) \in \mathbb{R}$ or, equivalently, $x \in \mathbb{R}$. Thus, we have $(\beta +1 ) / ( \beta  - 1) \not \in \mathbb{R}$. These facts are used repeatedly in this proof. There are three cases:
	
	{\bf Case 1:} $0 < \lvert \beta - 1 \rvert / \lvert \beta  + 1 \rvert < 1$. First, let us describe the Ising-program. We define $a_0 = \beta$ and $a_j = g_{\beta}(a_{j-1})$ for every $j$ with $j \ge 1$. Note that this is an Ising-program. Since $a_j = g_{\beta}^{j}(\beta)$ for $j \ge 1$, by Remark~\ref{rem:stretch} we obtain
    \begin{equation} \label{eq:ising-program:1:case-1}
    a_j - 1 =  \frac{2}{\left( \frac{\beta+1}{\beta-1}  \right)^{2j+1} - 1},
  \end{equation}
  By hypothesis we have $ \lvert \beta +1 \rvert / \lvert \beta  - 1 \rvert > 1$, so the right hand side of \eqref{eq:ising-program:1:case-1} converges to $0$. Moreover, since  $(\beta +1 ) / ( \beta  - 1) \not \in \mathbb{R}$, there are infinitely many positive integers $j$ such that the right hand side of  \eqref{eq:ising-program:1:case-1}  is not real. Therefore, we can find a positive integer $k$ with  $0 < \lvert 1 - a_k \rvert \le \epsilon$ and $a_k \not \in \mathbb{R}$.

  	{\bf Case 2:} $\lvert \beta - 1 \rvert / \lvert \beta  + 1 \rvert > 1$. First, we give an Ising-program  $b_0, b_1, \ldots$ with the property that $b_j$ converges to $-1$. We define $b_0 = \beta$ and $b_j = g_{\beta}(b_{j-1})$ for every $j$ with $j \ge 1$. By Remark~\ref{rem:stretch} we have
    \begin{equation*}
    b_j + 1 =  -\frac{2}{\left( \frac{\beta-1}{\beta+1}  \right)^{2j+1} - 1},
    \end{equation*}  	
  	so $b_j + 1$ converges to $0$ because $\lvert \beta - 1 \rvert / \lvert \beta  + 1 \rvert > 1$. Once we have this Ising program, we define $a_0 = \beta$, $a_{2j-1} = b_j$ and $a_{2j} = g_\beta(a_{2j - 1}^2) =  g_\beta(b_j^2)$ for all $j \ge 1$. From Remark~\ref{rem:stretch} we obtain
  	\begin{equation} \label{eq:ising-program:1:case-2}
  		  \frac{a_{2j} -1 }{a_{2j} +1}  =  \frac{g_{\beta}(b_{j}^2) -1 }{g_{\beta}(b_{k}^2) +1}  =  \left( \frac{\beta-1}{\beta+1} \right)^{2} \frac{b_j^2-1}{b_j^2+1}.
  	\end{equation}
    The right hand side of \eqref{eq:ising-program:1:case-2} converges to $0$, so $a_{2j}$ converges to $1$. Moreover, \eqref{eq:ising-program:1:case-2} in combination with $b_j = g_\beta^j(\beta)$ and Remark~\ref{rem:stretch} gives $(b_j + 1) / (b_j -1) = ( (\beta -1)/(\beta+1) )^{-2j-1}$ and
  	\begin{equation*}
  		  \frac{a_{2j} -1 }{a_{2j} +1}  =  \left( \frac{\beta-1}{\beta+1} \right)^{2} \frac{(b_j+1)(b_j-1)^2}{(b_j-1) (b_j^2+1)}  = \left( \frac{\beta-1}{\beta+1} \right)^{-2j + 1} \frac{(b_j-1)^2}{b_j^2+1} .
  	\end{equation*}
  	Since $(\beta-1)(\beta+1)$ is not real and $(b_j-1)^2/(b_j^2+1)$ converges to $2$, there are infinitely many values of $j$ such that $(a_{2j} -1)/(a_{2j} +1)$ is not real. Equivalently, there are infinitely many values of $j$ such that $a_{2j}$ is not real. Hence, for every $\epsilon > 0$, there is a positive integer $k$ such that  $0 < \lvert 1 - a_{2k} \rvert \le \epsilon$ and $a_{2k} \not \in \mathbb{R}$.

{\bf Case 3:} $\lvert \beta - 1 \rvert / \lvert \beta  + 1 \rvert = 1$. Then we note that $\beta \in \mathbb{R} i$ (Proposition~\ref{prop:R}, Item~\ref{item:R:3}). We can write $\beta = c i$ with $c$ a real number with $c \not \in \{0, 1, -1\}$, where we used that $\beta \not \in \{0, i, -i\}$. We consider $\gamma = g_\beta (\beta^2) $. We claim that $\gamma \not \in \{i, -i\} \cup \mathbb{R}$ and $\lvert \gamma - 1 \rvert / \lvert \gamma  + 1 \rvert > 1$. Assuming this, we obtain our Ising-program as $b_0 = \beta$, $b_1 = \gamma$ and $b_j = a_{j-1}$ for all $j\ge 2$, where $a_0, a_1, \ldots$ is the Ising-program of Case~2 with $\beta = \gamma$.  We study $\gamma$ to conclude the proof. We note that $\gamma$ is the edge interaction implemented by the series composition of three edges with edge interactions $\beta, \beta^2$ and $\beta$. Recall that series compositions are commutative when it comes to the weight they implement (see Section~\ref{sec:pre:implementations}) and that the weight implemented by the series composition of two graphs implementing $w_1$ and $w_2$ is $(w_1  w_2 + 1) / (w_1 + w_2) = h_{w_1}(w_2) = h_{w_2}(w_1)$. Thus, we have $\gamma = h_\beta(h_\beta(\beta^2))$. This is also the edge interaction implemented by the series composition of three edges with edge interactions  $\beta, \beta$ and $\beta^2$, so we can also write $\gamma = h_{\beta^2}(h_\beta(\beta))$. From the expression $\gamma = h_\beta(h_\beta(\beta^2))$ and Remark~\ref{rem:stretch} we find that
  	\begin{equation*}
	\left|\frac{\gamma-1 }{\gamma +1} \right| =  \left|\frac{g_{\beta}(\beta^2) -1 }{g_{\beta}(\beta^2) +1} \right| =  \left| \frac{\beta-1}{\beta+1}  \right|^{2} \left| \frac{\beta^2-1}{\beta^2+1} \right|  = \left| \frac{1 + c^2}{ 1 - c^2 } \right| > 1,
    \end{equation*}
  where we used that $c \ne \pm 1$. In particular, we have $\gamma \ne \pm i$. From the expression $\gamma = h_{\beta^2}(h_\beta(\beta))$ we are going to show that $\gamma \not \in \mathbb{R}$, which would complete the proof. We have
  \begin{equation*}
  	 h_{\beta^2}(x) =  h_{-c^2}(x) = \frac{- c^2 x + 1}{-c^2 + x} = \frac{ \left(- c^2 x + 1 \right) \left( \overline{x} - c^2 \right) }{\left| x -c^2  \right|^2} = \frac{ - c^2 \lvert x\rvert ^2 + \overline{x}  + c^4 x  - c^2  }{\left| x -c^2  \right|^2}.
  \end{equation*}
 Since $c^4 \ne 1$ because $c \ne \pm 1$, we find that $h_{\beta^2}(x)$ is non-real for any non-real $x$. In particular this is the case for $x = h_\beta(\beta)$ as  $h_\beta(\beta) = (1 + \beta^2) / (2 \beta) =  - i (1-c^2)/ (2c) \not \in \mathbb{R}$. We conclude that $\gamma = h_{\beta^2}(h_\beta(\beta))$ is not real as we wanted.
\end{proof}

Using the notation of Appendix~\ref{sec:A} (see Definition~\ref{def:pa}), the statement of Lemma~\ref{lem:ising-program:1}  implies ``$1$ is program-approximable for $g_\beta$ and $a_0 = \beta$ for any $\beta \in \mathbb{C} \setminus (\mathbb{R} \cup \{i,-i\})$''. This is one of the three conditions that we have to check to apply Lemma~\ref{lem:efficient-covering} with  $\omega = 1$ in our current setting. Lemma~\ref{lem:ising-program:conditions} shows that the two other conditions hold for some edge interactions~$\beta$.

\begin{lemma} \label{lem:ising-program:conditions}
  Let $d$ be an integer with $d \ge 2$ and let $\beta \in \mathbb{C} \setminus  (\mathbb{R} \cup \{i, -i\})$. Let $f_\beta(x) = g_\beta(x^d)$, where $g_\beta$ is as in Definition~\ref{def:ising-program}. Let $z = f_\beta'(1)/d$. If $0 < \lvert \beta -1 \rvert / \lvert \beta  + 1 \rvert < 1$ and $\lvert \beta \rvert \ne 1$, then 
  \begin{enumerate}
  	\item $g_\beta'(1)\not \in \{0, \infty\}$, $ g_\beta''(1) \ne \infty$;
  	\item  $0 < \lvert z \rvert < 1$ and $z \not \in \mathbb{R}$.
  \end{enumerate}
\end{lemma}

\begin{proof}
  Let us determine $z$, $g_{\beta}'(1)$ and $g_{\beta}''(1)$. We have $h_{\beta}'(x) = (\beta^2 -1)/(\beta + x)^2$ and $h_{\beta}''(x) = 2 (1-\beta^2)/(\beta + x)^3$. Hence, we obtain $g_{\beta}'(1) = h_{\beta}'(h_{\beta}(1)) h_{\beta}'(1) = (\beta-1)^2 / (\beta+1)^2$. Since $0 < \lvert \beta -1 \rvert / \lvert \beta  + 1 \rvert < 1$, we have $0 < \lvert g_\beta'(1) \rvert < 1$, so $g_\beta'(1) \not \in \{0, \infty\}$.  Moreover, from $z = f'(1)/d = g_{\beta}'(1)$, we obtain $0 < \lvert z \rvert < 1$. Note that $(\beta-1)^2/ (\beta+1)^2 \in \mathbb{R}$ if and only if $(\beta-1) / (\beta+1) \in \mathbb{R} \cup \mathbb{R}i$. Also note that $(\beta-1)/ (\beta+1) = 1 - 2 / (\beta +1)$, so $(\beta-1) / (\beta+1) \in \mathbb{R}$ if and only if $\beta  \in \mathbb{R}$.   If $(\beta-1) / (\beta+1) = ci$ for some $c \in (-1,1)$, then we obtain 
   \begin{equation*}
     \beta =\frac{1+ci}{1-ci} = \frac{1-c^2}{1+c^2} + \frac{2 c}{1 + c^2} i,
   \end{equation*}
so $\vert \beta \rvert^2 = 1$. Since $\beta \not \in \mathbb{R}$ and $\lvert \beta \rvert \ne 1$ by hypothesis, we find that $z = g_{\beta}'(1) = (\beta-1)^2 / (\beta+1)^2\not \in \mathbb{R}$ as we wanted. Finally, let us determine $g_{\beta}''(1)$. We have 
$g''_{\beta}(1) = -4{(\beta-1)}^2b/{\beta+1)}^4 \not \in \{0, \infty\}$, where we used that $\beta \not \in \{1, -1\}$. This finishes the proof.
\end{proof}

\begin{remark} \label{rem:Ef}
  The map $f_\beta(z) = g_\beta(z^d)$ does not have exceptional points. To see this, we apply Lemma~\ref{lem:Ef}. First, let us determine the  points of $f_\beta$ with multiplier $0$. We have
  \begin{equation*}
      f_\beta(z) = \frac{(\beta^2 + 1)z^d  + 2\beta}{\beta^2 + 1 + 2 \beta z^d} \qquad \text{and} \qquad f_\beta'(z) = d z^{(d-1)} \frac{(\beta^2 - 1)^2}{(1 + 2  \beta z^d + \beta^2)^2},
  \end{equation*}
  so the only point with multiplier $0$ is $z= 0$. However, $0$ is not a fixed point of $f_\beta$ because $f_\beta(0) = 2 \beta / (1 + \beta^2)$, so $f_\beta$ does not have any exceptional points.
\end{remark}

Now we combine all the results obtained so far in this section obtaining Corollaries~\ref{cor:ising:close-1:1} and~\ref{cor:ising:implementations:1}.

\begin{corollary}  \label{cor:ising:close-1:1}
	Let $\Delta$ be an integer with $\Delta \ge 3$ and let $\beta \in \algcomplex \setminus \mathbb{R}$ with $\lvert \beta \rvert \ne 1$ and  $0 < \lvert \beta - 1\rvert / \lvert \beta + 1 \rvert < 1$. There is a rational number $r \in (0,1)$ and  a polynomial-time algorithm such that, on input $\lambda \in B(1, r) \cap  \algcomplex$  and rational $\epsilon > 0$, computes a graph $G$ that  $(\Delta, \beta)$-implements a complex number $\hat{\lambda}$ with $\lvert \lambda - \hat{\lambda} \rvert \le \epsilon$.
\end{corollary}
\begin{proof}
	Set $d := \Delta -1$. Lemma~\ref{lem:ising-program:1} and Lemma~\ref{lem:ising-program:conditions} provide us with the three conditions that the fixed point $\omega = 1$ of $f_\beta(x) = g_\beta(x^d)$ has to satisfy to apply Lemma~\ref{lem:efficient-covering}. We find that $1$ is densely program-approximable in polynomial time for $g_\beta$, $d$ and $a_0 = \beta$. In terms of Ising-programs, this gives (Definition~\ref{def:dpa} with $k = 1$) that there is $r > 0$ and an algorithm, on inputs a rational $\epsilon > 0$ and  $\lambda \in B(1, r) \cap \algcomplex$, that computes, in polynomial time in $\mathrm{size}(\epsilon)$ and $\mathrm{size}(\lambda)$ a complex number $\hat{\lambda}$ generated by an Ising-program with $\lvert \lambda - \hat{\lambda} \rvert \le \epsilon$.  This can be then translated to the result given in the statement by applying Proposition~\ref{prop:ising-program}.
\end{proof}

\begin{corollary}  \label{cor:ising:implementations:1}
	Let $\Delta$ be an integer with $\Delta \ge 3$ and let $\beta \in \algcomplex \setminus \mathbb{R}$ with $\lvert \beta \rvert \ne 1$ and  $1/\sqrt{\Delta -1} < \lvert \beta - 1\rvert / \lvert \beta + 1 \rvert < 1$. Then the pair $(\Delta, \beta)$ implements the complex plane in polynomial time for the Ising model.
\end{corollary}
\begin{proof}
  Set $d := \Delta -1$. The proof starts the same way as the proof of Corollary~\ref{cor:ising:close-1:1}. The difference here is that once we show that $1$ is densely program-approximable in polynomial time for $g_\beta$, $d$ and $a_0 = \beta$, we use this property to apply Lemma~\ref{lem:implementations}. First, we have two check the other two hypothesis of Lemma~\ref{lem:implementations}. The first hypothesis that $1$ is a repelling fixed point of $f_\beta$ or, equivalently, $\lvert f_\beta'(1) \rvert > 1$. This follows from  $1/\sqrt{\Delta -1} < \lvert \beta - 1\rvert / \lvert \beta + 1 \rvert$ since $f_\beta'(1) =  d (\beta - 1 )^2 / ( \beta + 1 )^2$. The second hypothesis is that $0$ and $\infty$ are not exceptional points of the rational map $f_\beta$, which holds because $f_\beta$ does not have exceptional points, see Remark~\ref{rem:Ef}.  We conclude by Lemma~\ref{lem:implementations} that there is a polynomial-time algorithm such that, on input $\lambda \in  \algcomplex$  and rational $\epsilon > 0$, computes an element $a_k$ of an Ising-program with $\lvert \lambda - a_k \rvert \le \epsilon$. The result  now follows by applying the algorithm of Proposition~\ref{prop:ising-program} to translate the obtained Ising-program to a graph that $(\Delta, \beta)$-implements $a_k$. 
\end{proof}

Finally, we extend Corollaries~\ref{cor:ising:close-1:1} and~\ref{cor:ising:implementations:1} to the rest of the complex plane when possible.

\begin{lemma}  \label{lem:ising:close-1:2}
	Let $\Delta$ be an integer with $\Delta \ge 3$ and let $\beta \in \algcomplex \setminus \mathbb{R}$ with $\beta \not \in \{i, -i\}$. There is a rational number $r \in (0,1)$ and  a polynomial-time algorithm such that, on input $\lambda \in B(1, r) \cap   \algcomplex$  and rational $\epsilon > 0$, computes a graph $G$ that  $(\Delta, \beta)$-implements a complex number $\hat{\lambda}$ with $\lvert \lambda - \hat{\lambda} \rvert \le \epsilon$.
\end{lemma}
\begin{proof}
	We recall that $\lvert \beta - 1\rvert / \lvert \beta + 1 \rvert < 1$ if and only if $\mathrm{Re}(\beta) > 0$, and $\lvert \beta - 1\rvert / \lvert \beta + 1 \rvert = 1$ if and only if $\mathrm{Re}(\beta) = 0$, see Proposition~\ref{prop:R}, Item~\ref{item:R:3}. We distinguish three cases based on this observation: 
	
	{\bf Case 1:} $\lvert \beta \rvert \ne 1$ and  $0 < \lvert \beta - 1\rvert / \lvert \beta + 1 \rvert < 1$. This case is exactly Corollary~\ref{cor:ising:close-1:1}.
	
{\bf Case 2:} $\lvert \beta \rvert = 1$ and $0 < \lvert \beta - 1\rvert / \lvert \beta + 1 \rvert < 1$. We have $\mathrm{Re}(\beta) > 0$. We consider the edge interaction $\beta' = h_\beta(\beta)$ that is implemented by a path of length two with weights $\beta$. We have $\beta' = (\beta^2 + 1)/ (2 \beta) = (\beta + \beta^{-1})/2 = \mathrm{Re}(\beta) \in (0, 1)$, where we used that $\lvert \beta \rvert = 1$ and $\mathrm{Re}(\beta) > 0$. We now consider the Mobius map $h_{\beta'}(x) = (\beta' x + 1)/(\beta' + x) = (\beta' x + 1)/(\overline{\beta'} + x)$, where we used that $\beta'$ is real. It is well known that this Mobius map fixes $\{x \in \mathbb{C} : \lvert x \rvert = 1\}$ (Proposition~\ref{prop:mobius:disk}). Moreover, $h_{\beta'}(0) = 1/ \mathrm{Re}(\beta) \in (1, \infty)$. Hence, the Mobius map $h_{\beta'}$ sends the open unit disk $\mathbb{D} := B(0,1)$ to $\widehat{\mathbb{C}}\setminus \mathbb{D}$ and it sends $\widehat{\mathbb{C}}\setminus \mathbb{D}$ to $\mathbb{D}$. We conclude that $g_{\beta'}(\mathbb{D}) = \mathbb{D}$. Therefore, $\beta \cdot \beta' \in \mathbb{D}$ and $\gamma := g_{\beta'} (\beta \cdot \beta') \in \mathbb{D}$. We can implement the edge interaction $\gamma$ using the graph given in Figure~\ref{fig:graph-gamma}.  
We have
	\begin{equation*}
		h_{\beta'}(x) = \frac{\beta'^2 x + \beta' + \beta' \lvert x \rvert^2 + \overline{x}}{\lvert \beta' + x \rvert^2},
	\end{equation*}
 	so, since $\beta' \in (0, 1)$, the Mobius map $h_\beta'(x)$ sends points with positive real part to points with positive real part, and non-real points to non-real points. Hence, the Mobius map $g_{\beta'}(x) = h_{\beta'}(h_{\beta'}(x))$ also has these properties. We conclude that $\gamma = g_{\beta'} (\beta \cdot \beta')$ has positive real part and is not real. Putting all this together, $\gamma$ is a non-real number with $\lvert \gamma \rvert < 1$ and  $0 < \lvert \gamma - 1\rvert / \lvert \gamma + 1 \rvert < 1$, so $\gamma$ is in the first case of this proof. We can translate the algorithm of the first case of this proof for $\gamma$ to an algorithm for $\beta$ because we can $(\Delta, \beta)$-implement $\gamma$, see Section~\ref{sec:pre:implementations} for the transitivity property of implementations.  
	
	\begin{figure}[H]
		\centering
			\begin{tikzpicture}
			\node [circle, draw, inner ysep=1mm] (S0) {$s$};
			\node [circle, draw, inner ysep=1mm, right = 1cm of S0] (S1) {};
			\node [circle, draw, inner ysep=1mm, right = 1cm of S1] (S2) {};
			\node [circle, draw, inner ysep=1mm, right = 1.75cm of S2] (T2) {};
			\node [circle, draw, inner ysep=1mm, above right = 0.9cm and 0.7cm of S2] (A) {};
			\node [circle, draw, inner ysep=1mm, right = 1cm of T2] (T1) {};
			\node [circle, draw, inner ysep=1mm, right = 1cm of T1] (T0) {$t$};
			\path[-] 
			(S0) edge node [above] {} (S1)
			(S1) edge node [above] {} (S2)
			(T1) edge node [above] {} (T2)
			(T0) edge node [above] {} (T1)
			(S2) edge[bend left=30] node [above] {} (A)
			(A) edge[bend left=30] node [above] {} (T2)
			(S2) edge[bend right=60] node [below] {} (T2);
		\end{tikzpicture}
		\caption{A graph that $(3, \beta)$-implements $\gamma$.} 
		\label{fig:graph-gamma}
	\end{figure}
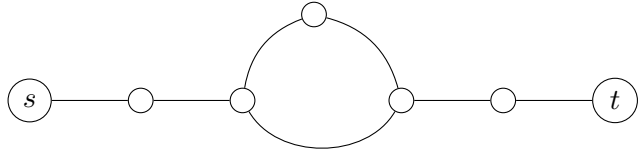

{\bf Case 3:} $\lvert \beta - 1\rvert / \lvert \beta + 1 \rvert \geq 1$. We can use the Ising program of Lemma~\ref{lem:ising-program:1} to generate $a_k \in \mathbb{C} \setminus \mathbb{R}$ with $\lvert 1 - a_k \rvert < 1/2$. We can $(\Delta, \beta)$-implement $a_k$ with the help of Proposition~\ref{prop:ising-program}. Note that  $0 < \lvert a_k - 1\rvert / \lvert a_k + 1 \rvert < 1$, so the edge interaction $a_k$ is in one of the first two cases of the proof. Again from the transitivity property of implementations, we can translate the algorithm of the first two cases for $a_k$ to an algorithm for $\beta$, concluding the proof.
\end{proof}

\begin{lemimplementations}
\statelemimplementations
\end{lemimplementations}
\begin{proof}
Set $d := \Delta - 1$. Let $r$ be the positive real number given in Lemma~\ref{lem:ising:close-1:2}. Let $f_\beta(x) = g_\beta(x^d)$. As argued in Corollary~\ref{cor:ising:implementations:1}, $1$ is a repelling fixed point of $f_\beta$ so, by Lemma~\ref{lem:repelling}, $1$ belongs to the Julia set of $f$ and, thus, by Theorem~\ref{thm:cd}, $\bigcup_{n = 0}^{\infty} f_\beta^n(B(1, r)) = \widehat{\mathbb{C}} \setminus E_{f_\beta}$, where $E_{f_\beta}$ is the set of exceptional points of $f_\beta$. In view of Remark~\ref{rem:Ef}, $E_{f_\beta}$ is empty. Let $\gamma = 10(1 + i)$. There is a positive integer $N$ such that $\gamma \in f_\beta^N(B(1, r))$. 
Thus, there is $x^* \in B(1, r) \cap  \algcomplex$ such that $f^N(x^*) = \gamma$. By the continuity of the rational function $f^N$ at $x^*$, there is $\delta \in (0,r)$ such that $\lvert f^N(x^*) - f^N(x) \rvert \le 0.01$ for every $x \in B(x^*, \delta)$. The constants $\gamma$ and $0.01$ are not chosen to be optimal but to make the notation and proof simpler for the reader. In view of Lemma~\ref{lem:ising:close-1:2} we can compute (in constant time) a graph $G$ that $(\Delta, \beta)$-implements a complex number $\hat{x}$ with $\lvert x^* - \hat{x} \rvert \le \delta$. Let $\hat{\gamma} = f_\beta^N(\hat{x})$. Note that we can $(\Delta, \beta)$ implement $\hat{\gamma}$ using the construction from Proposition~\ref{prop:ising-program}  and the fact that we can $(\Delta, \beta)$ implement $\hat{x}$. By continuity, we have $\lvert \gamma - \hat{\gamma} \rvert \le 0.01$. Note that $\hat{\gamma}$ is in $\algcomplex$. Moreover, we have $\mathrm{Re}(\hat{\gamma}) > 0$, so $\lvert \hat{\gamma} - 1\rvert / \lvert \hat{\gamma} + 1 \rvert < 1$. From  $\lvert \gamma - \hat{\gamma} \rvert \le 0.01$ and the triangle inequality we have
\begin{equation*}
   \left| \frac{\hat{\gamma}  - 1}{\hat{\gamma} + 1} \right| = \left| \frac{(\hat{\gamma}-\gamma) + \gamma  - 1}{(\hat{\gamma}-\gamma) + \gamma + 1} \right| \ge \frac{\left| \gamma  - 1\right| - 0.01}{\lvert \gamma +1\rvert + 0.01} > 1 / \sqrt{2}.
\end{equation*}
Hence, the edge interaction $\hat{\gamma}$ is in the region covered by Corollary~\ref{cor:ising:implementations:1}, 
so the pair $(\Delta, \hat{\gamma})$ implements the complex plane in polynomial time for the Ising model. We conclude that the pair $(\Delta, \beta)$ implements the complex plane in polynomial time for the Ising model thanks to the transitivity property of implementations.
\end{proof}

	The complex dynamic argument presented in the proof of Lemma~\ref{lem:ising:implementations} is one of the main ideas behind the results of \cite{Bezb} and is applied twice in  Appendix~\ref{sec:A}. The proof of Lemma~\ref{lem:ising:implementations} is simpler than the ones presented in Appendix~\ref{sec:A} because here we are only trying to approximate $\gamma$ instead of approximating any number in a neighbourhood of $\gamma$. This allows us to use the continuity of $f^N$ at $x^*$ instead of having to use Lipschitz properties of $f^N$ and careful approximations of the quantities involved. 

\begin{remark}
	Lemma~\ref{lem:ising:implementations} can be extended to other points with  $1/\sqrt{\Delta -1} > \lvert \beta - 1\rvert / \lvert \beta + 1 \rvert$. However, we have not found a systematic way to do this. Rather we are aware of points $\beta$ with $1/\sqrt{\Delta -1} > \lvert \beta - 1\rvert / \lvert \beta + 1 \rvert > 1/ (\Delta -1)$ that can be used to $(\Delta, \beta)$-implement edges interactions that are covered by Lemma~\ref{lem:ising:implementations}. For example, this is the case of those points $\beta$ such that there is a ``nice'' graph $G$ with $\ising(G; \beta) = 0$. This is made precise in Section~\ref{sec:zeros}.
\end{remark}

\subsection{Reducing exact computation to approximate computation} \label{sec:hardness:reduction}

In this section we use our implementation results to prove the hardness of approximating the partition function of the Ising model on bounded degree graphs. 
A basic building block for the reduction is the 
binary search (interval-shrinking) technique developed by Goldberg and Jerrum  in the context of the Tutte polynomial \cite{Goldberg2014}.
Since the partition function of the Ising model is a special case of the Tutte polynomial,
this building block is also applicable here.
The interval-stretching  technique requires us to be able to implement the real line in polynomial time, and this is the motivation behind  the results of Section~\ref{sec:hardness:implementations}. 

We need to use a version of the interval-shrinking technique that applies in the
context of non-real edge interactions.  
This has been developed previously~\cite{Galanis2020, Goldberg2017}.
The reduction of \cite{Galanis2020} is particularly relevant
for us because the starting point for their hardness result
is the problem of exactly evaluating the  
the Tutte polynomial,
and crucially this remains $\numP$-hard even in the $q=2$ case (corresponding to the
Ising model) and even when the input is restricted to be a  $3$-regular graph \cite{Kowa2016} (which we require here). In the rest of this section, we briefly explain the reduction given in \cite{Galanis2020} and show how it applies to bounded degree graphs (which were not considered in the previous work). 

First, let us introduce some notation. For numbers $q$ and $\gamma$,  the Tutte polynomial of a graph $G = (V, E)$ is given by
\begin{equation} \label{eq:tutte}
	Z_{\mathrm{Tutte}}(G; q, \gamma) = \sum_{A \subseteq E} q^{k(A)} \gamma^{|A|}, 
\end{equation}
where $k(A)$ denotes the number of connected components in the graph $(V, A)$ (isolated vertices do count).  We have $\ising(G; \beta)  = Z_{\mathrm{Tutte}}(G; 2, \beta - 1)$, see, for instance,~\cite{Sokal2005}. Let $s$ and $t$ be two distinct vertices of G. We define
\begin{equation*}
	Z_{st}(G; q, \gamma) = \sum_{\substack{A \subseteq E: \\ s \text{ and } t \text{ in the same component}}} q^{k(A)} \gamma^{\lvert A \rvert}
\end{equation*}
Analogously, let $Z_{s\vert t}$ be the contribution to $Z_{\text{Tutte}}(G; q, \gamma)$ from the configurations $A \subseteq E$ such that $s$ and $t$ are in different connected components in $(V, A)$. That is, $Z_{s \vert t}(G; q, \gamma) = Z_{\text{Tutte}}(G; q, \gamma) - Z_{st}(G; q, \gamma)$. We define the following computational problems for any rational numbers $q > 0$, $\gamma > 0$, any integer $\Delta \ge 3$ and any $\beta \in \algcomplex$.

\prob{$\I(\Delta, \beta)$.}{A graph $G = (V, E)$ with maximum degree at most $\Delta$.}{The number $\ising(G; \beta ) \in \algcomplex$.}

\prob{$\RT(\Delta, q, \gamma)$.}{A graph $G = (V, E)$ with maximum degree at most $\Delta$ and an edge $(s,t)$ of $G$.}{The rational number $Z_{s|t}(G; q, \gamma ) / Z_{st}(G; q, \gamma )$.}

In \cite{Galanis2020} the authors consider $\RT(\Delta, q, \gamma)$ without restrictions on the maximum degree of the input graph. They only ask for the vertices $s$ and $t$ to be in the same connected component of $G$. Moreover, they let $q$ and $\gamma$ be any non-zero algebraic numbers (possibly non-real or negative real), so they have to study carefully the possibility that $Z_{st}(G; q, \gamma ) = 0$. Their results can be applied to our simplified version of $\RT(\Delta, q, \gamma)$ to obtain the reductions given in Lemmas~\ref{lem:compute-fraction:q>1} and~\ref{lem:reduce-frac-to-exact}.

\begin{lemma}[{Bounded degree version of \cite[Lemmas 42 and 43]{Galanis2020}} for the Ising model] \label{lem:compute-fraction:q>1}
	Let $K$ be a real number with $K > 1$. Let $\Delta \ge 3$ be an integer and let $\beta \in \algcomplex$ such that $(\Delta, \beta)$ implements the real line in polynomial time. Let $y \in \mathbb{C}$ with $y > 1$.  Then we have the reductions
	\begin{align*}
		\RT(\Delta, 2, y-1) & \le_T \IN(\Delta, \beta, K), \\
		\RT(\Delta, 2, y-1) & \le_T \IA(\Delta, \beta, \pi/3). 
	\end{align*}
\end{lemma}
\begin{proof}
	The proof is the almost the same as that of  \cite[Lemmas 42 and 43]{Galanis2020}. Here we indicate how we adapt the reduction of \cite[Lemmas 42 and 43]{Galanis2020} to graphs with maximum degree $\Delta$. First, let us translate our Ising notation to the notation used in the proof of  \cite[Lemmas 42 and 43]{Galanis2020}. In the original proof we have two weights $\gamma_1 \in (-2, -1)$ and $\gamma_2 > 0$ and access to an oracle that approximates the norm or determines the sign of the multivariate Tutte polynomial on weighted graphs with weights in $\{\gamma_1, \gamma_2\}$. Note that determining the sign reduces to additively approximating the argument of this polynomial with error at most $\pi/3$, so we can use our oracle $\IA(\Delta, \beta, \pi/3)$ instead. The purpose of the weights $\gamma_1$ and $\gamma_2$ is implementing the real line in polynomial time for the Tutte polynomial (\cite[Corollary 12]{Galanis2020}). Here the role of these weights is performed by $\beta$. Hence, every time \cite[Corollary 12]{Galanis2020} is used in the proof of \cite[Lemmas 42 and 43]{Galanis2020} we use the fact that $(\Delta, \beta)$ implements the real line in polynomial time instead. The reduction of \cite[Lemmas 42 and 43]{Galanis2020} computes the ratios $Z_{s|t}(H; q, \gamma ) / Z_{st}(H; q, \gamma)$ for some positive number $\gamma$ that can be implemented using $\gamma_1$ and $\gamma_2$. Here we set $\gamma = y-1$. The only relevant properties of $\gamma$ in the proof of \cite[Lemmas 42 and 43]{Galanis2020} are $\gamma > 0$ and the fact that $\gamma$ can be implemented exactly.
	
	There are two differences between this proof and the proof of \cite[Lemmas 42 and 43]{Galanis2020}. Let $H$ and $(s,t)$ be the inputs of $ \RT(\Delta, 2, y-1)$. The first difference in the proof is that we restrict ourselves to computing ratios $Z_{s|t}(H; q, \gamma ) / Z_{st}(H; q, \gamma )$ where $(s,t)$ is an edge of $H$. This is so that all the graphs considered in the reduction have maximum degree at most $\Delta$. The original proof applies one of the oracles $\IN(\Delta, \beta, K)$ and $\IA(\Delta, \beta, \pi/3)$ to a copy of $H$ with an extra edge joining $s$ and $t$. This extra edge has a weight $\gamma'$ that is updated repeatedly during the binary search. The weight $\gamma'$ is implemented using \cite[Corollary 12]{Galanis2020} or, in our case, using the fact that $(\Delta, \beta)$ implements the real line in polynomial time. Instead of adding an extra edge between $s$ and $t$, here we modify the edge $(s,t)$ so that its weight is $\gamma \cdot \gamma'$, producing the same effect as adding an extra edge from $s$ to $t$ with weight $\gamma'$. This time we have to implement $\gamma \cdot \gamma'$ instead. Let $H'$ be the graph obtained by copying $H$ and substituting the edge $(s,t)$ with an appropriate graph that $(\Delta, \beta)$-implements $\gamma \cdot \gamma'$. Then the graph $H'$ also has maximum degree at most $\Delta$. Moreover, for $\varepsilon = \gamma'+1$ we have
	\begin{equation} \label{eq:reduction:1}
		\begin{aligned}
			Z_{\text{Tutte}}(H'; q, \gamma) & = Z_{st}(H; q, \gamma)(1+ \gamma') + Z_{s|t}(H; q, \gamma) \left( 1 + \frac{\gamma'}{q} \right) \\
			& = Z_{s|t}(H; q, \gamma)\left( 1 - \frac{1}{q} \right) + \varepsilon \left( Z_{st}(H; q, \gamma) + \frac{1}{q} Z_{s|t}(H; q, \gamma) \right) \\ & = f(\varepsilon; H, \gamma), 
		\end{aligned}
	\end{equation}
	where $f(\varepsilon; H, \gamma)$ is the linear function to which the binary search will be performed. The purpose of the binary search is finding a zero of $f(\varepsilon; H, \gamma)$, which allows us to compute the ratio $Z_{s|t}(H; q, \gamma ) / Z_{st}(H; q, \gamma )$.
	
The second difference is that we cannot implement $\gamma$ exactly.   We can bypass this by using a very close approximation $\hat{\gamma}$ of $\gamma$ instead. We use the fact that we can $(\Delta, \beta)$-implement $\hat{\gamma}$ with $\lvert \gamma - \hat{\gamma} \rvert \le \delta$ in polynomial time in the size of $\delta$. We perform the binary search on $f(\varepsilon; H, \hat{\gamma})$ instead. This allows us to compute the rational number $Z_{s|t}(H; q, \hat{\gamma} ) / Z_{st}(H; q, \hat{\gamma} )$. We can choose $\delta$ with $\mathrm{size}(\delta) \in \mathrm{poly}(\mathrm{size}(\epsilon), \mathrm{size}(H))$ such that $\lvert Z_{s|t}(H; q, \hat{\gamma} ) - Z_{s|t}(H; q, \gamma)\rvert \leq \epsilon$ and $\lvert Z_{st}(H; q, \hat{\gamma}) - Z_{st}(H; q, {\gamma}) \rvert \le \epsilon$, see for instance \cite[Lemma 55]{Galanis2020}. Therefore, the error that we make by outputting $Z_{s|t}(H; q, \hat{\gamma} ) / Z_{st}(H; q, \hat{\gamma} )$ instead of $Z_{s|t}(H; q, \gamma ) / Z_{st}(H; q, \gamma)$ can be made to be at most $\epsilon$ by choosing $\delta$ with $\mathrm{size}(\delta) \in \mathrm{poly}(\mathrm{size}(\epsilon), \mathrm{size}(H))$ thanks to the lower  and upper bounds on $\lvert Z_{s|t}(H; q, \cdot )\rvert $ and $\lvert Z_{st}(H; q, \cdot )\rvert $, see \cite[Section 6.1]{Galanis2020} for these bounds. 
We conclude that we can compute $Z_{s|t}(H; q, \gamma ) / Z_{st}(H; q, \gamma)$ exactly. This can be done using  \cite[Lemma 38]{Galanis2020} as it is done 
for algebraic numbers 
in the proof of  \cite[Lemma 42]{Galanis2020}. 
or a simpler version of \cite[Lemma 38]{Galanis2020} for rational numbers. 
\end{proof}

\begin{lemma}[{Particular case of \cite[Lemma 49]{Galanis2020}}]\label{lem:reduce-frac-to-exact}
	Let $\Delta \ge 3$ be an integer and let $\beta \in \mathbb{Q}$ with $\beta > 0$. Then we have the reduction
	\begin{align*}
		\I(\Delta, \beta) & \le_T \RT(\Delta, \beta).
	\end{align*}
\end{lemma}
\begin{proof}
	The reduction given in the proof of  \cite[Lemma 49]{Galanis2020} applies with the change of variables $q = 2$ and $\gamma = \beta-1$. It is important to note that this reduction only invokes the oracle for $\RT(\Delta, \beta)$ with inputs $(G, s,t)$ such that $e = (s,t)$ is an edge of $G$. The reduction reduces the computation of $Z_{\text{Tutte}}(G; q, \gamma)$ to that of $Z_{\text{Tutte}}(G \setminus e; q, \gamma)$, $Z_{s|t}(G; q, \gamma ) / Z_{st}( G; q, \gamma )$ and $Z_{s|t}( G \setminus e; q, \gamma ) / Z_{st} ( G \setminus e; q, \gamma )$, where $G \setminus e$ is the graph $G$ without the edge $e$. Hence, all the calls to the oracle $\RT(\Delta, \beta)$ involve subgraphs of $G$, that have maximum degree at most $\Delta$. Finally, because $q > 0$ and $\gamma > 0$ in our setting, we do not have to consider the cases when $Z_{st} ( G; q, \gamma ) = 0$, simplifying the result.
\end{proof}

Now we have the tools to obtain the desired reductions and the proof of Theorem~\ref{thm:hardness}.

\begin{lemma}[{\cite[Lemma 50]{Galanis2020} for the Ising model}] \label{lem:exact-reduction:q>1}
	Let $K$ be a real number with $K > 1$. Let $\Delta \ge 3$ be an integer and let $\beta \in \algcomplex$ such that $(\Delta, \beta)$-implements the real line in polynomial time. Let $y \in \mathbb{C}$ with $y > 1$.  Then we have the reductions
	\begin{align*}
		\I(\Delta, y) & \le_T \IN(\Delta, \beta, K), \\
		\I(\Delta, y) & \le_T \IA(\Delta, \beta, \pi/3).
	\end{align*}
\end{lemma}
\begin{proof}
	This result follows directly from combining Lemmas~\ref{lem:compute-fraction:q>1} and~\ref{lem:reduce-frac-to-exact}. The proof of \cite[Lemma 50]{Galanis2020}  takes a bit more work because one has to be careful about possible zeros of $Z_{st} (G; q, \gamma)$.
\end{proof}

\begin{thmhardness}
	\statethmhardness
\end{thmhardness}
\begin{proof}
	Our hardness theorem now follows from combining Lemmas~\ref{lem:ising:implementations} and~\ref{lem:exact-reduction:q>1} in conjunction with the fact that $\I(3, y) $ is $\numP$-hard for any $y > 1$ \cite{Kowa2016}.
\end{proof}

Since $1 / \sqrt{\Delta - 1}$ converges to $0$ as $\Delta$ diverges, Theorem \ref{thm:hardness}  gives a new proof  of~\cite[Theorem 3]{Galanis2020} which says that $\IN(\beta, \infty, 1.01)$ and $\IA(\beta, \infty, \pi/3)$ (where there are no restrictions on the maximum degree of the input graph) are $\numP$-hard for any algebraic number $\beta \in \mathbb{C} \setminus (\mathbb{R} \cup \{i, -i\})$.

\section{Zeros of the partition function and hardness} \label{sec:zeros} 
 
 In this section we give explicit evidence that zeros of the partition function imply hardness of approximation for the Ising model when the edge interaction $\beta$ is not in $\mathbb{R} \cup \{i, -i\}$. These are the first results that explicitly link zeros to hardness of approximation that we are aware of. Our main technical result Lemma is~\ref{lem:implementing-1}, which shows that implementing $-1$ implies hardness of approximation of the partition function of the Ising model. We then use zeros of the partition function to implement $-1$ and conclude hardness in Lemma~\ref{lem:zeros-hard} and Corollary~\ref{cor:zeros-hard}. Our proofs use the hardness and implementation results of Section~\ref{sec:hardness}. Finally, in Corollary~\ref{cor:hardness-inside}, we give an example of an edge interaction $\beta$ in the region $\mathcal{R}(1/\sqrt{2})$ and a graph $G$ with maximum degree $3$ such that $\ising(G; \beta) = 0$, showing that the hardness region given in Theorem~\ref{thm:hardness} is not optimal.
 
\begin{lemimplementingminusone}
  \statelemimplementingminusone
\end{lemimplementingminusone}
\begin{proof}
	There are two cases. The first case is when $\lvert \beta - 1 \rvert / \lvert \beta + 1 \rvert  > 1 / \sqrt{\Delta - 1}$. Then, since $\beta \not \in \mathbb{R} \cup \{i, -i\}$, we know that the problems $\IN(\Delta, \beta, 1.01)$ and $\IA(\beta, \Delta, \pi/3)$ are $\numP$-hard (Theorem~\ref{thm:hardness}). In the rest of the proof we assume that $\lvert \beta - 1 \rvert / \lvert \beta + 1 \rvert  \le 1 / \sqrt{\Delta - 1}$. We are going to reduce the approximation problems at  $(\Delta, \gamma)$ to the approximation problems at $(\Delta, \beta)$ for some $\gamma$ such that $\IN(\Delta, \gamma, 1.01)$ and $\IA(\Delta, \gamma, \pi/3)$ are $\numP$-hard. In this reduction we will use the fact that we can  $(\Delta, \beta)$ implement the edge interaction $-1$. Let  $\alpha \in \mathbb{C}$ be some edge interaction that we can  $(\Delta, \beta)$-implement. We fine-tune $\alpha$ later in the proof. We consider the weighted graph $J$ given in Figure~\ref{fig:J}. By the properties of series and parallel compositions, this graph implements the edge interaction $\gamma := h_\beta(h_\beta(-\alpha))$.
	\begin{figure}[H]
		\centering
			\begin{tikzpicture}
			\node [circle, draw, inner ysep=1mm] (S1) {$s$};
			\node [circle, draw, inner ysep=1mm, right = 1cm of S1] (S2) {};
			\node [circle, draw, inner ysep=1mm, right = 1.75cm of S2] (T2) {};
			\node [circle, draw, inner ysep=1mm, right = 1cm of T2] (T1) {$t$};
			\path[-] 
			(S1) edge node [above] {$\beta$} (S2)
			(T1) edge node [above] {$\beta$} (T2)
			(S2) edge[bend left] node [above] {$-1$} (T2)
			(S2) edge[bend right] node [below] {$\alpha$} (T2);
		\end{tikzpicture}
		\caption{The graph $J$ in the proof of Lemma~\ref{lem:implementing-1}.} 
		\label{fig:J}
	\end{figure}
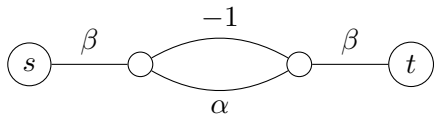
	 From Remark~\ref{rem:stretch} it follows that 
	 \begin{equation} \label{eq:implementing-1}
	 	\frac{\gamma - 1}{\gamma + 1} = \left(\frac{\beta - 1}{\beta+1}\right)^2 \frac{-\alpha - 1}{-\alpha + 1}.
	 \end{equation}
    The idea to complete this proof is  $(\Delta, \beta)$-implementing $\alpha$ so that the complex number in \eqref{eq:implementing-1} has norm larger than $1$ (hence larger than $1/\sqrt{\Delta-1}$ so Theorem~\ref{thm:hardness} applies). By Item~\ref{item:R:3} of Proposition~\ref{prop:R}, the norm is larger than~$1$ if and only if $\mathrm{Re}(\gamma) < 0$, which is what we are aiming for. We also want $\gamma$ to be non-real. Note that $\gamma$ is real if and only if $1 - 2 /(\gamma + 1) = (\gamma - 1) / (\gamma + 1)$ is real. Let $r \in (0,1)$ be the rational number in the statement of Lemma~\ref{lem:ising:close-1:2}. Let $\varepsilon = \lvert (\beta - 1) / (\beta + 1) \rvert^2 r/2$, so $\varepsilon$ is an algebraic number with $\varepsilon \in (0, 1/(\Delta - 1))$ and $\varepsilon \le r/4$. Let $\xi \in\algcomplex$ such that
    \begin{equation*}
    	\frac{\xi - 1}{\xi + 1} = \frac{r(\beta-1)^2}{4(\beta+1)^2} i.
    \end{equation*}
    We have $\lvert \xi - 1 \rvert / \lvert \xi + 1 \rvert = \varepsilon/2$, so $\xi \in \mathcal{R}(\varepsilon / 2)$. From Proposition~\ref{prop:R}, we obtain $\mathrm{Re}(\xi) \geq 0$ and
    \begin{equation*}
    	\lvert \xi \rvert \le  \frac{1 + \varepsilon / 2}{1 - \varepsilon / 2} \le \frac{1 + 1/4 }{1 - 1/4}= 5/3,
    \end{equation*}
  where we used that $0 < \varepsilon \le 1/(\Delta - 1) \le 1/2$. Thus, we have  
  \begin{equation*}
    \lvert \xi - 1 \rvert = \lvert \xi + 1 \rvert \frac{\varepsilon}{2}  \le (5/3 + 1) \frac{\varepsilon}{2} = \frac{4}{3} \varepsilon \le \frac{1}{3} r.
  \end{equation*}
    Therefore, we can use Lemma~\ref{lem:ising:close-1:2} to $(\Delta, \beta)$-implement $\alpha \in \algcomplex$ with $\lvert \xi - \alpha \rvert < \varepsilon / 8$. We have 
     \begin{equation*}
     	\left|\alpha - 1\right| \le \left| \alpha - \xi \right| + 	\left|\xi - 1\right| < (1/8 + 4/3) \varepsilon < 2\varepsilon < 1.
     \end{equation*}
 Hence, $\mathrm{Re}(\alpha) > 0$ and we find that
    \begin{equation*}
    	\left| \frac{\xi - 1}{\xi + 1} -  \frac{\alpha - 1}{\alpha + 1} \right| = 2 \left| \frac{\xi - \alpha}{(\xi + 1 ) (\alpha + 1)}  \right| < 2 \lvert \xi - \alpha \rvert < \varepsilon/4.
    \end{equation*}
  Let $a = (\alpha - 1) / (\alpha + 1)$, $b =  (\beta -1)^2/(\beta+1)^2 r/2$ and $z = (\xi - 1)/(\xi + 1) = i b /2$. The situation is plotted in Figure~\ref{fig:implementing-1}.
  
  \begin{figure}[H]
  	\centering
  	  \begin{tikzpicture}
  		\begin{axis}[
  			ticks=none,
  			xmin=-2,
  			xmax=2,
  			ymin=-2.7,
  			ymax=2.7,
  			axis equal,
  			axis lines=middle,
  			disabledatascaling]
  			
  			\addplot[draw=black, mark=o, only marks, nodes near coords=$0$,every node near coord/.style={anchor=45}] coordinates{(0,0)};
  			
  			\addplot[draw=black, mark=o, only marks, nodes near coords=$\varepsilon$,every node near coord/.style={anchor=225}] coordinates{(2.5,0)};
  			
  			\addplot[draw=black, mark=o, only marks, nodes near coords=$b$,every node near coord/.style={anchor=30}] coordinates{(1.76777,1.76777)};

  			\addplot[draw=black, mark=o, only marks, nodes near coords=$z$,every node near coord/.style={anchor=30}] coordinates{(-1.76777/2,1.76777/2)} ;

  			\addplot[draw=black, mark=o, only marks, nodes near coords=$a$,every node near coord/.style={anchor=30}] coordinates{(-1.76777/2+0.1,1.76777/2 +0.5)} ;

  			\draw (0,0) circle (2.5);
  			
  			\draw (-1.76777/2,1.76777/2) circle (2.5/4);
  			
  		\end{axis}
  	\end{tikzpicture} 
  \caption{The quantities $a, b, z$ in the proof of Lemma~\ref{lem:implementing-1}. We have $\lvert b\rvert = \varepsilon$, $z  = i b /2 $ and $\lvert a - z \rvert < \varepsilon / 4$.}
  \label{fig:implementing-1}
  \end{figure}
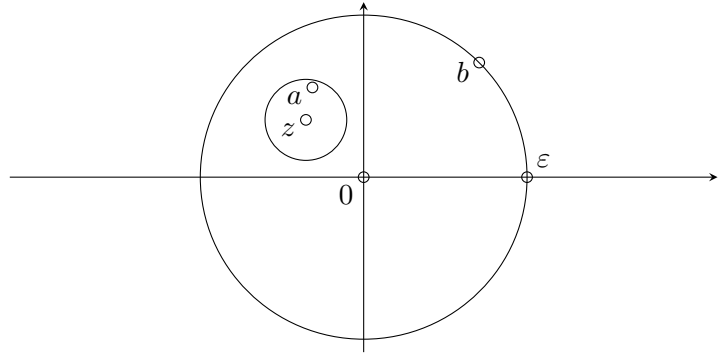

 Let $\overline{x y} = \{\lambda (y-x) : \lambda \in \mathbb{R}\}$ for any $x, y \in \mathbb{C}$. Note that $\overline{0z}$ and $\overline{0b}$ are perpendicular, so $0$ is the closest point of the line $\overline{0b}$ to $z$. Since $0 \not \in B(z, \epsilon/4)$, we conclude that $\overline{0b} \cap B(z, \varepsilon/4) = \emptyset$. In particular, $a$ is not in the line $\overline{0 b}$. Also note that $\lvert a \rvert < \varepsilon$ by the triangle inequality. Putting all this together with equation \eqref{eq:implementing-1}, we find that
 \begin{equation*}
 	\frac{\gamma - 1}{\gamma + 1} = - \frac{2}{r} b a^{-1} \not\in \mathbb{R} \qquad \text{and} \qquad \left| 	\frac{\gamma - 1}{\gamma + 1}  \right| = \frac{2}{r}  \frac{\varepsilon}{\lvert a \rvert} > 1
 \end{equation*}    
    as we wanted. We have shown how to $(\Delta, \beta)$-implement $\gamma \in \algcomplex$ with $\mathrm{Re}(\gamma) < 0$ and $\gamma \not \in \mathbb{R}$. In particular, we have $\lvert \gamma - 1 \rvert / \lvert \gamma + 1 \rvert  > 1 / \sqrt{\Delta - 1}$. As a consequence of Theorem~\ref{thm:hardness}, the problems $\IN(\Delta, \gamma, 1.01)$ and $\IA(\Delta, \gamma, \pi/3)$ are $\numP$-hard. These problems reduce to $\IN(\Delta, \beta, 1.01)$ and $\IA(\Delta, \beta, \pi/3)$ because we can $(\Delta, \beta)$-implement $\gamma$, and the result follows.
\end{proof}

The rest of this section exploits Lemma~\ref{lem:implementing-1} to obtain hardness for zeros of the partition function. Our approach uses a zero to implement $-1$ and conclude hardness with the help of Lemma~\ref{lem:implementing-1}.

\begin{lemma} \label{lem:zeros-hard}
  Let $\Delta$ be an integer with $\Delta \ge 3$. Let $\beta \in \algcomplex \setminus (\mathbb{R} \cup \{i, -i\})$.  Suppose that there is a graph with maximum degree at most $\Delta$ having terminals $s,t$ such that
\begin{enumerate}
\item the degree of $s$ and $t$ is at most $\Delta-1$;
\item $\ising(G; \beta) = 0$;
\item $Z^{ij}_{st}(G; \beta) \ne 0$ for some $i, j \in \{0,1\}$.
\end{enumerate}
Then $\IN(\Delta, \beta, 1.01)$ and $\IA( \Delta, \beta, \pi/3)$ are $\numP$-hard.
\end{lemma}

\begin{proof}
  By symmetry of the spins $0$ and $1$ in the definition of $\ising$, for any vertex $v$ of $G$ we have $Z_{v}^0(G; \beta) = Z_{v}^1(G; \beta)$. Let $i,j \in \{0,1\}$ as in the statement. We obtain $0 = \ising(G; \beta) = 2 Z_{s}^i(G; \beta)$ so
  \begin{equation} \label{eq:sum-0}
      0 = Z_{s t}^{i0}(G; \beta) + Z_{s t}^{i1}(G; \beta).
  \end{equation}
  Since either $Z_{s t}^{i0}(G; \beta)$ or $Z_{s t}^{i1}(G; \beta)$ is non-zero by hypothesis, both quantities are non-zero. Again, by symmetry of the spins $0$ and $1$, we have $Z_{s t}^{00}(G; \beta) = Z_{s t}^{11}(G; \beta)$ and  $Z_{s t}^{01}(G; \beta) = Z_{s t}^{10}(G; \beta)$. Thus, by dividing by $Z_{s t}^{01}(G; \beta)$ in \eqref{eq:sum-0} we find that
  \begin{equation*}
      -1 = \frac{ Z_{s t}^{11}(G; \beta)}{Z_{s t}^{01}(G; \beta)}.
  \end{equation*}
  We have shown that the graph $G$ $\beta$-implements $-1$. Consider the graph $H$ that is a copy of $G$ with two extra vertices, $s'$ and $t'$, and two extra edges, $(s,s')$ and $(t, t')$. By the properties of series compositions, see \eqref{eq:hbeta}, the graph $H$ $\beta$-implements $h_\beta(h_\beta(-1)) = -1$ for the terminals $s'$ and $t'$ (both of which have degree $1$). Moreover, $H$ has maximum degree at most $\Delta$ because $G$ has maximum degree at most $\Delta$ and the vertices $s$ and $t$ have at most $\Delta-1$ neighbours in $G$. We conclude that $H$ $(\Delta, \beta)$-implements $-1$, and hardness follows from Lemma~\ref{lem:implementing-1}.
\end{proof}

\begin{corollary} \label{cor:hardness-inside}
    Let $\Delta = 3$. There is a  $\beta \in \algcomplex \setminus (\mathbb{R} \cup \{i, -i\})$ with  $\lvert \beta - 1\rvert / \lvert \beta + 1 \rvert < 1/\sqrt{\Delta -1}$ such that $\IN(\Delta, \beta, 1.01)$ and $\IA( \Delta, \beta, \pi/3)$ are $\numP$-hard.
\end{corollary}

\begin{proof}
 Let us consider the graph $G$ given in Figure~\ref{fig:graph-zero} with distinguished vertices $s$ and $t$. 
 	\begin{figure}[H]
		\centering
			\begin{tikzpicture}
			\node [circle, draw, inner ysep=1mm] (S) {$s$};
			\node [circle, draw, inner ysep=1mm, right = 3cm of S] (T) {$t$};

			\node [circle, draw, inner ysep=1mm, above right = 0.5cm and 1cm of S] (S1) {};
			\node [circle, draw, inner ysep=1mm, below right = 0.5cm and 1cm of S] (S2) {};

			\node [circle, draw, inner ysep=1mm, above right = 0.5cm and 2cm of S] (T1) {};
			\node [circle, draw, inner ysep=1mm, below right = 0.5cm and 2cm of S] (T2) {};

			\path[-] 
			(S) edge node [above] {} (S1)
			(S) edge node [above] {} (S2)

			(S1) edge[bend left=50] node [above] {} (T1)
			(S1) edge[bend right=50] node [above] {} (T1)

			(S1) edge[bend left=50] node [above] {} (T1)
			(S1) edge[bend right=50] node [above] {} (T1)

			(S2) edge[bend left=50] node [above] {} (T2)
			(S2) edge[bend right=50] node [above] {} (T2)

			(T) edge node [above] {} (T1)
			(T) edge node [above] {} (T2);
		\end{tikzpicture}
		\caption{A graph that $G$ with maximum degree $3$ such that $\ising(G; x)$ has a zero $\beta \in \mathcal{R}(1/\sqrt{2})$.} 
		\label{fig:graph-zero}
	\end{figure}
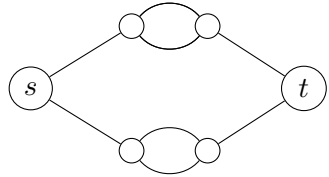
One can check that $Z^{01}_{st}(G; x) = (1 + x^2 + 2 x^3)^2$ and $Z^{11}_{st}(G; x) = x^2 (2 + x + x^3)^2$.  
We have
$\ising(G; x) = 2(1 + 6 x^2 + 8 x^3 + 2 x^4 + 8 x^5 + 6 x^6 + x^8)$. 
Using \texttt{Mathematica} we have determined that $\ising(G; x)$ has a zero at $\beta \approx 0.396608 + 0.917988 i$. Moreover, we have $\lvert Z^{01}_{st}(G; \beta) \rvert  > 2$, so $\beta$ and $G$ satisfy the hypothesis of Lemma~\ref{lem:zeros-hard}. We conclude that  $\IN(\Delta, \beta, 1.01)$ and $\IA( \Delta, \beta, \pi/3)$ are $\numP$-hard. Finally, we have $\lvert \beta - 1 \rvert / \lvert \beta + 1 \rvert  < \sqrt{2}$ since $\lvert \beta - 1 \rvert / \lvert \beta + 1 \rvert  \approx 0.6572981$.
\end{proof}

We point out that one can use the approach that Buys developed for the independent set polynomial to find more zeroes inside the region $\lvert \beta - 1\rvert / \lvert \beta + 1 \rvert \le 1/\sqrt{\Delta -1}$ \cite{Buys2019}.

Let $\beta \in \mathbb{C} \setminus (\mathbb{R}\cup \{i, -i\})$. Lemma~\ref{lem:zeros-hard}  
uses the existence of a graph~$G$
with maximum degree at most~$\Delta$ 
and 
$\ising(G; \beta) = 0$
to demonstrate the hardness of
$\IN(\beta, \Delta, 1.01)$ and $\IA(\beta, \Delta, \pi/3)$.
However, Lemma~\ref{lem:zeros-hard}
relies on the additional condition 
that 
$Z^{ij}_{st}(G; \beta) \ne 0$ for 
some $i, j \in \{0,1\}$
and two terminals~$s$ and~$t$ with
degree at most~$\Delta-1$. In the following conjecture, we conjecture that these additional
conditions are not necessary.

\begin{conjecturezeros}
\stateconjecturezeros
\end{conjecturezeros}

We make some progress on this conjecture in Corollary~\ref{cor:zeros-hard} (by changing the degree constraint from~$\Delta$ to~$\Delta-1$), but the full result seems to be out of reach for our implementation techniques.

\begin{corollary} \label{cor:zeros-hard}
Let $\Delta$ be an integer with $\Delta \ge 3$ and let $\beta \in \algcomplex \setminus (\mathbb{R}\cup \{i, -i\})$. Suppose that there is a graph $G$ of maximum degree at most $\Delta-1$ with $\ising(G; \beta) = 0$. Then $\IN(\beta, \Delta, 1.01)$ and $\IA(\beta, \Delta, \pi/3)$ are $\numP$-hard.
\end{corollary}
\begin{proof}
  Let $\mathcal{F} = \{G' : G' \text{ has maximum degree at most } \Delta-1 \text{ and } Z(G', \beta) = 0\}$, which is not empty by our hypothesis. We can choose $H \in \mathcal{F}$ with the minimum possible number of edges. Let $e = (s,t)$ be an edge of $H$. Let $H\setminus e$ be the graph obtained by deleting the edge $e$ from $H$. We have
  \begin{align*}
    Z_{st}^{00}(H; \beta) & = \beta Z_{st}^{00}(H \setminus e; \beta), \\
    Z_{st}^{01}(H; \beta) & = Z_{st}^{01}(H \setminus e; \beta).
  \end{align*}
Therefore, if $Z_{st}^{00}(H; \beta) = Z_{st}^{01}(H; \beta) = 0$, then $\ising(H \setminus e; \beta) = 2 Z_{s}^{0}(H\setminus e; \beta) = 2(Z_{st}^{00}(H; \beta) + Z_{st}^{01}(H; \beta)) = 0$, which contradicts the minimality of $H$. We conclude that either $Z_{st}^{00}(H; \beta) \ne 0$ or  $Z_{st}^{01}(H; \beta) \ne 0$. Since $s$ and $t$ have degree at most $\Delta -1$, the result follows from Lemma~\ref{lem:zeros-hard}.
\end{proof}

\bibliographystyle{plainurl}
\bibliography{IsingBib}

 \begin{appendices}

 \section{Mobius-programs: proofs of Lemmas~\ref{lem:efficient-covering} and~\ref{lem:implementations}} 
 
 \label{sec:neighbourhood}\label{sec:A}

In this appendix we prove Lemmas~\ref{lem:efficient-covering} and~\ref{lem:implementations}. These lemmas generalise the results on implementations for the independent set polynomial given in \cite{Bezb} to a more general setting so that they can be applied to other spin systems, including the Ising model. Some of the definitions and results required in this appendix have been stated in Section~\ref{sec:hardness:programs}, so we ask the reader to read Section~\ref{sec:hardness:programs} before this appendix. This appendix is organised as follows. In Section~\ref{sec:neighbourhood:covering} we show how to generate approximations of any point around a program-approximable fixed point as a first step towards the proof of Lemma~\ref{lem:efficient-covering}. In Section~\ref{sec:neighbourhood:precision} we prove Lemma~\ref{lem:efficient-covering}. Finally, in Section~\ref{sec:mobius:plane} we prove Lemma~\ref{lem:implementations}.

\subsection{From program-approximable to densely program-approximable} \label{sec:neighbourhood:covering}

In this section we generalise the results in \cite[Section 7.2]{Bezb} on hardcore-programs to Mobius-programs. The main result of this section is Lemma~\ref{lem:covering}, where we show that program-approximable fixed points are densely program-approximable under some hypothesis. The main idea behind the results given in \cite[Section 7]{Bezb} is that, locally around $\omega$, hardcore-programs behave as \emph{straight-line-programs}, which are much easier to study. This property is not specific to hardcore-programs, as illustrated in Lemma~\ref{lem:g:1}.

\begin{definition}
  Let $z \in \mathbb{C}$ with $z \ne 0$. A straight-line-program with operation
  \begin{equation}
  \label{eq:slp}
  (a_1, \ldots, a_d) \mapsto z \sum_{j = 1}^d a_j
\end{equation}
is a sequence of assignments starting with $a_0 = 0, a_1 = 1$ and
  \begin{equation*}
    a_k = z \left( a_{i_{k,1}} + \cdots + a_{i_{k,d}} \right), \quad \text{for } k = 2,3,\ldots,
  \end{equation*}
  where $i_{k,1}, \ldots, i_{k,d} \in \{0, \ldots, k-1\}$. We say that the straight-line-program generates $x \in \mathbb{C}$ if there exists integer $k \ge 0$ such that $a_k = x$.
\end{definition}

\begin{lemma}[{\cite[Lemma 7.9 for Mobius-programs]{Bezb}}] \label{lem:g:1}
  Let $d$ be an integer with $d \ge 2$ and let $g$ be a Mobius map. Let  $\omega \in \mathbb{C}$ be a fixed point of $f(z) = g(z^d)$ with $\omega \ne 0$ and  $g'(\omega^d), g''(\omega^d) \ne \infty$. Set $z := g'(\omega^d)\omega^{d-1}$. There exist reals $C_0 := C_0(g, d, \omega) > 1$ and $\delta_0 := \delta_0(g, d, \omega) > 0$ such that for any $a_1, \ldots, a_d \in \mathbb{C}$ with $|a_j| \le \delta_0$ (for $j \in [d]$) we have
  \begin{equation*}
    g\left(\left(\omega + a_1\right) \cdots \left(\omega + a_d \right) \right) = \omega + z \left(\sum\nolimits_{j = 1}^d a_j\right) + \tau,
  \end{equation*} 
  where  $\tau \in \mathbb{C}$ with $|\tau| \le C_0 \max_{j \in [d]} |a_j|^2$.
\end{lemma}
\begin{proof}
  The proof is analogous to that of \cite[Lemma 7.9]{Bezb}. The only difference is the determination of the constants $C_0$ and $\delta_0$. Here these constants are obtained by a continuity argument whereas in \cite[Lemma 7.9]{Bezb} $C_0$ and $\delta_0$ are determined explicitly. We include this proof to illustrate this continuity argument. Let $b_1, \ldots, b_d \in \mathbb{C}$ with $|b_j| \le 1$ for every $j \in [d]$. For $t \in \mathbb{R}$, we define
  \begin{equation*}
    F(t) =  g\left(\left(\omega + t b_1\right) \cdots \left(\omega + t b_d \right) \right).
  \end{equation*}
  Note that $F(0) = g(\omega^d) = \omega$. To simplify our notation, for each $j \in [d]$, let $x_j(t) = \omega + t b_j$,  and set $y(t) = x_1(t) \cdots x_d(t)$, so $F(t) = g(y(t))$. We have
  \begin{equation*}
    F'(t) =  g'\left( y \left( t \right) \right) \sum_{j = 1}^d b_j \prod_{i = 1, \, i \ne j}^{d} x_i(t).
  \end{equation*}
  In particular, we obtain
  \begin{equation} \label{eq:G:first-derivative}
    F'(0) =  g'\left( \omega^d\right)  \sum_{j = 1}^d b_j \omega^{d-1} = z \sum_{j = 1}^d b_j.
  \end{equation}
  We have
  \begin{equation} \label{eq:G:second-derivative}
    F''(t) =  g''\left( y \left( t \right) \right) \left(\sum\nolimits_{j = 1}^d b_j  \prod\nolimits_{\substack{i = 1 \\ i \ne j}}^{d} x_i(t)\right)^{2} + 2 g'\left( y \left( t \right) \right) \sum_{1 \le j < i \le d} b_j b_i \prod_{\substack{l=1\\ l \ne i, j}}^d x_l(t).
  \end{equation}
  Since $g'(\omega^d), g''(\omega^d) \ne \infty$ (by assumption) and $y(0) = \omega^d$, from the continuity of the maps $y$, $g'$ and $g''$ we find that there is $\delta_0 := \delta_0(g, d, \omega) \in (0, 1)$ such that $g'(y(t))$ and $g''(y(t))$ are bounded when $|t| \le \delta_0$. Note that $|x_j(t)| \le |\omega| + 1$ when $|t| \le \delta_0$. Therefore, \eqref{eq:G:second-derivative} can be upper bounded when $|t| \le \delta_0$ by a constant $C_0 := C_0(g, d, \omega) > 1$. By Taylor's formula we conclude that, for every $t \in \mathbb{R}$ with $|t| \le \delta_0$, 
  \begin{equation}\label{eq:G:second-derivative:bound}
    \left| F(t) - F(0) - F'(t) t \right| \le C_0 t^2.
  \end{equation}
  Finally, let $a_1, \ldots, a_d$ with $|a_j| \le \delta_0$. We choose $t = \max_{j \in [d]} |a_j|$. The result for $t = 0$ is equivalent to $F(0) = \omega$. Hence, we can assume that $t > 0$ and define $b_j = a_j / t$ for $j \in [d]$. The result then follows from $F(0) = \omega$, \eqref{eq:G:first-derivative} and \eqref{eq:G:second-derivative:bound}.
\end{proof}

\begin{remark}
  If the Mobius map $g$ of Lemma~\ref{lem:g:1} is given explicitly, then the constants $\delta_{0}(g, d, \omega)$ and  $C_0(g, d, \omega)$ can be determined explicitly as it is done for $g(x) = 1 / (1 +\lambda x)$ in the proof of \cite[Lemma 7.9]{Bezb}. 
\end{remark}

As noted in \cite[Section 7]{Bezb}, straight-line-programs can generate evaluations of any polynomial $p(z)$ with positive coefficients, up to a factor $z^n$. This property of straight-line-programs is used in  \cite[Lemma 2.10]{Bezb} in conjunction with a density result on evaluations of polynomials to come up with hardcore-programs that generate approximations of any number near a fixed point of $f(x) = 1 / (1 + \lambda x^d)$. Here we extend this result to Mobius-programs in Lemma~\ref{lem:covering}. Apart from differences in notation, the proof is the same as that of \cite[Lemma 2.10]{Bezb}; hence, we omit this proof and only highlight the notation differences. We also note that the difference between Lemmas~\ref{lem:covering} and~\ref{lem:efficient-covering} is that the latter gives an algorithm whereas the former only proves existence of these Mobius-programs.

\begin{lemma}[{\cite[Lemma 2.10]{Bezb} for Mobius-programs}] \label{lem:covering}
  Let $d$ be an integer with $d \ge 2$ and let $g$ be a Mobius map.  Let $f(x):= g(x^d)$ and let $\omega$ be a fixed point of $f$. Let us assume that the following assumptions hold.
  \begin{enumerate}
  	\item  $\omega$ is program-approximable for $g$ and $a_0 \in \mathbb{C}$;
  	\item \label{item:lem:covering:2}  $\omega \ne 0$ and  $g'(\omega^d), g''(\omega^d) \ne  \infty$;
  	\item Let $z := f'(\omega) / d = g'(\omega^d) \omega^{d-1}$. We have $0 < |z| < 1$ and $z \not \in \mathbb{R}$.
  \end{enumerate}
  Then, for any $\epsilon, \kappa > 0$ there exists a radius $\rho \in (0, \kappa)$ such that the following holds. For every $x \in B(\omega, \rho)$ there is a Mobius-program for $g$ starting at $a_0$ that generates $a_k$ with $|x-a_k| \le \epsilon \rho$.
\end{lemma}
\begin{proof}
  The proof is the same one as that of \cite[Lemma 2.10]{Bezb} apart from a few differences in notation. Here we point out these differences in notation so that the reader can translate the proof to our setting if needed. First of all, in our version we have an arbitrary Mobius map $g$ whereas \cite[Lemma 2.10]{Bezb} sets $g(z) = 1 / (1 + \lambda z)$ for some activity $\lambda \in \mathbb{Q}_{\mathbb{C}} \setminus \mathbb{R}$ of the independent set polynomial. The particular choice of $g$ does not affect the proof, so every instance of ``hardcore-program'' in \cite[Lemma 2.10]{Bezb} can be effectively replaced by ``Mobius-program for $g$ and $a_0$'', and every time that the proof invokes \cite[Lemma 7.9]{Bezb} we can use the more general Lemma~\ref{lem:g:1} instead. 

  The second main difference is that our statement adds a layer of generality in the choice of the fixed point $\omega \in \mathbb{C}$. In \cite[Lemma 2.10]{Bezb} $\omega$ is chosen as the fixed point of $f(z) = 1 / (1 + \lambda z^d)$ with the smallest norm. It turns out that such a fixed point satisfies the hypothesis of our statement. First, $\omega$ is program-approximable for $g(z) = 1 / (1 + \lambda z)$ and $a_0 = \lambda$ (see \cite[Lemma 2.7]{Bezb}). Secondly, we have $\omega$ and  $g'(\omega^d), g''(\omega^d) \ne  \infty$. Thirdly,  $0 < \lvert z \rvert < 1$ and $z \not \in \mathbb{R}$, see  \cite[Lemma 7.4]{Bezb} (we should point out that in  \cite{Bezb}  the authors set $z = \omega - 1$, which agrees with $z = g'(\omega^d)\omega^{d-1}$ for their choice of $g$). These are all the properties of $\omega$ needed to carry out the proof of  \cite[Lemma 2.10]{Bezb}.
  
  Finally, it is useful to note that if a hardcore-program generates a number $a_k$, then there is a tree of maximum degree at most $d+1$ that implements $\lambda a_k$. This explains why in the proof and statement of \cite[Lemma 2.10]{Bezb} there is an extra factor $\lambda$ when activities of the independent set polynomial are considered. Here we can omit this factor because we are not translating programs to gadgets.
\end{proof}

\subsection{Proof of Lemma~\ref{lem:efficient-covering}} \label{sec:neighbourhood:precision}

In this section we translate the results given in \cite[Section 7.3]{Bezb} to Mobius-programs. The main result of this section is Lemma~\ref{lem:efficient-covering}, which gives an algorithmic version of Lemma~\ref{lem:covering}. First, we need some technical results, Lemmas~\ref{lem:g:2} and~\ref{lem:g:3}, which extend \cite[Lemmas 7.10 and 7.11]{Bezb} to our more general setting.

\begin{lemma}[{\cite[Lemma 7.10 for Mobius-programs]{Bezb}}] \label{lem:g:2}
  Let $d$ be an integer with $d \ge 2$ and let $g$ be a Mobius map. Let  $\omega \in \mathbb{C}$ be a fixed point of $f(z) = g(z^d)$ with $\omega \ne 0$, $g'(\omega^d) \not \in \{0, \infty\}$ and $g''(\omega^d) \ne \infty$. Set $z := g'(\omega^d)\omega^{d-1}$. There exist reals $C_1 := C_1(g, d, \omega) > 1$ and $\delta_1 := \delta_1(g, d, \omega) > 0$ such that for any $a_1, \ldots, a_d \in \mathbb{C}$ with $|a_j| \le \delta_1$ (for $j \in [d]$) we have
\begin{equation*}
 \Phi^{-1}\left(\omega + a_d\right) = \omega + \frac{a_d}{z} - \sum_{j = 1}^{d-1} a_j + \tau,
\end{equation*} 
where  
\begin{equation*}
  \Phi \left( x \right) = g \left( x \prod\nolimits_{j = 1}^{d-1} \left( \omega + a_j \right)  \right)
\end{equation*}
and $\tau \in \mathbb{C}$ with $|\tau| \le C_0 \max_{j \in [d]} |a_j|^2$. 
\end{lemma}
\begin{proof}
  First, note that
  \begin{equation*}
    \Phi^{-1}\left( x \right) = g^{-1} \left( x \right) \prod_{j = 1}^{d-1} \left( \omega + a_j \right)^{-1}.
  \end{equation*}
  Let $b_1, \ldots, b_d \in \mathbb{C}$ with $|b_j| \le 1$ for every $j \in [d]$. For $t \in (-\lvert \omega \rvert, \lvert \omega \rvert)$, note that $\omega + t b_j \ne 0$, so we can define
  \begin{equation*}
    F(t) =  g^{-1}\left(\omega + t b_d \right) \prod_{j = 1}^{d-1} \left( \omega + t b_j \right)^{-1}.
  \end{equation*}
  We note that when $g$ is particularised to $g(x) = 1/(1+\lambda x)$, $F$ coincides with the definition of $F$ given in \cite[Lemma 7.10]{Bezb}. Moreover, $F(t)$ agrees with $\Phi^{-1}(\omega + a_d)$ for $t = \max_{j \in \{ 1, \ldots, d \}} \lvert a_j \rvert$ and $b_{j } = a_j / t$.  Note that $F(0) = g^{-1}(\omega) \omega^{-d+1} = \omega$. One can check that $F'(0) = b_d/z - \sum_{j=1}^{d-1} b_j$. The proof is now  analogous to that of \cite[Lemma 7.10]{Bezb}, with the difference that the constants $C_1 := C_1(g, d, \omega) > 1$ and $\delta_1 := \delta_1(g, d, \omega) > 0$ are not explicitly determined but rather obtained by a continuity argument as in Lemma~\ref{lem:g:1}  that uses the hypotheses $\omega \ne 0$, $g'(\omega^d) \not \in \{0, \infty\}$ and $g''(\omega^d) \ne \infty$. Hence, we do not repeat the rest of the proof here.
\end{proof}

\begin{lemma}[{\cite[Lemma 7.11 for Mobius-programs]{Bezb}}] \label{lem:g:3}
  Let $d$ be an integer with $d \ge 2$ and let $g$ be a Mobius map. Let  $\omega \in \mathbb{C}$ be a fixed point of $f(z) = g(z^d)$ with $g'(\omega^d), g''(\omega^d) \ne \infty$. Set $z := g'(\omega^d)\omega^{d-1}$.  There exist reals $C_2 := C_2(g, d, \omega) > 1$ and $\delta_2 := \delta_2(g, d, \omega) > 0$ such that for any $a_1, \ldots, a_d \in \mathbb{C}$ with $|a_j| \le \delta_2$ (for $j \in [d]$) we have
\begin{equation*}
 \Phi'\left(\omega + a_d\right) = z + \tau,
\end{equation*} 
where  
\begin{equation*}
  \Phi \left( x \right) = g \left( x \prod\nolimits_{j = 1}^{d-1} \left( \omega + a_j \right)  \right)
\end{equation*}
and $\tau \in \mathbb{C}$ with $|\tau| \le C_0 \max_{j \in [d]} |a_j|$. 
\end{lemma}
\begin{proof}
  First, note that
  \begin{equation*}
    \Phi'\left( x \right) = g' \left( x \prod\nolimits_{j = 1}^{d-1} \left( \omega + a_j \right) \right) \prod_{j = 1}^{d-1} \left( \omega + a_j \right).
  \end{equation*}
  Let $b_1, \ldots, b_d \in \mathbb{C}$ with $|b_j| \le 1$ for every $j \in [d]$. For $t \in \mathbb{R}$, we define
  \begin{equation*}
    F(t) =  g'\left( \prod\nolimits_{j = 1}^d  \left(\omega + t b_j\right)\right) \prod_{j = 1}^{d-1} \left(\omega + t b_j\right),
  \end{equation*}
  so $F(t)$ agrees with $\Phi'(\omega + a_d)$ for $t = \max_{j \in \{ 1, \ldots, d \}} \lvert a_j \rvert$ and $b_{j } = a_j / t$.   Note that $F(0) = g'(\omega^d) \omega^{d-1} = z$. At this point the proof is analogous to that of Lemma~\ref{lem:g:1}, so we are not repeating it again. The only difference is that this time we have to bound $F'(t)$, instead of $F''(t)$, in a neighbourhood of $t = 0$, obtaining $\delta_2 := \delta_2(g, d, \omega) \in (0, 1)$ and $C_2 := C_2(g, d, \omega) > 1$ such that $\lvert F'(t)\rvert \le C_2$ for all $t \in (-\delta_2, \delta_2)$. In \cite[Lemma 7.11]{Bezb} the constants  $\delta_2$ and $C_2$ are made precise for the choice $g(x) = 1/ (1 + \lambda x)$, whereas here we obtain them by continuity of $F'(t)$ and the fact that $g'(\omega^d), g''(\omega^d) \ne \infty$.
\end{proof}

\begin{lemma}[{\cite[Lemma 7.12 for Mobius-programs]{Bezb}}] \label{lem:covering:setup}
  Let $d$ be an integer with $d \ge 2$ and let $g$ be a Mobius map.  Let $f(x):= g(x^d)$ and let $\omega$ be a fixed point of $f$. Let us assume that the following assumptions hold.
  \begin{enumerate}
  	\item  $\omega$ is program-approximable for $g$ and $a_0 \in \mathbb{C}$;
  	\item \label{item:assumptions:2} $\omega \ne 0$ and  $g'(\omega^d), g''(\omega^d) \not \in \{0, \infty\}$; 
  	\item Let $z := f'(\omega) / d = g'(\omega^d) \omega^{d-1}$. We have $0 < |z| < 1$ and $z \not \in \mathbb{R}$.
  \end{enumerate}
 Then there are Mobius-programs for $g$ starting at $a_0$ that generate  $\{\lambda_0, \lambda_1, \ldots, \lambda_t\} \subseteq \mathbb{C}$, and a real $r > 0$ such that the following hold for all $\hat{\omega} \in B(\omega, r)$.
\begin{enumerate}
\item For $i = 0$, $\lambda_0 \in B(\hat{\omega}, 2r)$. 
\item For $i = 1, \ldots, t$, the map $\Phi_i$ given by $\Phi_i(x) = g(x \lambda_i \lambda_0^{d-2} )$ is contracting on the ball $B(\hat{\omega}, 2r)$.
\item $B(\hat{\omega}, 2r) \subseteq \bigcup_{i = 1}^t \Phi_i(B(\hat{\omega}, 2r))$.
\end{enumerate}
\end{lemma}
\begin{proof}
  The proof is exactly the same one as that of \cite[Lemma 7.12]{Bezb} apart from the differences in notation mentioned in the proof of Lemma~\ref{lem:covering}, and the fact that we use the more general Lemma~\ref{lem:g:2} instead of \cite[Lemma 7.10]{Bezb} and the more general Lemma~\ref{lem:g:3} instead of \cite[Lemma 7.11]{Bezb}. 
\end{proof}

When one has maps $\Phi_1, \ldots, \Phi_t$ with the properties $2$ and $3$ of Lemma~\ref{lem:covering:setup}, there is an efficient algorithm to approximate numbers using sequential applications of the maps $\Phi_1, \ldots, \Phi_t$, see Lemma~\ref{lem:speed-up}.

\begin{lemma}[{\cite[Lemma 2.8]{Bezb}}] \label{lem:speed-up}
Let $z_0 \in \algcomplex $, $r \in \algebraic_{> 0}$   and $U$ be the ball $B(z_0, r)$. Further, suppose that $\Phi_1, \ldots, \Phi_t$ are Mobius maps (with coefficients in 
$\algcomplex $) that satisfy the following:
\begin{enumerate}
\item for each $i \in [t]$, $\Phi_i$ is contracting on the ball $U$,
\item $U \subseteq \bigcup_{i = 1}^t \Phi_i(U)$.
\end{enumerate}
Then there is a polynomial-time algorithm which, on input $(i)$ a starting point 
$x_0 \in U \cap \algcomplex$, $(ii)$ a target $x \in U \cap \algcomplex$, and $(iii)$ a rational $\epsilon > 0$, outputs a number $\hat{x} \in U \cap \algcomplex$ and a sequence $i_1, i_2, \ldots, i_k \in [t]$ such that 
\begin{equation*}
  \hat{x} = \Phi_{i_k} \left( \Phi_{i_{k-1}} \left( \cdots \Phi_{i_1} \left( x_0 \right) \cdots \right) \right) \text{ and } \left| x - \hat{x} \right| \le \epsilon.
\end{equation*}
\end{lemma}

Even though \cite[Lemma 2.8]{Bezb} is stated for particular maps $\Phi_i$ that arise in the context of the independent set polynomial, its proof is more general and works in the setting of Lemma~\ref{lem:speed-up}. Now Lemma~\ref{lem:efficient-covering} follows from combining Lemmas~\ref{lem:covering:setup} and~\ref{lem:speed-up}.

\begin{lemefficientcovering}[{\cite[Proposition 2.6 for Mobius-programs]{Bezb}}] 
  \statelemefficientcovering
\end{lemefficientcovering}
\begin{proof}
The proof is the same as that of \cite[Proposition 2.6]{Bezb}, the main differences being that we invoke the more general Lemmas~\ref{lem:covering:setup} and~\ref{lem:speed-up} instead of \cite[Lemmas 7.12 and 2.8]{Bezb}. We also we stop the proof once we have obtained the desired program instead of translating the program to a gadget for the independent set polynomial. Finally, in the definition of densely program-approximable in polynomial time we ask the algorithm to compute $k$ approximations $x_1, \ldots, x_k$ of $\lambda'$.  This can be done by running $k$ versions of the algorithm given in the proof of  \cite[Proposition 2.6]{Bezb}  and setting  a different value for $x_0$ in each version when applying Lemma~\ref{lem:speed-up}. In the proof of \cite[Proposition 2.6]{Bezb} the value $x_0$ is a good approximation of the fixed point $\omega$. These distinct values for $x_0$ are obtained by generating a better approximation of the fixed point $\omega$ each time. The generated elements $x_1, \ldots, x_k$ will be of the form $\Phi_{i_j} \left( \Phi_{i_{j-1}} \left( \cdots \Phi_{i_1} \left( x_0 \right) \cdots \right) \right)$, so all of them are distinct because the starting points for $x_0$ are distinct and the maps $\Phi_i$ are bijective.
\end{proof}

\subsection{Proof of Lemma~\ref{lem:implementations}} \label{sec:mobius:plane}

In this section we prove Lemma~\ref{lem:implementations}, that is, we show how to generate approximations of any complex number with a Mobius-program. This generalises \cite[Proposition 2.2]{Bezb} to Mobius-programs. Up to this point the results of \cite{Bezb} on hardcore-programs have been generalised to Mobius-programs without much effort. In this section we have to refine the arguments given in the proof of \cite[Proposition 2.2]{Bezb} to make it work for any Mobius map $g$, although the main idea stays the same: starting from a repelling fixed point and applying results of complex dynamics (see Section~\ref{sec:pre:cd}) to come up with an appropriate Mobius-program.

First, let us make some remarks about the proof of \cite[Proposition 2.2]{Bezb}. This result shows how to efficiently implement approximations of any complex activity of the independent set polynomial via a hardcore-program. The proof is divided into three steps. First, the authors show how to generate approximations of any activity sufficiently large. Then they use the fact that $g(x) = 1 / (1 + \lambda x)$ tends to $0$ when $x$ diverges to generate approximations of any complex number near $0$. Finally, they combine both results to generate an approximation of any complex number. Unfortunately, the second step breaks for arbitrary Mobius maps and, in particular, for the Mobius maps that we use when particularising these results to the Ising model. This motivates the work presented in this section. 

This section is organised as follows. In Lemma~\ref{lem:efficient-covering:any-point} we show that if a repelling fixed point of $f(z) := g(z^d)$ is densely program-approximable in polynomial-time, then any complex point is densely program-approximable in polynomial-time. In particular, this includes the point $0$ that escapes from the arguments given in \cite{Bezb}, which rely heavily on the fact that complex holomorphic maps are locally Lipschitz. Instead of this local property, here we use the fact that rational maps are Lipschitz on the Riemann sphere with respect to the chordal metric (Lemma~\ref{lem:lipschitz}), which simplifies the proofs because we do not have to deal so carefully with the poles of the rational map.
In Lemma~\ref{lem:efficient-covering:infinity} we show how to generate approximations of any complex number that is sufficiently large. Although Lemma~\ref{lem:efficient-covering:infinity} could follow from the technical proof given in \cite[Proposition 2.2]{Bezb}, we include a simpler proof that goes along the same lines as the proof of Lemma~\ref{lem:efficient-covering:any-point}. Finally, we combine Lemmas~\ref{lem:efficient-covering:any-point} and~\ref{lem:efficient-covering:infinity} to prove Lemma~\ref{lem:implementations}.

\begin{lemma} \label{lem:efficient-covering:any-point}
  Let $d$ be an integer with  coefficients in $ \algcomplex$. 
  Let $\omega \in \mathbb{C}$ be a repelling fixed point of $f(z) := g(z^d)$ that is densely program-approximable in polynomial time for $g$ and $a_0 \in \algcomplex$.   
  Let $E_f$ be the exceptional set of the rational map $f$ and let $\gamma \in \mathbb{C} \setminus E_f$. Then $\gamma$ is  densely program-approximable in polynomial time for $g$ and $a_0$.
\end{lemma}
\begin{proof}
  Let $r_{\omega} > 0$ from the definition of densely program-approximable point for $\omega$ (Definition~\ref{def:dpa}). Since $\omega$ is a repelling fixed point of $f$, by Lemma~\ref{lem:repelling} it belongs to the Julia set of $f$ and, thus, by Theorem~\ref{thm:cd}, $\bigcup_{n = 0}^{\infty} f^n(B(\omega, r_{\omega})) = \widehat{\mathbb{C}} \setminus E_f$. Let $N$ be the smallest non-negative integer such that $\gamma \in  f^N(B(\omega, r_{\omega}))$. If $N = 0$, then the fact that $\gamma$ is densely program-approximable in polynomial time for $g$ and $a_0$ is trivial: let $r_{\gamma}$ be any positive real number with $B(\gamma, r_{\gamma}) \subseteq B(\omega, r_{\omega})$ and use the algorithm from the definition of densely program-approximable (in polynomial time) point for $\omega$ on the inputs $\lambda \in B(\gamma, r_{\gamma}) \cap \algcomplex$ and $\epsilon > 0$ rational. In the rest of the proof we deal with the case $N \ge 1$.

   Let $x \in  B(\omega, r_{\omega})$ such that $f^N(x) = \gamma \in \mathbb{C}$. Let 
  \begin{equation*}
    \mathcal{P} = \{z \in \mathbb{C} : z \text{ is a pole of } f^n \text{ for some } n \in [N]\}.
  \end{equation*}
  Note that $\mathcal{P}$ is a finite set so there is $r > 0$ such that 
  \begin{equation} \label{eq:covering:P}
  	 \overline{B}(x, r) \subseteq B(\omega, r_{\omega}) \qquad \text{and} \qquad  \overline{B}(x, 2r) \cap \mathcal{P} \subseteq \{x\}.
  \end{equation}
  We point out here that $x \in \mathcal{P}$ if and only if there is $n \in [N-1]$ such that $f^n(x) = \infty$.  Since $f^N(\overline{B}(x, 2r))$ is a compact set of complex numbers ($f^N$ is continuous on $\overline{B}(x, 2r)$ as a complex function due to the lack of poles), there is a rational constant $C > 0$ depending on $\gamma$ such that 
  \begin{equation}
    \label{eq:covering:bound1}
    \lvert f^N(z) \rvert \le C \quad \text{for every } z \in \overline{B}(x, 2r).
  \end{equation}
  Since $f^N$ is a rational function, $f^N(B(x, r))$ is an open set in $\widehat{\mathbb{C}}$ by the open mapping theorem (Proposition~\ref{thm:openmapping}).
  Hence, there is a rational $r_{\gamma} > 0$ with $B(\gamma, r_{\gamma}) \subseteq f^N(B(x, r))$ (here we used that $\gamma \ne \infty$). This is the radius in the definition of densely program-approximable point for $\gamma$ (Definition~\ref{def:dpa}).  By Lemma~\ref{lem:lipschitz}, there is a rational $L \ge 1$ such that $f^N$ is Lipschitz with constant $L$ in $\widehat{\mathbb{C}}$ with respect to the chordal metric.

  Now we proceed to give the polynomial-time algorithm. Let $k$ be a positive integer. Let $\lambda \in B(\gamma, r_{\gamma}) \cap \algcomplex  $ and $\epsilon > 0$ rational be the inputs of the algorithm. We are going to compute $k$ elements of Mobius-programs that approximate $\lambda$ up to an error $\epsilon$.  We set
  \begin{equation}
    \label{eq:any-point:epsilon}
    \epsilon' = \frac{\epsilon}{1 + C^{2}} \qquad \text{ and } \qquad \epsilon'' = \min \left\{\epsilon' / L, r\right\}.
  \end{equation}
  We can write $f^N(z) = P(z) / Q(z)$ where $P$ and $Q$ are polynomials with coefficients in $\algcomplex$.  Note that the equation $P(z) / Q(z) = \lambda$ in $z \in B(x, r) \cap \algcomplex$ is equivalent to $P(z) = \lambda  Q(z)$ in $z \in B(x, r) \cap \algcomplex$  because $Q$ has no zeros in $B(x, r)$. We can solve this polynomial equation numerically as described in the proof of \cite[Proposition 2.2, case I]{Bezb} to compute $x' \in B(x, r) \cap \algcomplex$ with $\lvert x^* - x' \rvert \le \epsilon'' / 2$ for some solution $x^*$ of  $P(z) = \lambda  Q(z)$ with $x^* \in B(x, r) \cap \algcomplex$, so $f^N(x^*) = \lambda$. Since $x' \in  B(x, r) \cap \algcomplex   \subseteq B(\omega, r_{\omega})$, we can use the algorithm of the definition of densely program-approximable in polynomial time for $\omega$ to compute $k+1$ distinct elements $\hat{x}_1, \ldots, \hat{x}_{k+1}$ of a Mobius-program for $g$ and $a_0$ with $\lvert x' - \hat{x}_j \rvert \le \epsilon''/2$.  Let $\hat{x}$ be any of these $k+1$ elements and let us analyse how close $f^N ( \hat{x})$ is to $\lambda$. We have  $\lvert x^* - \hat{x} \rvert \le \epsilon''$ by the triangle inequality. We claim that $\lvert \lambda - f^N(\hat{x}) \rvert \le \epsilon$. In view of the Lipschitz property of $f^N$ and \eqref{eq:any-point:epsilon}, we have 
    \begin{equation*}
     d \left(  \lambda, f^N \left( \hat{x} \right) \right) \le L d \left( x^*, \hat{x} \right) = L \frac{\left| x^* - \hat{x} \right|}{ \left( \left( 1 + \lvert x^* \rvert^2 \right) \left( 1 + \lvert \hat{x} \rvert^2 \right) \right)^{1/2} } \le  L \left| x^* - \hat{x} \right| \le L \epsilon'' \le  \epsilon'.
  \end{equation*}
Note that by the triangle inequality, $\lvert x - \hat{x} \rvert \le \lvert x - x' \rvert + \lvert x' - \hat{x} \rvert \le r + \epsilon'' \le  2r$. We can now use the upper bound \eqref{eq:covering:bound1} with $z = \hat{x}$ and $z = x^{*}$ to conclude that $\lvert f^N(x^*)\rvert = \lvert \lambda \rvert \le C$ and
  \begin{equation*}
    \left| \lambda - f^N \left( \hat{x} \right) \right| = \left( \left( 1 + \lvert f^N \left( \hat{x} \right) \rvert^2   \right)  \left( 1 + \lvert \lambda \rvert^2   \right) \right)^{1/2} d \left(  \lambda, f^N \left( \hat{x} \right) \right) \le \left(  1 + C^2 \right)  d \left( \lambda, f^N \left( \hat{x} \right) \right) = \epsilon.
  \end{equation*}
  The algorithm chooses $k$ numbers in $\{f^N \left( \hat{x}_1 \right), \ldots, f^N \left( \hat{x}_{k+1} \right)\}$ in conjunction with its representation as a sequence of tuples (which comes from the representation of $\hat{x}$ and the corresponding applications of $f$). Recall that $N$ does not depend on the inputs $\epsilon$ and $\lambda$, so computing $f^N \left( \hat{x} \right)$ from $\hat{x}$ only adds a constant factor to the running time. In order to conclude the proof, we have to ensure that $f^1(\hat{x} ), f^2(\hat{x} ), \ldots, f^N(\hat{x} )$ is a sequence of complex numbers, that is, $\infty$ does not appear in the sequence. It is enough to show that this is the case for at least $k$ of the $k+1$ outputs  $f^N \left( \hat{x}_1 \right), \ldots, f^N \left( \hat{x}_{k+1} \right)$, since we only wanted $k$ outputs to begin with. This follows from \eqref{eq:covering:P} and the fact that at most one of the numbers $\hat{x}_1, \ldots, \hat{x}_{k+1}$ is equal to $x$.
\end{proof}

\begin{lemma}[{\cite[Proposition 2.2, case I]{Bezb}}] \label{lem:efficient-covering:infinity}
  Let $d$ be an integer with $d \ge 2$ and let $g$ be a Mobius map  with coefficients in $\algcomplex$. Let $\omega \in \mathbb{C}$ be a repelling fixed point of $f(z) := g(z^d)$ that is densely program-approximable for $g$ and $a_0 \in \algcomplex$.  Let $E_f$ be the exceptional set of the rational map $f$. If $\infty \not \in E_f$, then there exists a rational number $M > 1$ such that the following holds.

  There is a polynomial-time algorithm such that, on input $\lambda \in  \algcomplex$ with $\lvert \lambda\rvert > M$ and rational $\epsilon > 0$, computes an element $a_k$ of a Mobius-program for $g$ starting at $a_0$ with $\lvert \lambda - a_k \rvert \le \epsilon$.
\end{lemma}
\begin{proof}
  This lemma can be proven following the argument given in the first case of the proof of \cite[Proposition 2.2]{Bezb} for the independent set polynomial. Here we give a simpler proof that works even when $g$ is just a rational function. The original proof is significantly more technical because the authors first get close to a pole and then apply one more iteration of $f$ to get near the desired point $\lambda$. To do this, they have to make sure that the poles of $f^1, \ldots, f^N$ are excluded from the all the domains considered and that all the applications of $f$ are locally Lipschitz.
  
  Our proof follows the same structure as that of Lemma~\ref{lem:efficient-covering:any-point} with $\gamma = \infty$, but it requires a slightly different analysis because in the proof $x$ is a pole of $f^N$.  Let $r_{\omega} > 0$ from the definition of  densely program-approximable fixed point for $\omega$. Since $\omega$ is a repelling fixed point of $f$, by Lemma~\ref{lem:repelling} it belongs to the Julia set of $f$ and, thus, by Theorem~\ref{thm:cd}, $\bigcup_{n = 0}^{\infty} f^n(B(\omega, r_{\omega})) = \widehat{C} \setminus E_f$. Since $\infty \not \in E_f$, we can consider the smallest non-negative integer $N$ such that $\infty \in  f^N(B(\omega, r_{\omega}))$. Note that $N \ge 1$ because $\infty \not \in B(\omega, r_{\omega})$.

  Let $x \in  B(\omega, r_{\omega})$ such that $f^N(x) =\infty$. Let  
  \begin{equation*}
    \mathcal{P} = \{z \in \mathbb{C} : z \text{ is a pole of } f^n \text{ for some } n \in [N]\}
  \end{equation*}
  and
  \begin{equation*}
    \mathcal{Z} = \{z \in \mathbb{C} : z \text{ is a zero of } f^n \text{ for some } n \in [N]\}.
  \end{equation*}
  Note that $\mathcal{P}$ and $\mathcal{Z}$ are  finite sets and $x  \in \mathcal{P}$ because $f^N(x) = \infty$.   Let $\delta$ be the minimum distance between any two distinct numbers in $\mathcal{P} \cup \mathcal{Z}$. There is $0 < r < \delta/4$ such that
    \begin{equation} \label{eq:infinity:PZ}
  	\overline{B}(x, r) \subseteq B(\omega, r_{\omega}) \qquad \text{and} \qquad  \overline{B}(x, 2r) \cap (\mathcal{P} \cup \mathcal{Z}) = \{ x \}.
  \end{equation}
  Since $f^N$ is a rational function, $f^N(B(x, r))$ is an open set in $\widehat{\mathbb{C}}$, see Section~\ref{sec:pre:cd}. Hence, by the topology of the Riemann sphere and the fact that $\infty \in f^N(B(x, r))$, there is $M > 1$ with $U =  \{z \in \mathbb{C} : \lvert z \rvert > M \} \subseteq f^N(B(x, r))$. This is the constant given in the statement. 	We can write $f^N(z) = P(z) / Q(z)$ where $P$ and $Q$ are polynomials with coefficients in $\algcomplex$  that do not share any root, so the set of roots of $P$ is $\mathcal{Z}$ and the set of roots of $Q$ is $\mathcal{P}$.  We are going to bound $\lvert P(z) \rvert$ and $\lvert Q(z) \rvert$ in $B(x, r)$.  First, let us bound $\lvert P(z) \rvert$. Let $k$ be the multiplicity of the pole $x$ of $f^N$. For any $z \in \overline{B}(x, 2r)$ and $\zeta \in \mathcal{Z} \cup  \mathcal{P}$ with $\zeta \ne x$, in view of $r < \delta/4$ and \eqref{eq:infinity:PZ}, we have 
  \begin{equation*}
	\lvert z - \zeta \rvert \le \lvert x - \zeta  \rvert + \lvert z -x \rvert \le  \lvert x - \zeta  \rvert  + 2r \le 2r + \max_{\zeta' \in  \mathcal{Z} \cup  \mathcal{P}} \left| x - \zeta' \right|
\end{equation*}
 and
  \begin{equation*}
  	\lvert z - \zeta \rvert \ge \lvert x - \zeta  \rvert - \lvert x -z \rvert \ge \delta - 2r \ge  \delta/2.
  \end{equation*}
 Hence,  there are real numbers $C_0, C_1 > 0$ such that, for any  $z \in \overline{B}(x, 2r)$,
  \begin{equation} \label{eq:infinity:pre-bounds}
  	C_0 (\delta / 2)^{\deg P} \le \left| P(z)\right|  \le  C_1 \quad   \text{ and } \quad  C_0 (\delta/2)^{(\deg Q) - k} \left| x - z\right|^k \le  \left| Q(z)\right| \le  C_1 \left| x - z\right| ^k.
  \end{equation}
  Combining the bounds in \eqref{eq:infinity:pre-bounds} and setting  $D_0 =  C_0 (\delta/2)^{\deg P} / C_1  $ and $D_1 = C_1 (\delta/2)^{k- (\deg Q)}  /  C_0$  we find that, for any  $z \in \overline{B}(x, 2r)$,
  \begin{equation} \label{eq:lem:infinity:bounds}
  	D_0  \left| x - z\right| ^{-k}	\le \left| f^N(z)\right|  \le D_1  \left| x - z\right| ^{-k}.
  \end{equation}
 Note that $D_0$ and $D_1$ are positive.  The bounds \eqref{eq:lem:infinity:bounds} play an important role in our algorithm. We also need Lemma~\ref{lem:lipschitz}, which gives  a rational $L \ge 1$ such that $f^N$ is Lipschitz with constant $L$ in $\widehat{\mathbb{C}}$. 
  
  Now we proceed to give the polynomial-time algorithm. Let $\lambda \in \algcomplex  $ with $\lvert \lambda \rvert > M $ and $\epsilon > 0$ rational be the inputs. 
We set $\tau = \frac{D_1}{D_0} 2^k \lvert \lambda \rvert$ and 
	\begin{equation}
		\label{eq:infinity:epsilon}
		\epsilon' = \frac{\epsilon}{  \left( \left(1 + \lvert \tau \rvert^2\right) \left(1 + \lvert \lambda\rvert^2\right) \right)^{1/2}} \qquad \text{ and } \qquad \epsilon'' = \min \left\{\epsilon' / L, r,  \frac{1}{2}\left( D_0 / \lvert \lambda \rvert \right)^{1/k} \right\}.
	\end{equation}
Note that $\mathrm{size} ( \epsilon'')  = \mathrm{poly} (\mathrm{size}(\epsilon), \mathrm{size}(\lambda ))$ since $L, D_0, D_1, r, k$ are constants that do not depend on the inputs.   The equation $P(z) / Q(z) = \lambda$ in $z \in B(x, r) \cap \algcomplex$ is equivalent to $P(z) = \lambda  Q(z)$ in $z \in B(x, r) \cap \algcomplex$  because $Q$ has no zeros in $B(x, r)$ other than $x$.  We can solve this polynomial equation numerically as described in the proof of \cite[Proposition 2.2, case I]{Bezb} to compute $x' \in B(x, r) \cap
\algcomplex$ with $\lvert x^* - x' \rvert \le \epsilon'' / 2$ for some solution $x^*$ of  $P(z) = \lambda  Q(z)$ with $x^* \in B(x, r) \cap \algcomplex$, so $f^N(x^*) = \lambda$. Since $x' \in  B(x, r) \cap \algcomplex   \subseteq B(\omega, r_{\omega})$, we can use the algorithm of the definition of densely program-approximable in polynomial time for $\omega$ to compute $2$ distinct elements $\hat{x}_1, \hat{x}_{2}$ of a Mobius-program for $g$ and $a_0$ with $\lvert x' - \hat{x}_j \rvert \le \epsilon''/2$ for any $j \in \{1,2\}$. By  the triangle inequality,  we have $\lvert x^* - \hat{x}_j \rvert \le \epsilon''$ and $\lvert x - \hat{x}_j \rvert \le \lvert x - x^* \rvert  + \lvert x^* - \hat{x}_j \rvert  \le r + \epsilon'' \le 2r$, where we used \eqref{eq:infinity:epsilon}. In light of the choice of $r$ in \eqref{eq:infinity:PZ}, for any $z \in \overline{B}(x, 2r)$, we have $f^N(z) = \infty$ if and only if $z = x$. Hence, we can check if $\hat{x}_j$ is $x$ by evaluating $f^N(\hat{x}_j)$ and checking if the result is $\infty$ or not. Since $\hat{x}_1$ and $\hat{x}_{2}$ are distinct, at least one of the two is not $x$, so we can pick $\hat{x} \in \{\hat{x}_1, \hat{x}_{2}\}$ with $\hat{x} \ne x$. We work with $\hat{x}$ in the rest of the proof. We claim that $\lvert \lambda - f^N(\hat{x}) \rvert \le \epsilon$. 
Recall that $x^*, \hat{x} \in \overline{B}(x, 2r)$. In view of the bounds \eqref{eq:lem:infinity:bounds} and  \eqref{eq:infinity:epsilon} we have
\begin{equation*}
	\frac{1}{2}\left| x - x^* \right| \ge  \frac{1}{2}  \left(\frac{D}{\lvert f^N(x^*)\rvert}\right)^{1/k}  =  \frac{1}{2} \left(\frac{D}{\lvert \lambda \rvert}\right)^{1/k} \ge  \epsilon'' ,
\end{equation*}
so $\left| x - \hat{x} \right| \ge \left| x - x^* \right| -  \left| x^* - \hat{x} \right| \ge \left| x - x^* \right| -  \epsilon'' \ge \left| x - x^* \right| /2$ and 
\begin{equation} \label{eq:infinity:upper-bound}
	\left| f^N \left( \hat{x} \right) \right|  \le   \frac{D_1}{ \left| x - \hat{x} \right|^k  } \le  D_1  \left(  \frac{2}{ \left| x - x^* \right| } \right) ^k   \le \frac{D_1	}{D_0} 2^k \lvert \lambda \rvert = \lvert \tau \rvert.
\end{equation}
In view of the Lipschitz property of $f^N$ and \eqref{eq:infinity:epsilon}, we have 
	\begin{equation*}
		d \left(  \lambda, f^N \left( \hat{x} \right) \right) \le L d \left( x^*, \hat{x} \right) = L \frac{\left| x^* - \hat{x} \right|}{ \left( \left( 1 + \lvert x^* \rvert^2 \right) \left( 1 + \lvert \hat{x} \rvert^2 \right) \right)^{1/2} } \le L \left| x^* - \hat{x} \right| \le L \epsilon'' \le  \epsilon',
	\end{equation*}
  which in combination with \eqref{eq:infinity:upper-bound} yields
	\begin{equation*}
		\left| \lambda - f^N \left( \hat{x} \right) \right| = \left( \left( 1 + \lvert f^N \left( \hat{x} \right) \rvert^2   \right)  \left( 1 + \lvert \lambda \rvert^2   \right) \right)^{1/2} d \left(  \lambda, f^N \left( \hat{x} \right) \right) \le  \left( \left(1 + \lvert \tau \rvert^2\right) \left(1 + \lvert \lambda\rvert^2\right) \right)^{1/2} \epsilon' = \epsilon
	\end{equation*}
  as we wanted to prove. The algorithm outputs $f^N \left( \hat{x} \right)$ in conjunction with its representation as a sequence of tuples (which comes from the representation of $\hat{x}$ and the corresponding applications of $f$). In order to conclude the proof, we have to guarantee that the sequence $f^1(\hat{x} ), f^2(\hat{x} ), \ldots, f^N(\hat{x} )$ does not contain the point $\infty$.  In light of \eqref{eq:infinity:PZ},  $\hat{x}$ is a pole of $f^n$ for some $n \in [N]$ if and only if $\hat{x} = x$. But we chose $\hat{x} \in \{ \hat{x}_1, \hat{x}_2\}$ with $\hat{x} \ne x$, concluding the proof.
\end{proof}

Now we can combine Lemmas~\ref{lem:efficient-covering:any-point} and~\ref{lem:efficient-covering:infinity} to prove Lemma~\ref{lem:implementations}. The proof follows the same idea as that of \cite[Proposition 2.2, case I]{Bezb} with the difference that we can not use the particular shape of $g$ to simplify some steps. Hence, we have again to use the Lipschitz property of rational functions on the Riemann sphere (Lemma~\ref{lem:lipschitz}).

\begin{lemmobiusimplementations}[{\cite[Proposition 2.2 for Mobius-programs]{Bezb}}]  
  \statelemmobiusimplementations
\end{lemmobiusimplementations}
\begin{proof}
 By Lemma~\ref{lem:lipschitz} there is  a rational $L \ge 1$ such that $f^N$ is Lipschitz with constant $L$ in $\widehat{\mathbb{C}}$. Note that $g$ only has one complex pole because $g$ is a Mobius map with $g(\infty) \in \mathbb{C}$.  Let $p \in \mathbb{C}$ be the pole of $g$, so $g(p) = \infty$.  We can write $g$ as $g(z) =  a/(z - p) + b$ for some $a,b \in \algcomplex$ with $a \ne 0$. Let $r_0 \in (0,1)$ be rational number as in Lemma~\ref{lem:efficient-covering:any-point} for $\gamma = 0$. Let $M > 1$ be a rational number as in Lemma~\ref{lem:efficient-covering:infinity}. There is a rational $r > 0$ such that 
  \begin{equation} \label{eq:implementations:bound:1}
 	 \lvert g(z) \rvert > 2M \qquad \text{for all } z \in B(p, r).
 \end{equation}
 Since $g(\infty) = b \in \mathbb{C}$ and $p$ is the only pole of $g$, we find that $g$ is bounded on $\mathbb{C} \setminus B(p, r/2)$. Let $C$ be a positive real number with 
 \begin{equation} \label{eq:implementations:bound:2}
 	\left| g(z) \right| \le C \qquad \text{for all } z \in \mathbb{C} \setminus B(p, r/2). 
 \end{equation}

  We can now specify the algorithm announced in the statement.  Let $\lambda \in \algcomplex$ and $\epsilon > 0$ rational be its inputs.  Our algorithm distinguishes two cases depending on $\lambda \in B(p, r_p)$. 
  
  \begin{enumerate}
  	\item $\lvert \lambda \rvert > M$. Then we can use the algorithm of Lemma~\ref{lem:efficient-covering:infinity} to compute  an element $a_k$ of a Mobius-program for $g$ starting at $a_0$ with $\lvert \lambda - a_k \rvert \le \epsilon$.
  	\item $\lvert \lambda \rvert \le M$. Let $x^* = g^{-1}(\lambda)$. We have  $x^* \not \in  B(p, r)$ because $\lvert g(x^*) \rvert = \lvert \lambda \rvert < 2M$, see \eqref{eq:implementations:bound:1}. Let 
  	\begin{equation} \label{eq:implementations:epsilon}
  		\epsilon' = \min\left\{\frac{\epsilon}{ L(1 + C^2)}, r/2 \right\}
  	\end{equation}
  	Our algorithm computes $x^*$ and distinguish two more cases depending on $x^*$. In each of the two cases the algorithm is going to compute  an element $\hat{x}$ of a Mobius-program for $g$ starting at $a_0$ with $\lvert x^* - \hat{x} \rvert \le \epsilon'$.
  	\begin{itemize}
  		\item $\lvert x^* \rvert > M$.  Our algorithm uses the algorithm of Lemma~\ref{lem:efficient-covering:infinity} with inputs $x^*$ and $\epsilon'$ to compute  an element $\hat{x}$ of a Mobius-program for $g$ starting at $a_0$ with $\lvert x^* - \hat{x} \rvert \le \epsilon'$, and returns $a_k$.
  		\item $\lvert x^* \rvert \le M$.  Note that $2M/r_0  > M$. Our algorithm first uses the algorithm of Lemma~\ref{lem:efficient-covering:infinity} with inputs $\lambda = 2M/r_0$ and $\epsilon = 1$ to compute an element $\lambda_4$ of a Mobius-program for $g$ starting at $a_0$ with $\lvert  2M/r_0 - \lambda_4 \rvert \le 1$, so $\lvert \lambda_4 \rvert > M/r_0$.  This step takes constant time since all the quantities involved are constants stored in our algorithm. The idea for this part of the proof is borrowed from  \cite[Proposition 2.2, case III]{Bezb}. Here $\lambda_4$ plays the same role as the activity $\lambda_4$ implemented in \cite[Proposition 2.2, case III]{Bezb}, hence the choice of the name. We have $ \lvert x^* / \lambda_4 \rvert < r_0$, so we can use the  algorithm of Lemma~\ref{lem:efficient-covering:any-point} for $\gamma =0$ with inputs $x^* / \lambda_4$ and $\epsilon' / \lvert \lambda_4 \rvert$ to compute an element $\hat{x}$ of a Mobius-program for $g$ starting at $a_0$ with $\lvert x^* / \lambda_4 - \hat{y} \rvert \le \epsilon'/ \lvert \lambda_4 \rvert$.  We set $\hat{x} = \lambda_4 \hat{y}$ and note that $\lvert x^* - \hat{x} \rvert \le \epsilon'$.
  	\end{itemize}
   From the definition of $\epsilon'$ \eqref{eq:implementations:epsilon} and the triangle inequality we have 
   \begin{equation*}
   	\lvert p - \hat{x} \rvert \ge \lvert p - x^* \rvert - \lvert x^* - \hat{x} \rvert \ge r - \epsilon' \ge r/2.
   \end{equation*}
  	Hence, \eqref{eq:implementations:bound:2} yields
  	\begin{equation} \label{eq:implementations:bound:3}
  		\left| g \left( \hat{x} \right) \right| \le C.
  	\end{equation}  	
  	In view of the Lipschitz property of $f^N$ we have
  		\begin{equation*}
  		d \left(  \lambda, g \left( \hat{x} \right) \right) \le L d \left( g^{-1}(\lambda), \hat{x} \right) = L \frac{\left| x^* - \hat{x} \right|}{ \left( \left( 1 + \lvert x^* \rvert^2 \right) \left( 1 + \lvert \hat{x} \rvert^2 \right) \right)^{1/2} } \le L \left| x^* - \hat{x} \right| \le L \epsilon' \le  \frac{\epsilon}{ L(1 + C^2)}
  	\end{equation*}
    which in combination with \eqref{eq:implementations:bound:3} and $\lvert \lambda \rvert \le M \le C$ yields
  \begin{equation*}
  	\left| \lambda - g \left( \hat{x} \right) \right| = \left( \left( 1 + \lvert g \left( \hat{x} \right) \rvert^2   \right)  \left( 1 + \lvert \lambda \rvert^2   \right) \right)^{1/2} d \left(  \lambda, g \left( \hat{x} \right) \right) \le L(1 + C^2) d \left(  \lambda, g \left( \hat{x} \right) \right) \le \epsilon.
  \end{equation*}
  Our algorithm returns the representation of $g(\hat{x})$ as an element of a Mobius-program. \qedhere
  \end{enumerate}
\end{proof}

\begin{remark}
 The hypothesis $g(\infty) \in \mathbb{C}$ can be removed from Lemma~\ref{lem:implementations}. This would require us to study the case $g(\infty) = \infty$ in the proof of Lemma~\ref{lem:implementations}. Note that $g(z) = \infty$ if and only if $g$ is of the form $az + b$ for $a, b \in \mathbb{Q}$ with $a \ne 0$. This case is not relevant for this work, hence why we left it out of the statement. 
 Also, for convenience, \cite{Bezb} restricted attention to complex numbers whose real and imaginary parts are rational, but it suffices for them to be algebraic.
\end{remark}

\section{\texttt{Mathematica} code   for the proof of Lemma~\ref{lem:region-R}}
\label{sec:B}

The following Mathematica code shows that
\begin{equation*}
	\left| \frac{e^{\tau i}-e^{\theta i}}{2 + \delta(e^{\tau i}+e^{\theta i})} \right| \le 1,
\end{equation*}
which was promised in the proof of Lemma~\ref{lem:region-R}.
The output of the code is {\tt False}. The code uses
  the (rigorous) {\tt Resolve} function of Mathematica.

\begin{verbatim}
	cos1 = 2 t1/(1 + t1^2);
	sin1 = (1 - t1^2)/(1 + t1^2);
	cos2 = 2 t2/(1 + t2^2);
	sin2 = (1 - t2^2)/(1 + t2^2);
	Z = ComplexExpand[((cos1 + I sin1) - (cos2 + I sin2)) /
	(2 +  delta (cos1 + cos2 + I (sin1 + sin2)))];
	X = Simplify[(Z + ComplexExpand[Conjugate[Z]])/2];
	Y = Simplify[(Z - ComplexExpand[Conjugate[Z]])/(2 I)];
	Resolve[Exists[{t1, t2, delta}, 
	X^2 + Y^2 > 1 && -1 <= t1 <= 1 && -1 <= t2 <= 1 && 0 < delta < 1]]
\end{verbatim}

 \end{appendices}
\end{document}